\pdfoutput=1
\documentclass[12pt,a4paper]{article}
\usepackage{amsthm}
\usepackage{caption}
\usepackage{subcaption}
\usepackage{jheppub}
\usepackage[all]{xy}

\newtheorem*{conjecture}{Conjecture}
\newtheorem*{problem}{Problem}
\newtheorem{theorem}{Theorem}[section]
\newtheorem{lemma}[theorem]{Lemma}
\newtheorem{proposition}[theorem]{Proposition}

\def\ket#1{\langle #1 \rangle}
\newcommand{\lms}{\longmapsto}
\newcommand{\lra}{\longrightarrow}
\newcommand{\Q}{\mathbb{Q}}
\newcommand{\R}{\mathbb{R}}
\newcommand{\C}{\mathbb{C}}
\newcommand{\Z}{\mathbb{Z}}
\newcommand{\PP}{\mathbb{P}}

\DeclareMathOperator{\Alt}{Alt}
\DeclareMathOperator{\Conf}{Conf}
\DeclareMathOperator{\Gr}{Gr}
\DeclareMathOperator{\Li}{Li}
\DeclareMathOperator{\Id}{Id}
\DeclareMathOperator{\sgn}{sgn}
\DeclareMathOperator{\B}{B}
\DeclareMathOperator{\Tr}{Tr}

\title{Motivic Amplitudes and Cluster Coordinates}

\author[a]{J.~K.~Golden,}
\author[b]{A.~B.~Goncharov,}
\author[a,c]{M.~Spradlin,}
\author[d]{C.~Vergu,}
\author[a,c]{and A.~Volovich}

\affiliation[a]{Department of Physics, Brown University,\\
Box 1843,\\
Providence, RI 02912-1843,
USA}
\affiliation[b]{Department of Mathematics, Yale University,\\
PO Box 208283,\\
New Haven, CT 06520-8283,
USA}
\affiliation[c]{Theory Division, Physics Department, CERN,\\
1211 Geneva 23, Switzerland}
\affiliation[d]{ETH Z\"urich, Institut f\"ur Theoretische Physik,\\
Wolfgang Pauli Strasse 27,\\
8093 Z\"urich, Switzerland}

\preprint{Brown-HET-1643}

\abstract{In this paper we study \emph{motivic amplitudes}---objects which contain all of the essential mathematical content of scattering amplitudes in planar SYM theory in a completely canonical way, free from the ambiguities inherent in any attempt to choose particular functional representatives.  We find that the cluster structure on the kinematic configuration space $\Conf_n(\PP^3)$ underlies the structure of motivic amplitudes.  Specifically, we compute explicitly the coproduct of the two-loop seven-particle MHV motivic amplitude $\mathcal{A}_{7,2}^{\mathcal{M}}$ and find that like the previously known six-particle amplitude, it depends only on certain preferred coordinates known in the mathematics literature as \emph{cluster $\mathcal{X}$-coordinates} on $\Conf_n(\mathbb{P}^3)$.
We also find intriguing relations between motivic amplitudes and the geometry of generalized associahedrons, to which cluster coordinates have a natural combinatoric connection.
For example, the obstruction to $\mathcal{A}_{7,2}^{\mathcal{M}}$ being expressible in terms of classical polylogarithms is most naturally represented by certain quadrilateral faces of the appropriate associahedron.
We also find and prove the first known functional equation for the trilogarithm
in which all 40 arguments are cluster $\mathcal{X}$-coordinates of a single algebra.
In this respect it is similar to Abel's 5-term dilogarithm identity.
}

\begin{document}
\maketitle

\section{Introduction}
\label{sec:introduction}

In the past several years a great amount of attention has been focused
on the problem of understanding the hidden mathematical structure
of scattering amplitudes
(for reviews see~\cite{Mangano:1990by,Dixon:1996wi,Cachazo:2005ga,Bern:2007dw,GreyBook,Feng:2011np}),
particularly (but certainly not exclusively)
in supersymmetric theories such as ${\mathcal{N}} = 4$ Yang-Mills
(SYM) theory~\cite{Brink:1976bc,Gliozzi:1976qd}.
As amplitudeologists, our mathematical interest in
planar SYM theory stems from imagining it as a vast and mysterious encyclopedia,
recovered from some long-lost desert cave, filled
with functions having remarkable properties and interrelationships.
This encyclopedia has many volumes, but
beyond the most introductory sections,
we can only make out bits and scraps of text here and there.

It is hardly our ambition to greatly ameliorate this situation.
Rather, our goal in this work is to describe some
general mathematical properties of and techniques for
analyzing amplitudes---to provide a kind of archaeologist's toolkit.
In particular, one overarching aim of our work is to
point out that SYM theory is an ideal setting in which to study
\emph{motivic amplitudes}, as proposed a decade ago in~\cite{G02}
(see in particular sec.~\ref{sec:examples}).
Why motivic amplitudes? It remains an important outstanding problem in physics to determine
explicit effective constructions for general amplitudes. 
However the abundance of functional identities amongst generalized
polylogarithms apparently precludes the existence of any
particular preferred or canonical functional representation or `formula' 
for general multi-loop
amplitudes (the only exception is reviewed
in section~\ref{sec:GSVV}).
Our goal is rather to investigate, following~\cite{G02},  their motivic avatars---motivic amplitudes---which are
mathematically more sophisticated, but at the same time much more structured
and canonical objects. In particular they are elements of a Hopf algebra.
This Hopf algebra is the algebra of functions on the so-called motivic Galois group. 
The group structure of the latter is encoded in the coproduct of the Hopf algebra. 
So by studying the coproduct of motivic amplitudes---a structure totally invisible
if we remain on the level of functions---we uncover their hidden motivic Galois symmetries.
One cannot resist to think that these new symmetries will eventually play an essential role in physics.

A similar upgrading,
from multi-zeta values to motivic multi-zeta values,
has recently played a crucial role in unlocking the structure of
tree-level
superstring amplitudes in the $\alpha'$
expansion~\cite{Schlotterer:2012ny,Drummond:2013vz,Broedel:2013tta,Broedel:2013aza}.  In SYM theory we expect the motivic approach to be even more powerful
since the amplitudes we deal with are not merely numbers but highly
nontrivial functions
on the $3(n-5)$-dimensional kinematic configuration space
$\Conf_n(\mathbb{P}^3)$, the space of collections of $n$ points in the projective space
$\mathbb{P}^3$, considered modulo the action of the projective linear group PGL${}_4$.

The one-sentence slogan of our paper is that
\emph{we find that
the cluster structure of the space $\Conf_n(\mathbb{P}^3)$ underlies
the structure of  amplitudes in SYM theory}.
The technical aspects of our work which support this
conclusion can be divided into two parts,
which one can think of very roughly as \emph{kinematics} and \emph{dynamics}.

We can phrase the `kinematic' question we are interested in roughly
as: \emph{which variables do motivic amplitudes in SYM theory depend on?}
Thanks to dual conformal symmetry~\cite{Drummond:2006rz,Bern:2006ew,Alday:2007hr,Drummond:2007aua,Drummond:2007cf,Alday:2007he,Drummond:2007au,Drummond:2008vq}
it is known that
appropriately defined
$n$-particle SYM
scattering amplitudes depend on  $3(n-5)$ algebraically
independent dual conformal cross-ratios.  However, all experience
to date indicates that the functional dependence of amplitudes
on these variables always takes very special forms.  For example, in
the case of the two-loop MHV amplitude for $n=6$ (reviewed
in section~\ref{sec:GSVV}), which can
be completely expressed in terms of the classical polylogarithm functions
$\Li_m$~\cite{Goncharov:2010jf},
only very particular algebraic functions of the three
independent cross-ratios appear
as arguments of the $\Li_m$'s.  It is natural to wonder why these
particular arguments appear, and not others, and to ask about
the arguments appearing in more general
amplitudes (including $n>6$, higher loops,
and non-MHV).

There is a more specific purely mathematical reason to concentrate on the study
of the motivic two-loop MHV amplitudes.
These amplitudes are polylogarithm-like functions of
weight (also known as transcendentality) four.
Any such function of weight one on a space $X$ is necessarily of the form $\log F(x)$, the logarithm of a rational function on $X$.
Next, any weight two function can be expressed 
as a sum of 
$\Li_2$'s  and products of two logarithms of rational functions. Similarly, any weight three function 
is a linear combination of $\Li_3$'s and products of lower weight polylogarithms of rational functions. 
So the question `which variables these functions depend on' is
well-defined,  up to the functional equations satisfied by $\Li_2$ and $\Li_3$---a
beautiful subject on its own. However this is no longer true for functions of weight four~\cite{G91b}.
There is an invariant associated to any weight four function, with values
in an Abelian group $\Lambda^2\B_2$,  reviewed in sec.~\ref{sec:math},
which is the obstruction
for the function to be expressible as a sum of products of classical polylogarithms $\Li_m$.   
This makes the above question  ill-defined.
However, upgrading weight four functions  to their motivic avatars one sees that
their coproducts are expressible via classical motivic polylogarithms of weights $\leq 3$,
and so the question makes sense again. Finally,  the coproduct 
preserves all information about motivic amplitudes but an additive constant. 
So to understand two-loop amplitudes we want first to find explicit formulas 
for the coproduct of their motivic avatars. 
On the other hand, the two-loop amplitudes are
 natural weight four functions, and so one can hope that their detailed analysis
might shed a new light on the fundamental unsolved mathematical
problems which we face starting at weight four.
Therefore, although the level of precision one can reach studying 
the two-loop motivic MHV amplitudes is unsustainable for higher loops, 
one can hope to discover general features by looking at the simplest case. 

In this paper we propose that the variables which appear in the study of the  
MHV amplitudes in SYM theory belong to a class  known in the
mathematics literature as
\emph{cluster $\mathcal{X}$-coordinates}~\cite{FG03b}
on the configuration space $\Conf_n(\mathbb{P}^3)$. 

Cluster $\mathcal{X}$-coordinates in general 
describe Poisson spaces which are in duality with cluster  algebras, originally
discovered in~\cite{1021.16017,1054.17024}. 
In particular, the space $\Conf_n(\mathbb{P}^3)$ is equipped with a natural Poisson structure, 
invariant under cyclic shift of the points. This Poisson structure 
looks especially simple in the cluster $\mathcal{X}$-coordinates:  
the logarithms of the latter  provide collections  
of canonical Darboux coordinates. 
It seems remarkable that the arguments of the amplitudes 
have such special Poisson properties, although at the moment we 
do not know how to fully exploit this connection. 

An immediate consequence of the cluster structure of the space $\Conf_n(\mathbb{P}^3)$ 
is that its real part $\Conf_n(\mathbb{RP}^3)$ contains the domain 
$\Conf^+_n(\mathbb{RP}^3)$ of positive 
configurations of $n$ points in $\mathbb{RP}^3$. This positive domain is evidently 
a part of the Euclidean domain in $\Conf_n(\mathbb{CP}^3)$, the domain
where amplitudes are singularity-free.

The configuration space $\Conf_n(\mathbb{P}^3)$ can be  
realized as a quotient of the
Grassmannian $\Gr(4,n)$ by the action of the group $(\C^*)^{n-1}$.
This Grassmannian,
describing the external kinematic data of an amplitude,
may look unrelated to those which star
in~\cite{ArkaniHamed:2009dn,ArkaniHamed:2009vw,ArkaniHamed:2009sx,ArkaniHamed:2009dg,ArkaniHamed:2010kv,ArkaniHamed:2012nw} and involve
also `internal' data related to loop integration variables.
However the cluster structure and in particular the positivity play a key role in the Grassmannian 
approach to amplitudes~\cite{ArkaniHamed:2012nw}, and we have no doubt that a tight connection between these objects
will soon emerge. 

Once one accepts the important role played by cluster coordinates
as the kinematic variables which, in particular,  the coproduct of the two-loop 
motivic MHV amplitudes are `allowed'
to depend on, it is natural to ask the `dynamic' question:
\emph{what exactly is the dependence on these coordinates}?  For example, what
explains the precise linear combination of $\Li_4$ functions appearing
in the two-loop MHV amplitude for $n=6$?
There is a vast and rich mathematical literature on cluster algebras,
which are naturally connected~\cite{1057.52003}
to beautiful combinatorial structures
known as cluster complexes and, more specifically, generalized associahedrons (or generalized Stasheff polytopes)~\cite{0114.39402}. 
We defer most of the dynamic question to subsequent work but report here
the first example of a connection between these mathematical
structures and motivic amplitudes:  we find that the `distance'
between a two-loop amplitude and the classical $\Li_4$
functions, expressed in the $\Lambda^2 \B_2$-obstruction, is
naturally formulated in terms of certain two-dimensional quadrilateral faces of
the associahedron for $\Conf_n(\mathbb{P}^3)$. Equivalently, the pairs of functions 
entering the Stasheff polytope $\Lambda^2 \B_2$-obstruction for the two-loop MHV amplitudes Poisson commute. 

The outline of this paper is as follows.  In section~\ref{sec:kinematics}
we briefly review various notations for configurations
of points in $\mathbb{P}^{k-1}$ and the appearance of the $3(n-5)$-dimensional
space $\Conf_n(\mathbb{P}^3)$ as the space on which
$n$-particle scattering amplitudes in SYM theory are defined. 
We also review the relationship with the Grassmannian $\Gr(4,n)$. 
In section~\ref{sec:GSVV} we call attention to
the very special arguments appearing inside the $\Li_4$
functions in the two-loop MHV amplitude for $n=6$.
Section~\ref{sec:math} reviews the mathematics necessary for the 
calculus of motivic amplitudes.
We present our result for the coproduct of the
two-loop $n=7$ MHV motivic amplitude in section~\ref{sec:mc} (results for
all higher $n$
will be given in a subsequent publication).
In section~\ref{sec:intr-clust-algebr} we turn to
cluster algebras related to $\Gr(4,n)$, the construction of cluster coordinates,
and the Stasheff polytope and cluster $\mathcal{X}$-coordinates
for $\Conf_n(\mathbb{P}^3)$.  Finally section~\ref{sec:examples} exhibits
these concepts for $n=6,7$ in detail and contains some
analysis of the structure of the two-loop $n=7$ MHV motivic amplitude and
its relation to the Stasheff polytope. While the $n=6$ case is well-known
in the mathematical literature, the geometry of the
cluster $\mathcal{X}$-coordinates in the $n=7$ case is more intricate. 
In Appendix~\ref{sec:parity-conj-twist} we discuss parity conjugation, and show how to calculate its 
action on the cluster $\mathcal{X}$-coordinates.  In Appendix~\ref{sec:trilog-ident} we 
discuss and prove the 40-term functional equation for the trilogarithm, which plays a role in sec.~\ref{sec:mc}.

\section{The Kinematic Configuration Space \texorpdfstring{$\Conf_n(\mathbb{P}^3)$}{Conf(n,P3)}}
\label{sec:kinematics}

Having argued that scattering amplitudes
are a collection of very interesting functions, we begin by
addressing a seemingly simple-minded question:  what variables do these
functions depend on?
Despite initial appearances this is far from a trivial question, and
somewhat surprisingly a completely satisfactory understanding
has only emerged rather recently.

\subsection{Momentum twistors}

The basic problem is essentially this:  a scattering amplitude
of $n$ massless particles depends on $n$ four-momenta
$p_i$ (which we can take to be complex),
but these are constrained variables. First of all each one
should be light-like, $p_i^2 = 0$ with
respect to the Minkowski metric for all $i$, and secondly energy-momentum
conservation requires that $p_1 + \cdots + p_n = 0$.
These constraints carve out a non-trivial subvariety
of $\mathbb{C}^{4n}$.  It is desirable to employ
a set of unconstrained variables which
parametrize precisely this subvariety. 
A solution to the problem is
provided by momentum twistors~\cite{Hodges:2009hk}, whose construction we now review.

In the planar limit of SYM theory we have an additional, and crucial,
piece of structure:  the $n$ particles come together with a specified
cyclic ordering.
This arises because each particle lives in the adjoint representation of a gauge
group and each amplitude is multiplied by some invariant constructed
from the gauge group generators of the participating particles.  In the
planar limit we take the gauge group to be U$(N)$ with $N \to \infty$,
in which case only amplitudes multiplying a single trace $\Tr[T^{a_1} \cdots T^{a_n}]$ of
gauge group generators are nonvanishing.

Armed with a specified cyclic ordering of the particles,
the conservation constraint is solved trivially by parameterizing each $p_i = x_{i-1} - x_i$ in
terms of $n$ dual coordinates $x_i$.  The $x_i$ specify the vertices,
in $\mathbb{C}^4$, of an $n$-sided polygon whose edges are the vectors $p_i$,
each of which is null.
A very special
feature of SYM theory in the planar limit is that all amplitudes
are invariant under
conformal transformations in this dual space-time~\cite{Drummond:2006rz,Bern:2006ew,Alday:2007hr,Drummond:2007aua,Drummond:2007cf,Alday:2007he,Drummond:2007au,Drummond:2008vq}

It is often useful, especially when one is interested in discussing
aspects of conformal symmetry, to compactify the space-time.  
For example, in
Euclidean signature, a single point at infinity has to be included
in order for conformal inversion to make sense.  
It is also convenient to complexify
space-time. 
Different real sections of this complexified space correspond to
different signatures of the space-time metric.
The complexified compactification $\widetilde{M}_4$
of four-dimensional space-time is
the Grassmannian manifold $\Gr(2,4)$ of two-dimensional vector
spaces in a four-dimensional complex vector space $V_4$; in other words, 
there is a one-to-one
correspondence between points in $\widetilde{M}_4$
and two-dimensional vector subspaces in $V_4$.
We can projectivize this picture to say that
the correspondence is between points in complexified
compactified space-time and lines in $\mathbb{P}^3$.

In the Grassmannian picture two
points are light-like separated if their corresponding $2$-planes 
intersect.
So after projectivization, 
a pair of light-like separated points in $\widetilde{M}_4$
corresponds to a pair of intersecting lines in $\mathbb{P}^3$. 
Conformal transformations in space-time correspond to PGL${}_4$
transformations on $\mathbb{P}^3$.

This $\mathbb{P}^{3}$ space is called twistor space in the physics literature.
The importance of this space was first noted in the work of Penrose~\cite{Penrose:1967wn, Penrose:1972ia} and more recently emphasized by Witten~\cite{Witten:2003nn}
in the context of Yang-Mills theory scattering amplitudes.
However the twistors we need here are not the ones associated to the space-time in which the scattering takes place, but rather
the ones associated to the dual space mentioned above, where the $x_i$ live and on which dual conformal symmetry acts.
These were called `momentum twistors' in ref.~\cite{Hodges:2009hk}, where they were first introduced.
Momentum twistor space has both a chiral supersymmetric
version
(see ref.~\cite{Mason:2009qx}) and a non-chiral supersymmetric version
(see refs.~\cite{Witten:1978xx,Beisert:2012gb,Beisert:2012xx}), but we will not make
use of these extensions in this paper.

To summarize: a scattering amplitude depends on a cyclically ordered collection
of points $x_i$ in the complexified momentum space $\mathbb{C}^4$, each of which corresponds to a projective line
in momentum twistor space.  Since each $x_i$ is null separated from its neighbors
$x_{i-1}$ and $x_{i+1}$, their corresponding lines in momentum twistor space
intersect.  We denote by $Z_i \in \mathbb{P}^3$ the intersection of
the lines corresponding to the points $x_{i-1}$ and $x_i$.  Conversely, an
ordered sequence of points $Z_1,\ldots,Z_n \in \mathbb{P}^3$ determines a
collection of $n$ lines which intersect pairwise and therefore correspond
to $n$ light-light separated points $x_i$ in the dual Minkowski space.

\subsection{Bracket notation}
The space we have just described---the collection of $n$ ordered points in $\mathbb{P}^3$ modulo the action of PGL${}_4$---defines
$\Conf_n(\mathbb{P}^3)$, read as `configurations of $n$ points in $\mathbb{P}^3$'.
Scattering amplitudes of $n$ particles in SYM theory are functions on this $3(n-5)$-dimensional kinematic domain.
This space can be essentially realized as the quotient $\Gr(4,n)/(\mathbb{C}^*)^{n-1}$ of the Grassmannian by considering
the space
of $4 \times n$ matrices (being the homogeneous coordinates of the $n$ points in $\mathbb{P}^3$) modulo
the left-action of PGL${}_4$ as well as independent rescaling of the $n$ columns.
In this presentation the natural dual conformal covariant objects are four-brackets of the form $\langle i j k l\rangle := \det(Z_i Z_j Z_k Z_l)$, which is just the $\mathbb{C}^{4}$ volume of the parallelepiped built on the vectors $(Z_{i}, Z_{j}, Z_{k}, Z_{l})$.

More precisely, emphasizing the structures involved, 
given a volume form $\omega_4$ in a four-dimensional vector space $V_4$, 
we can define the bracket $\langle v_1,v_2,v_3,v_4\rangle := \omega_4(v_1, v_2, v_3, v_4)$.

These four-brackets are key players in the rest of our story, so we list here a few of their important features.
The Grassmannian duality $\Gr(k,n) = \Gr(n-k,n)$ means that configurations of $n$ points in $\mathbb{P}^{k}$ are dual to configurations of $n$ points in $\mathbb{P}^{n-k-2}$. Explicitly, at six points the relationship between four-brackets in $\mathbb{P}^3$ and two-brackets in $\mathbb{P}^1$ given by
\begin{equation}
\ket{ijkl} = \frac{1}{2!} \epsilon_{ijklmn} \ket{mn}, \qquad
\ket{ij} = \frac{1}{4!} \epsilon_{ijklmn} \ket{klmn},
\end{equation}
while the
relationship at seven points between four-brackets in $\mathbb{P}^{3}$ and three-brackets in $\mathbb{P}^{2}$ is
clearly
\begin{equation}
\label{eq:sevenduality}
  \langle i j k l\rangle = \frac 1 {3!} \epsilon_{i j k l m n p} \langle m n p\rangle, \quad \langle i j k\rangle = \frac 1{4!} \epsilon_{i j k l m n p} \langle l m n p\rangle.
\end{equation}
We find it useful to exploit this duality for six and seven points when doing so leads to additional clarity. 
An invariant treatment of this duality is given below in sec.~\ref{sec:cg}.

More complicated PGL${}_4$ covariant objects can be formed naturally by using
projective geometry inside four-brackets.  Such objects will appear later in
sec.~\ref{sec:intr-clust-algebr}, so we review the standard notation for them here.
Following the $\cap$ notation introduced in ref.~\cite{ArkaniHamed:2010kv}
we define the four-brackets with an intersection to be

\begin{equation} \label{not}
  \langle a b (c d e) \cap (f g h)\rangle \equiv \langle a c d e\rangle \langle b f g h\rangle - \langle b c d e\rangle \langle a f g h\rangle.
\end{equation}  This composite four-bracket vanishes when the line $(ab)$ and the intersection of planes $(c d e) \cap (f g h)$ lie in a common hyperplane.

Here is a slightly different way to think about (\ref{not}). Consider a pair of vectors $v_1, v_2$ in 
a four-dimensional vector space $V_4$, and a pair of covectors $f_1, f_2 \in V_4^*$. 
Then we set 
\begin{equation} \label{not1}
  \langle v_1, v_2; f_1, f_2\rangle \equiv f_1(v_1)f_2(v_2) - f_1(v_2)f_2(v_1).
\end{equation} 
To get (\ref{not}) we just take the covectors $f_1(\ast):= \langle c,
d, e, \ast\rangle$ and 
$f_2(\ast):= \langle f, g, h, \ast\rangle$.

If we pick a vector $c$ in the intersection of the two hyperplanes, writing them as $(c a_2b_2)$ and $(ca_3b_3)$, then 
we can rewrite it in a slightly different notation, making more  symmetries manifest: 
\begin{equation} \label{not2}
  \langle a_1b_1 (ca_2b_2) \cap (ca_3b_3)\rangle =
-\langle c | a_1\times b_1, a_2\times b_2, a_3\times b_3\rangle,
\end{equation}
Precisely, consider the three-dimensional quotient $V_4/\langle c \rangle$ of 
the space $V_4$ along the subspace generated by the vector $c$. 
The volume form $\omega_4$ in $V_4$ induces a volume form $\omega_4(c, \ast, \ast, \ast)$
in  $V_4/\langle c \rangle$, and therefore in  the dual space $(V_4/\langle c \rangle)^*$. So we can define three-brackets $\langle \ast, \ast, \ast\rangle_c$ in 
$(V_4/\langle c \rangle)^*$.
The other six vectors in eq.~(\ref{not2})
project to  the quotient.
 Taking the cross-products $\times$ of consecutive pairs of these vectors, we get three covectors in 
$(V_4/\langle c \rangle)^*$. Their volume $-\langle a_1\times b_1, a_2\times b_2, a_3\times b_3\rangle_{c}$ 
equals the invariant (\ref{not1}). So we get formula (\ref{not2}).

Notice the expansions, where we 
use $\epsilon_{\alpha \beta \gamma} \epsilon^\alpha(\cdot, a, b) = a_\beta b_\gamma - a_\gamma b_\beta$:
\begin{align}
  \langle a_1 \times b_1, a_2 \times b_2, a_3 \times b_3\rangle &= \epsilon_{\alpha \beta \gamma} \epsilon^\alpha(\cdot, a_1, b_1) \epsilon^\beta(\cdot, a_2, b_2) \epsilon^\gamma(\cdot, a_3, b_3)\\ &=
  \langle a_1 a_2 b_2\rangle \langle b_1 a_3 b_3\rangle - \langle b_1 a_2 b_2\rangle \langle a_1 a_3 b_3\rangle \\
  &= -\langle a_2 a_1 b_1\rangle \langle b_2 a_3 b_3\rangle + \langle a_2 a_3 b_3\rangle \langle b_2 a_1 b_1\rangle \\
  &= \langle a_3 a_1 b_1\rangle \langle b_3 a_2 b_2\rangle - \langle a_3 a_2 b_2\rangle \langle b_3 a_1 b_1\rangle.
\end{align}

\subsection{Configurations and Grassmannians}
\label{sec:cg}

Let us formulate now the relationship between the Grassmannians $\Gr(k,n)$ and the configuration spaces 
more accurately. We start with the notion of configurations.

Let $V_k$ be a vector space of dimension $k$.
Denote by $\Conf_n(k)$ the space of
orbits of the group GL${}_k$ acting on the space of
$n$-tuples of vectors
in $V_k$. We call it the space of \emph{configurations of $n$ vectors in $V_k$}.
It is important to notice that the sets of
configurations of vectors in two different vector spaces of the same dimension are canonically isomorphic.
Denote by $\Conf_n(\PP^{k-1})$ the space of
PGL${}_k$-orbits  on the space of
$n$-tuples of points
in $\PP^{k-1}$, called \emph{configurations of $n$ points in $\PP^{k-1}$}.

We consider an $n$-dimensional `particle vector space' $\C^n$ with a given basis 
$(e_1, \ldots ,  e_n)$. Then a generic $k$-dimensional subspace $h$ in $\C^n$ determines 
a configuration of $n$ vectors $(f_1, \ldots , f_n)$ in the dual space $h^*$: these are 
the restrictions to $h$ of the coordinate linear functionals in $\C^n$ dual to the basis 
$(e_1, \ldots ,  e_n)$. This way we get a well-defined bijection only for generic $h$,
referred to mathematically as a birational isomorphism, 
\begin{equation} \label{GRassid}
\Gr(k, n) \stackrel{\sim}{\lra} \Conf_n(k). 
\end{equation}
The group $(\C^*)^n$ acts by rescaling in the directions of the coordinate axes in $\C^k$. 
This action transforms into rescaling of the vectors of the configuration space 
$\Conf_n(k)$. The diagonal subgroup $\C^*_{\rm diag}\subset (\C^*)^n$ acts 
trivially. So the quotient group 
$(\C^*)^{n-1} = (\C^*)^n/\C^*_{\rm diag}$ acts effectively. Passing to the quotients we get a birational isomorphism
\begin{equation}
\Gr(k, n)/(\C^*)^{n-1} \stackrel{\sim}{\lra} \Conf_n(\mathbb{P}^{k-1}).
\end{equation}

The dualities
\begin{equation}
\Conf_n(k)\stackrel{=}{\lra} \Conf_n(n-k), \qquad
\Conf_n(\mathbb{P}^k)\stackrel{=}{\lra} \Conf_n(\mathbb{P}^{n-k-2})
\end{equation}
are best understood via the identification with the Grassmannian (\ref{GRassid}), followed by the obvious 
isomorphism 
$\Gr(k, n) = \Gr(n-k, n)$, obtained by  taking the orthogonal planes.  

\subsection{The Euclidean region}

Scattering amplitudes in field theory have a complicated singularity structure, including poles and branch cut singularities.  However, there are regions in the kinematic space where such singularities are absent.  In particular, amplitudes are expected on physical grounds to be real-valued and singularity-free everywhere in the \emph{Euclidean region}, reviewed in this section.  It was discussed in ref.~\cite{Bern:2008ap} in connection with MHV amplitudes in SYM theory.

The Euclidean region is defined most directly in the dual space parametrized by the $x_i$.  We impose that the coordinates of the vectors $x_i$ are real and
\begin{equation}
  \label{eq:euclidean-region}
  (x_{i} - x_{i+1})^2 = 0, \qquad
  (x_{i} - x_{j})^2 < 0, \text{~otherwise},
\end{equation}
where the distance is computed with a metric of signature $(+,-,-,-)$ or $(+,+,-,-)$.  These constraints define the Euclidean region in terms of the $x_i$ coordinates.

When transformed to twistor coordinates the first constraint in eq.~\eqref{eq:euclidean-region} is always satisfied.  However, the constraint that the components of the vectors $x_i$ should be real is harder to impose.  We can think about twistors as being spinor representations of the complexified dual conformal group.  This complexified dual conformal group has several real sections: $SU(4)$ which corresponds to Euclidean signature, $SL(4, \mathbb{R})$ which corresponds to split signature $(+,+,-,-)$ and $SU(2,2)$ which corresponds to $(+,-,-,-)$ signature.

In fact, there are two kinds of spinor representations which we call twistors (denoted by $Z$) and conjugate twistors (denoted by $W$).  There is a $Z$ and a $W$ for every particle in a scattering process, which we denote by $Z_{i}$, $W_{i}$.  Under the dual conformal group the $Z$ and $W$ twistors and transform in the opposite way.  That is, if $M$ is a dual conformal transformation,
\begin{equation}
  \label{eq:twistor-transformation}
  W \to W' = W M^{-1}, \qquad Z \to Z' = M Z.
\end{equation}  This implies that there is an invariant product $W \cdot Z$.

Now we can study the reality conditions.  We will not discuss the Euclidean signature $(+,+,+,+)$ any further since it does not allow light-like separation.  For split signature, the twistors transform under $SL(4, \mathbb{R})$ so they can be taken to be real and independent.  For Lorentzian signature $(+,-,-,-)$ the symmetry group is $SU(2,2)$.  If $M \in SU(2,2)$, then $M^\dagger \mathcal{C} M = \mathcal{C}$, where $\mathcal{C}$ is a $(2,2)$ signature matrix which we will take to be real and symmetric.  Then $\mathcal{C} Z$ transforms in the same way as $W^\dagger$ so we can consistently impose a reality condition $W^\dagger = \mathcal{C} Z$.  This implies that $(W_{i} \cdot Z_{j})^{*} = W_{j} \cdot Z_{i}$.  The light-like conditions imply that $W_{i} = Z_{i-1} \wedge Z_{i} \wedge Z_{i+1}$, so the previous reality condition can be written in terms of the $Z$ twistors alone as $\langle i-1 i i+1 j\rangle^{*} = \langle j-1 j j+1 i\rangle$.

Finally, let us translate the second condition in eq.~(\ref{eq:euclidean-region}) into twistor language.  Space-time distances $(x_i - x_j)^2$ cannot be expressed in twistor variables without first making an arbitrary choice of `infinity twistor'.  However this choice cancels in conformal ratios, and for these the dictionary between space-time and momentum twistor space then implies that
\begin{equation}
  \frac {\langle i i+1 j j+1\rangle \langle k k+1 l l+1\rangle}{\langle i i+1 l l+1\rangle \langle j j+1 k k+1\rangle} > 0,
\end{equation} for all $i,j,k,l$ for which none of the four-brackets vanishes.  This condition is certainly guaranteed if $\langle i i+1 j j+1\rangle > 0$ for all nonvanishing four-brackets of this type.

We therefore define the Euclidean region in momentum twistor space by the condition that $\langle i i+1 j j+1\rangle > 0$.  It has two sub-regions
\begin{align}
  \text{$(2,2)$ signature:} &\qquad \langle i j k l\rangle \in   \mathbb{R},\\
  \text{$(3,1)$ signature:} &\qquad \langle i-1 i i+1 j\rangle^{*} = \langle j-1 j j+1 i\rangle.
\end{align}  Note that the $(2,2)$ signature region contains the positive Grassmannian which is well-studied mathematically.  In contrast, the $(3,1)$ region does not seem to have been studied in the mathematical literature.

\section{Review of the Two-Loop \texorpdfstring{$n=6$}{n=6} MHV Amplitude}
\label{sec:GSVV}

In the previous section we reviewed that $n$-particle
scattering amplitudes in SYM theory are functions on
the $3(n-5)$-dimensional space $\Conf_n(\mathbb{P}^3)$. It is further
believed~\cite{ArkaniHamed:2012nw} that
any MHV or next-to-MHV (NMHV) amplitude, at any loop order
$L$ in perturbation theory, can be expressed in terms of functions of uniform
transcendentality
weight $2L$.  A goal of this paper is to make a sharper statement
about the mathematical structure of
these functions.  Specifically: that their structure is described by a certain
preferred collection of functions on $\Conf_n(\mathbb{P}^3)$ which are known in
the mathematics literature as cluster $\mathcal{X}$-coordinates.
In this section we provide a simple but illustrative example of this
phenomenon.

The simplest nontrivial multi-loop scattering amplitude
is the two-loop MHV amplitude for $n=6$ particles.  This was originally
computed numerically in~\cite{Bern:2008ap,Drummond:2008aq},
then
analytically in~\cite{DelDuca:2009au,DelDuca:2010zg}
in terms of generalized polylogarithm functions, and finally
in a vastly simplified
form in terms of only the classical $\Li_m$ functions
in~\cite{Goncharov:2010jf}.  We present it here very mildly reexpressed as
\begin{equation}
\label{eq:GSVV}
R^{(2)}_6 = \sum_{i=1}^3\left( L_i - \frac{1}{2} \Li_4(-v_i) \right)
- \frac{1}{8}
\left( \sum_{i=1}^3 \Li_2(-v_i)\right)^2
+ \frac{1}{24} J^4 + \frac{\pi^2}{12} J^2 + \frac{\pi^4}{72},
\end{equation}
in terms of the functions
\begin{equation}
\begin{aligned}
L_i &= \frac{1}{384} P_i^4 + \sum_{m=0}^3 \frac{(-1)^m}{(2m)!!}
P_i^m
(\ell_{4-m}(x_i^+) + \ell_{4-m}(x_i^-)),\\
P_i &= 2 \Li_1(-v_i) - \sum_{j=1}^3 \Li_1(-v_j),
\end{aligned}
\end{equation}
and
\begin{equation}
\begin{aligned}
J &= \sum_{i=1}^3 \ell_1(x_i^+) - \ell_1(x_i^-),\\
\ell_n(x) &= \frac{1}{2} (\Li_n(-x) - (-1)^n \Li_n(-1/x)).
\end{aligned}
\end{equation}
Our aim in reproducing this formula here is to highlight two
rather astonishing facts.
The first is that the argument of each $\Li_n$ function is the negative
of one of the simple cross-ratios
\begin{align}
\label{eq:nineratios}
v_1 &= \frac{\ket{35} \ket{26}}{\ket{23}\ket{56}},&
v_2 &= \frac{\ket{13} \ket{46}}{\ket{16}\ket{34}},&
v_3 &= \frac{\ket{15} \ket{24}}{\ket{45}\ket{12}},
\nonumber\\
x^+_1 &= \frac{\ket{14}\ket{23}}{\ket{12}\ket{34}},&
x^+_2 &= \frac{\ket{25}\ket{16}}{\ket{56}\ket{12}},&
x^+_3 &= \frac{\ket{36}\ket{45}}{\ket{34}\ket{56}},
\\
x^-_1 &= \frac{\ket{14}\ket{56}}{\ket{45}\ket{16}},&
x^-_2 &= \frac{\ket{25}\ket{34}}{\ket{23}\ket{45}},&
x^-_3 &= \frac{\ket{36}\ket{12}}{\ket{16}\ket{23}}
\nonumber
\end{align}
(or their inverses).
We caution the reader that the $x_i^\pm$ here are the \emph{negative}
of the $x_i^\pm$ used in~\cite{Goncharov:2010jf}, while the $v_i$
used here are related to the three $u_i$ cross-ratios most commonly
seen in the literature by $v_i = (1-u_i)/u_i$.
Of course, these 9 variables are not independent---the dimension
of $\Conf_6(\mathbb{P}^3)$ is only three---so one could choose
any three of them in terms of which to express all
of the others algebraically.
It is striking that the
argument of each $\Li_m$ function in~(\ref{eq:GSVV})
is expressible as one of these simple cross-ratios
rather than, as might have been the case, some arbitrary algebraic function
of cross-ratios.

The second striking fact about~(\ref{eq:GSVV})
is that out of the 45 distinct cross-ratios of the form
\begin{equation}
\label{eq:crossratio}
r(i,j,k,l) = \frac{\ket{i j} \ket{kl}}{\ket{j k}\ket{i l}}
\end{equation}
only the 9 shown in~(\ref{eq:nineratios}) actually appear.
Note that here, as throughout the paper, we shall never count both $x$
and $1/x$ separately.

The presentation of~(\ref{eq:GSVV}) we have given here also highlights
another theme which will pervade this paper: positivity.  The
cross-ratios defined in eq.~(\ref{eq:nineratios}) all have the
manifest property
that they are positive whenever each ordered bracket is positive,
i.e.\ whenever $\ket{ij} > 0~\forall~i < j$ (see
appendix~\ref{app:positive} for additional details on
positive configurations).
As this example and others to be discussed below suggest,
we expect all MHV amplitudes will have particularly rich
structure on the positive subset of
the domain
$\Conf_n(\mathbb{P}^3)$.
The formula~(\ref{eq:GSVV}) is expressed in terms of the natural polylogarithm
function on the domain of positive real-valued $x$:
\begin{equation}
\Li_n(-x) =
\int_{\Delta_x} \log(1 + t_1)\, d \log t_2 \wedge \cdots \wedge d \log t_n
:= L_n(x)
\end{equation}
where $\Delta_x = \{ (t_1,\ldots,t_n) : 0 < t_1 < t_2< \cdots < t_n < x\}$.
The proper continuation of eq.~(\ref{eq:GSVV}) to the part of the Euclidean
region outside the positive domain was discussed in~\cite{Goncharov:2010jf}.

In the rest of this paper we will work almost exclusively not with amplitudes
but with coproducts of motivic amplitudes, reviewed in the next section.
For such purposes it is sufficient to highlight in $R_6^{(2)}$
only the leading terms
\begin{equation}
R^{(2)}_6 = \sum_{i=1}^3 L_4(x^+_i) + L_4(x^-_i) - \frac{1}{2} L_4(v_i)
+ \cdots,
\label{eq:GSVVleading}
\end{equation}
where the dots stand for products of functions of lower weight,
which are killed by the coproduct $\delta$ reviewed in the next section.

In a certain sense this example is too simple, as this amplitude
is likely unique in SYM theory in being expressible in terms
of classical polylogarithm functions $\Li_m$ only.  We do not aim to
write explicit formulas for more general amplitudes as there is apparently no
particular preferred or canonical functional form,
so the question of \emph{what variables the function depends on} requires
a more precise definition involving the more sophisticated
mathematics to which we turn our attention in the next section.

\section{Polylogarithms and Motivic Lie Algebras}
\label{sec:math}

In this section we review some of the necessary mathematical preliminaries
on transcendental functions and explain ways of distilling
the essential motivic content of such functions.
The precise mathematical definitions of motivic avatars of poly\-log\-arithm-like functions is 
given in~\cite{G02}.
Taking for granted that such avatars exist, our goal is to provide the elements of motivic calculus necessary to describe their basic properties.

\subsection{The motivic avatars of (generalized) polylogarithms}

Let us start with the motivic background.
Given any field $F$, there is an as yet hypothetical mathematical object
called the \emph{motivic Tate Lie coalgebra} $\mathcal{L}_\bullet(F)$ of this field~\cite{B}. It
is graded by positive integers, the weights, i.e.\ one has
\begin{equation}
\mathcal{L}_\bullet(F)=\bigoplus_{n=1}^\infty\mathcal{L}_n(F).
\end{equation}
There is a cobracket $\delta$, which is a weight preserving linear map
\begin{equation}
\label{cocom}
\delta: \mathcal{L}_\bullet(F) \lra \Lambda^2\mathcal{L}_\bullet(F).
\end{equation}
It satisfies the property that the following composition is zero:
\begin{equation}
\label{comp}
\mathcal{L}_\bullet(F) \stackrel{\delta}{\lra} \Lambda^2\mathcal{L}_\bullet(F)\stackrel{\delta \wedge \operatorname{Id} -
\operatorname{Id}\wedge \delta }{\lra}  \Lambda^3\mathcal{L}_\bullet(F).
\end{equation}
The very existence of this object is known only when
$F$ is a number field~\cite{DG}.

Denote by $V^\ast$ the dual vector space to a vector space $V$.
If each of the weight components $\mathcal{L}_n(F)$ were a finite-dimensional
vector space\footnote{However this is rarely the case; see~\cite{G91b} for the treatment of duals in the infinite-dimensional situation.},
this would mean that the dual graded vector space, defined as
\begin{equation}
\operatorname{L}_\bullet(F):= \bigoplus_{n=1}^\infty\operatorname{L}_{-n}(F), \qquad \operatorname{L}_{-n}(F): = (\mathcal{L}_n(F))^*,
\end{equation}
is a graded Lie algebra, with the bracket dual to the map $\delta$.
Then the condition~(\ref{comp}) follows from the Jacobi identity.

Consider the universal enveloping algebra $\operatorname{U}_\bullet(F)$ of the Lie algebra $\operatorname{L}_\bullet(F)$.
It is graded by non-positive integers. By definition, $\operatorname{U}_0(F) =\Q$.
Its graded dual
\begin{equation}
\mathcal{A}_\bullet(F):= \bigoplus_{n=0}^\infty\mathcal{A}_n(F), \qquad \mathcal{A}_n(F):= (\operatorname{U}_{-n}(F))^\ast
\end{equation}
has the structure of a commutative graded
Hopf algebra with a coproduct $\Delta$. One has
\begin{equation}
\mathcal{L}_\bullet(F) = \mathcal{A}_\bullet(F)/(\mathcal{A}_{>0}(F)\cdot \mathcal{A}_{>0}(F)).
\end{equation}
So the elements of $\mathcal{L}_\bullet(F)$ are the elements of $\mathcal{A}_\bullet(F)$ considered modulo the
sums of products of
elements of positive weight.

Now let $X$ be a complex variety and denote by $\C(X)$ the field of rational functions on $X$.
To give a first idea why the Lie coalgebra $\mathcal{L}_\bullet(F)$ and the Hopf algebra $\mathcal{A}_\bullet(F)$
are relevant to the analytic theory of polylogarithms and their generalizations, let us start with a
vague statement:

\emph{Any weight $n$ polylogarithm-like function $\mathcal{F}$ on $X$ gives rise to
an element $\mathcal{F}^{\mathcal{M}}$ of $\mathcal{A}_\bullet(\C(X))$. Considered modulo products of such
functions, it provides
an element of $\mathcal{L}_\bullet(\C(X))$.}

Precisely, but using terminology which we are not going to explain
here, a `weight $n$ polylogarithm-like function $\mathcal{F}$ on $X$' is a period of a weight $n$ framed variation
of mixed $\Q$-Hodge structures on an open part of $X$, with the Hodge weights $h^{p,q}$ being zero unless $p=q$,
which is of `geometric origin'. We call such functions \emph{Hodge-Tate periods}.

It is conjectured that in passing from $\mathcal{F}$ to its motivic
avatar $\mathcal{F}^{\mathcal{M}}$ we do not `lose any information
about $\mathcal{F}$'. See~\cite{G02}, where the motivic avatars of multiple polylogarithms were defined,
for a detailed account on the subject.
The main point is this:

\emph{To know a Hodge-Tate period function $\mathcal{F}$ is the same thing as to know its motivic avatar
$\mathcal{F}^{\mathcal{M}}$. The vector space $\mathcal{A}_n(\C(X))$  is precisely the linear vector space
spanned by motivic avatars of the weight $n$ Hodge-Tate period functions on open parts of $X$.}

The benefit of replacing $\mathcal{F}$ by $\mathcal{F}^{\mathcal{M}}$ is that
the latter lie in a Hopf algebra. Since this Hopf algebra is graded by non-negative integers,
its elements can be studied by applying the coproduct to them,
which is expressible via similar objects of lower weight.
The fundamental fact is that the kernels of the coproduct maps
\begin{equation}
\Delta: \mathcal{A}_\bullet(\C(X))\lra \otimes^2\mathcal{A}_\bullet(\C(X)), \qquad
\delta: \mathcal{L}_\bullet(\C(X))\lra \Lambda^2\mathcal{L}_\bullet(\C(X))
\end{equation}
are given by constants. Therefore,
taking the coproduct does not discard any essential information about the function.

\subsection{Higher Bloch groups}

So the key question is to describe the Hopf algebra $\mathcal{A}_\bullet(\C(X))$, or, equivalently, the
Lie coalgebra $\mathcal{L}_\bullet(\C(X))$.
The structure of $\mathcal{L}_\bullet(F)$ for any field $F$ is essentially
predicted by the Freeness Conjecture~\cite{G91b}. We start from its low weight consequences.

\vspace{3mm}
\noindent
Weight 1. First of all, one has
\begin{equation}
\label{wt1}
\mathcal{L}_1(F) = F^*\otimes_\Z \Q.
\end{equation}

\vspace{3mm}
\noindent
Weight 2. Let us recall the definition of the Bloch group $B_2(F)$~\cite{Bl,Su}.
Let $\Q[F]$ be the $\Q$-vector space
with basis elements $\{x\}$ for $x\in F$. Recall the cross-ratio
\begin{equation}
\label{eq:crossratio6}
r(x_1, x_2, x_3, x_4) = \frac{(x_1-x_2)(x_3-x_4)}{(x_2-x_3)(x_1-x_4)}.
\end{equation}
Notice the unusual normalization of the cross-ratio: $ r(\infty, -1,
0, x) = x$.

Given any $5$ points
$x_1, \dotsc, x_5$ on the projective line $\mathbb{P}^1(F)$ over $F$, set
\begin{equation} \label{5term}
\sum_{i=1}^5\{r(x_i, x_{i+1}, x_{i+2}, x_{i+3})\} \in \Q[F].
\end{equation}
Here the indices are considered modulo $5$.

Let $R_2(F)$ be the subspace generated
by $\{0\}$ and the \emph{five-term relations}~(\ref{5term}). Over the complex numbers, (\ref{5term}) provides
Abel's famous pentagon relation for the dilogarithm.
Precisely, consider the Bloch-Wigner single valued version of the dilogarithm,  altered by the $z\lms -z$ argument change:
\begin{equation}
\mathcal{L}_2(z):= \Im(\Li_2(-z) + \arg(1+z) \log|z|), \quad z\in \C.
\end{equation}
Then Abel's pentagon relation for the Bloch-Wigner dilogarithm is
\begin{equation}
\label{eq:fiveterm}
\sum_{i=1}^5\mathcal{L}_2(r(x_i, x_{i+1}, x_{i+2}, x_{i+3}))=0.
\end{equation}
The Bloch-Wigner function also satisfies the reality condition
$\mathcal{L}_2(z)+ \mathcal{L}_2(\overline z)=0$. Any functional
equation for the  Bloch-Wigner function can be deduced from this and
Abel's equation.  Since the complex-valued dilogarithm certainly does not
satisfy in general any reality condition, one refers to  Abel's
pentagon equation  as the \emph{generic functional equation for the
dilogarithm}.

Now the Bloch group is a $\Q$-vector space given by the quotient
\begin{equation}
\label{wt2}
\B_2(F):= \frac{\Q[F]}{R_2(F)}.
\end{equation}
In general we denote elements of $\B_n(F)$ by $\{x\}_n$, with $x \in F$.
It can be deduced from Beilinson's conjectures~\cite{B} and  Suslin's theorem~\cite{Su},
that one should have
\begin{equation}
\mathcal{L}_2(F) = \B_2(F).
\end{equation}
Set  $F^*_\Q:= F^*\otimes \Q$.
The weight $2$ part of the cocommutator map~(\ref{cocom}) is a map
\begin{equation}
\delta: \mathcal{L}_2(F) \to \Lambda^2 \mathcal{L}_1(F).
\end{equation}
The claim is that using the isomorphisms~(\ref{wt1}) and (\ref{wt2}), it becomes a map
\begin{equation}
\delta: \B_2(F) \lra \Lambda^2 F^*_\Q, \quad \{x\}_2 \lms (1+x) \wedge x, \quad \{0\}_2, \; \{-1\}_2 \lms 0.
\end{equation}
This can also be deduced from Suslin's theorem.
A non-trivial but not difficult fact to check is that the map $\{x\} \lms (1+x) \wedge x$ kills the five-term relations~(\ref{5term}), and thus descends to a map of the space $\B_2(F)$.

Notice that, unlike in the more traditional way to present Abel's
pentagon identity, all terms in this formula appear with a plus
sign. Moreover, the arguments  of the pentagon equation are nothing
else but the cluster $\mathcal{X}$-coordinates on the configuration
space $\Conf_5(\PP^1)$ (see sec.~\ref{sec:intr-clust-algebr}).
In particular, this explains the origin of the
non-standard normalization of the classical cross-ratio used in the
definition (\ref{eq:crossratio6}).

\vspace{3mm}
\noindent
Weight 3. Let us describe now the space $\mathcal{L}_3(F)$,  following~\cite{G91b}.  Consider the \emph{triple ratio of $6$ points $(z_1, \dotsc, z_6)$ in  $\mathbb{P}^2$}, given by the formula
\begin{equation}
\label{9.7.12.10}
r_3(z_1, \dotsc, z_6):= -\frac{\langle 124\rangle\langle 235\rangle\langle 316\rangle}{\langle 125\rangle
\langle 236\rangle\langle 314\rangle}.
\end{equation}
Here we pick  vectors $(l_1, \dotsc, l_6)$ in a three-dimensional
vector space $V_3$
projecting onto the
points $(z_1, \dotsc, z_6)$, and set $\langle i j k\rangle := \omega_3(l_i, l_j, l_k)$, where
$\omega_3$ is a volume form in $V_3$.

It was proved in~\cite{G91b} that
the triple ratio $r_3$ plays  a similar role for the trilogarithm as the classical cross-ratio in eq.~(\ref{eq:crossratio6}) does for the dilogarithm.
Precisely, consider the following single-valued version of the trilogarithm, which is the function from~\cite{Z} with argument modified by the change $z \lms -z$:
\begin{equation}
\mathcal{L}_3(z):= \Re\Bigl(\Li_3(-z) - \Li_2(-z) \log|z| - \frac{1}{3}\log^2|z|\log(1+z)\Bigr), \quad z\in \C.
\end{equation}
The functional equations for the trilogarithm are provided by
configurations of $7$ points $(z_1, \dotsc, z_7)$ in $\C\mathbb{P}^2$. Specifically,
\begin{equation}
\label{genericfet} \sum_{i=1}^7(-1)^i\Bigl(\Alt_{6}\mathcal{L}_3(r_3(z_1, \dotsc, \widehat
z_i, \dotsc, z_7))\Bigr)=0.
\end{equation}
Here $\Alt_{6}$ stands for the skew-symmetrization of the six points $z_1, \dotsc, \widehat
z_i, \dotsc, z_7$. 

The function 
$\mathcal{L}_3(z)$ satisfies the reality equation $\mathcal{L}_3(z) = \mathcal{L}_3(\overline z)$. 
Just like in the case of the Bloch-Wigner function, any functional equation for the 
function $\mathcal{L}_3(z)$ can be deduced from the reality equation and the 
equation (\ref{genericfet}). So the latter is referred to as the 
generic functional equation for the trilogarithm.

For an arbitrary field $F$, given any $7$ points on the projective plane over $F$,
consider an element
\begin{equation}
\label{9.7.12.11}
\sum_{i=1}^7 (-1)^i \Bigl(\Alt_{6} \{ r_3(z_1, \dotsc, \widehat
z_i, \dotsc, z_7)\}\Bigr) \in \Q[F].
\end{equation}
Let $R_3(F)$ be the subspace
generated by the elements~(\ref{9.7.12.11}),  where
$(z_1, \dotsc, z_7)$ are  points in  the projective plane over $F$, and $\{0\}$. Set
\begin{equation}
\label{wt3}
\B_3(F):= \frac{\Q[F]}{R_3(F)}.
\end{equation}
One deduces from the work~\cite{G91a,G91b}
on the proof of Zagier's conjecture~\cite{Z} on special values of
Dedekind $\zeta$-functions at $s=3$ that one should have
\begin{equation}
\mathcal{L}_3(F) = \B_3(F),
\end{equation}
However
the nature of the triple ratio~(\ref{9.7.12.10}) was a mystery. It was realized much later that the triple ratio is a  cluster $\mathcal{X}$-coordinate on the configuration space $\Conf_3(\PP^2)$, and moreover it is one  which cannot be reduced to cross-ratios of the type in eq.~(\ref{eq:crossratio6}). We will return to this later on.

The weight $3$ part of the cocommutator map~(\ref{cocom}) is a map\footnote{Here and in all that follows we use $V \otimes W$ to denote the summand in $\Lambda^2 (V \oplus W)$ given by vectors of the form $v \otimes w - w \otimes v$, since the map $v \otimes w \in V \otimes W \lms v \otimes w - w \otimes v \in \Lambda^2 (V \oplus W)$ is injective.}
\begin{equation}
\mathcal{L}_3(F) \to \mathcal{L}_2(F)\otimes \mathcal{L}_1(F).
\end{equation}
It follows from~\cite{G91b} that, using the
isomorphisms~(\ref{wt1}), (\ref{wt2}) and (\ref{wt3}), it becomes a
map
\begin{equation}
\delta: \B_3(F) \lra \B_2(F) \otimes F^*_\Q, \qquad \{x\}_3 \lms \{x\}_2 \otimes x.
\end{equation}
A quite non-trivial fact to check here is that the map $\{x\} \lms \{x\}_2  \otimes x$ kills the relations~(\ref{9.7.12.11}) and thus descends to a map defined on the space $\B_3(F)$, see~\cite{G95}.

\vspace{3mm}
The higher analogs of the $\B$-groups were defined in~\cite{G91a, G91b}. The group
$\mathcal{B}_n(F)$ is the quotient of the $\Q$-vector space $\Q[F]$ by the subspace of functional equations for
the classical $n$-logarithm. Although the functional equations are not known explicitly in general,
the subgroup they generate is defined for all $n$ inductively. Namely, consider a map
\begin{equation}
\delta'_n: \Q[F] \lra
\begin{cases}
\mathcal{B}_{n-1}(F)\otimes F^*_\Q, &n \ge 3\\
\Lambda^2 F^*_\Q, &n = 2,
\end{cases}
\quad\{x\}\lms
\begin{cases}
\{x\}_{n-1}\otimes x, &n \ge 3,\\
(1-x) \wedge x, &n = 2.
\end{cases}
\end{equation}
Now replace $F$ by the field $F(t)$ of rational functions in one variable, and
take an element $\sum_i a_i\{f_i(t)\} \in \Q[F(t)]$ killed by the map $\delta'_n$. We define a subspace
$\mathcal{R}_n(F) \subset \Q[F]$ as the subspace generated by the elements
$\sum_i a_i(\{f_i(0)\} - \{f_i(1)\})$, where we added $\{\infty\}:=0$. Then we set
\begin{equation}
\mathcal{B}_n(F):= \frac{\Q[F]}{\mathcal{R}_n(F)}.
\end{equation} The change of notation from ${\B}_n(F)$
to $\mathcal{B}_n(F)$ emphasizes that we deal with the definition where the functional equations are not known explicitly.
The map $\delta'_n$ induces a map
\begin{equation}
\label{db}
\delta_n:\mathcal{B}_n(F) \lra
\begin{cases}
\mathcal{B}_{n-1}(F)\otimes F^*_\Q, &n \ge 3,\\
\Lambda^2 F^*_Q, &n=2,
\end{cases}
\quad \{x\}_n \lms
\begin{cases}
\{x\}_{n-1}\otimes x, &n \ge 3,\\
(1-x) \wedge x, &n = 2.
\end{cases}
\end{equation}
One has $\mathcal{B}_n(F) = {\B}_n(F)$ for $n=2,3$.

At this point one might ask whether we have $\mathcal{L}_n(F) = \mathcal{B}_n(F)$ for all $n$.
It was shown in~\cite{G91b} that this is not the case starting with $n=4$.
Since this is the case we deal with when studying two-loop amplitudes, let us discuss it in detail.

\vspace{3mm}
\noindent
Weight 4. It is conjectured in~\cite{G91b}\footnote{This is a conjecture about a conjectural object. The point is that the very existence of
the Lie coalgebra $\mathcal{L}_\bullet(F)$, although still conjectural, follows from some `standard'
conjectures in algebraic geometry. Although there is a lot of evidence for the conjecture on $\mathcal{L}_4(F)$,
it is not known how to reduce it to any `first principles' conjectures.} that the $\Q$-vector
space  $\mathcal{L}_{4}(F)$ is an extension
\begin{equation}
\label{9.7.12.12}
0 \lra \mathcal{B}_4 \lra \mathcal{L}_{4} \lra \Lambda^2 {\B}_2 \lra 0.
\end{equation}
Here we start skipping the field $F$ in the notation.

To understand the nature of this extension, let us look at the coproduct map
\begin{equation}
\label{9.7.12.2}
\delta: \mathcal{L}_{4} \lra \Lambda^2 \mathcal{L}_2 \bigoplus \mathcal{L}_3\otimes \mathcal{L}_1
= \Lambda^2 {\B}_2 \bigoplus {\B}_3\otimes F^*_\Q.
\end{equation}
The map
\begin{equation}
\label{BBmap} \delta_{2,2}: \mathcal{L}_{4} \lra \Lambda^2 {\B}_2
\end{equation} in~(\ref{9.7.12.12})
is just a part of the coproduct. It is known to be surjective.
The restriction of the coproduct to the subspace $\mathcal{B}_4$ in~(\ref{9.7.12.12})
is  described as the map~(\ref{db}) for $n=4$. 

One can reformulate this as follows. Given an element $l \in {\cal L}_4(F)$, 
one can ask whether it can be written as a sum of the (motivic) classical $4$-logarithms. 
If $\delta_{2,2}(l) \not =0$, the answer is no. Indeed, the coproduct 
of the classical motivic $4$-logarithm is given by 
$$
\delta: \{x\}_4 \lms \{x\}_3 \otimes x \in \B_3(F) \otimes F^*_\Q. 
$$
Therefore the $\B_2\wedge \B_2$-component is zero. 
The deeper part of the conjecture tells that 
if $\delta_{2,2}(l) =0$, the answer is yes. So the map (\ref{BBmap}) is precisely the obstruction 
for an element $l$ to be  a sum of the (motivic) classical $4$-logarithms.

\vspace{3mm}
In summary, one can express $\mathcal{L}_{4}$ in terms of the higher Bloch groups,
which reflect properties of the classical polylogarithms, which are function of a single variable.
The Freeness Conjecture tells that a similar description is expected for all $\mathcal{L}_{n}$.
Here is its essential part:

\begin{conjecture}
Let $\operatorname{Lieb}_\bullet(F)$ be the free graded Lie algebra generated by the $\Q$-vector spaces
$\mathcal{B}_{n}(F)^*$ dual to $\mathcal{B}_{n}(F)$, where $n\geq 2$ and the weight of $\mathcal{B}_{n}(F)^*$ is $-n$.
Let us denote by $\operatorname{Lieb}^*_\bullet(F)$ the graded dual Lie coalgebra, graded by $n=2,3,\ldots$
Then one has
\begin{equation}
\mathcal{L}_\bullet(F) = {\operatorname{Lieb}}^*_\bullet(F).
\end{equation}
\end{conjecture}

For example, when $n=1, 2,3$ we cannot present $n$ as a sum of two integers $\geq 2$, and thus
the Freeness Conjecture implies  $\mathcal{L}_{n} = \mathcal{B}_n$ for $n=1, 2,3$.
In contrast with this, $4=2+2$, and so, besides $\mathcal{B}_{4}$, we have an extra contribution to
$\mathcal{L}_{4}$ given by $\Lambda^2 {\B}_{2}$. It is
the dual to the commutators of the weight $-2$ elements in the free Lie algebra
$\operatorname{Lieb}_\bullet$.

\subsection{Symbols}
Recall the motivic Hopf algebra $\mathcal{A}_\bullet(F)$.
One has $\mathcal{A}_1(F) = F^*_\Q$.
Let
$\Delta_{n-1,1}: \mathcal{A}_n(F) \lra \mathcal{A}_{n-1}(F)\otimes \mathcal{A}_1(F)$
be the $(n-1, 1)$-component of the coproduct. Iterating it, we get a sequence of maps
\begin{align}
\mathcal{A}_n(F) &\lra \mathcal{A}_{n-1}(F)\otimes \mathcal{A}_1(F)\\
\nonumber
&\lra \mathcal{A}_{n-2}(F)\otimes \mathcal{A}_1(F)
\otimes \mathcal{A}_1(F)\\
\nonumber
&\lra \ldots\\
\nonumber
&\lra \otimes^n\mathcal{A}_1(F) =
\otimes^n F^*_\Q.
\end{align}
Its composition is a map, called the symbol map:
\begin{equation}
{\rm S}: \mathcal{A}_n(F) \lra \otimes^n F^*_\Q.
\end{equation}

\subsection{Motivic scattering amplitudes}

The $L$-loop $n$-particle MHV motivic amplitudes  in
SYM theory, considered modulo products,  are elements
$\mathcal{A}^{\mathcal{M}}_{n, L} \in \mathcal{L}_{2 L}$.
The two-loop motivic amplitudes are therefore elements
\begin{equation}
\label{9.7.12.1}
\mathcal{A}^{\mathcal{M}}_{n,2}(Z_1, \dotsc, Z_n) \in \mathcal{L}_{4}(F), \quad F = \Q(Z_1, \dotsc, Z_n)
\end{equation}
defined by a generic configuration of $n$ points $(Z_1, \dotsc, Z_n)$ in $\mathbb{P}^3$.
Therefore, as was explained above, according to the conjectural description~(\ref{9.7.12.12}) of $\mathcal{L}_{4}$,
they can be expressed via $\Li_4(z)$ if and only if the $\Lambda^2 {\B}_2$
obstruction vanishes. This is exactly what happened in~\cite{Goncharov:2010jf} where the  two-loop  $n=6$ MHV amplitude
was calculated as a sum of classical $4$-logarithms.
The problem set out for us here and subsequent work is therefore:

\begin{problem}
\label{problem1.2}
\emph{Calculate the motivic
$n$-particle two-loop MHV amplitudes for $n>6$.
More specifically, this amounts to computing the coproduct}
\begin{equation}
\delta_{\mathcal{M}}(\mathcal{A}^{\mathcal{M}}_{n,2}(Z_1, \dotsc, Z_n)) \in \Lambda^2 {\B}_2(F) \bigoplus \B_3(F)\otimes F^*_\Q, \qquad
F = \Q(Z_1, \dotsc, Z_n).
\label{eq:problem}
\end{equation}
\end{problem}

The coproduct determines the amplitude as a function up to a constant and products of similar functions of lower weight. Unlike the mysterious extension~(\ref{9.7.12.12}), which is non-split,
the coproduct is given in terms of the groups ${\B}_n$, $n=1,2,3$, and so its calculation is a precise problem.
Let us emphasize that due to the `one-variable' nature of the groups $\B_n$, to write an element in $\Lambda^2 {\B}_2 \bigoplus \B_3\otimes F^*_\Q$  we need a collection of functions  on the configuration space $\Conf_n(\PP^3)$---the arguments of the dilogarithms and
trilogarithms.   The only ambiguity of such a presentation  results
from the functional equations they satisfy.

As we show below, these functions for the $2$-loop $n$-particle MHV motivic amplitudes, where
$n=6,7$, are cluster $\mathcal{X}$-coordinates on the space $\Conf_n(\PP^3)$.  Moreover, although functional equations do come into the picture, the ones we see for the two-loop amplitudes are also of cluster nature, and discussed in Appendix~\ref{sec:trilog-ident}.

We would like to stress that without the motivic approach one cannot even formulate
Problem~\ref{problem1.2}---there is no way to define the coproduct just on the level of functions.
Moreover, due to non-split nature of the extension~(\ref{9.7.12.12}),
there is even no canonical way to write a general element of  $\mathcal{L}_{4}$.

\section{The Coproducts of Two-Loop MHV Motivic Amplitudes}
\label{sec:mc}

Using the symbol of the two-loop $n$-point MHV amplitude, as computed in~\cite{CaronHuot:2011ky}, one can calculate the coproduct $\delta(\mathcal{A}_{n,2}^{\mathcal{M}})$ of the corresponding motivic amplitude as defined in the previous section and summarized in eq.~(\ref{eq:problem}).
For $n=6$ the $\Lambda^2 \B_2$ component is trivial, as was noted already in~\cite{Goncharov:2010jf}, while the $\B_3 \otimes \mathbb{C}^*$ part of the coproduct may be read off immediately from~(\ref{eq:GSVVleading}):
\begin{equation}
\delta(\mathcal{A}_{6,2}^{\mathcal{M}})\rvert_{\B_3 \otimes \mathbb{C}^*}=
\sum_{i=1}^3 \{x_i^+\}_3 \otimes x_i^+ + \{x_i^-\}_3 \otimes x_i^-
- \frac{1}{2} \{v_i\}_3 \otimes v_i.
\end{equation}

We defer results for general $n$ to a subsequent publication and
present here explicit results only for $n=7$, since our main goal
at the moment is to call the reader's attention to the
same two non-trivial features that we emphasized in
section~\ref{sec:GSVV}.
The first feature is that the entry $z$ appearing inside each $\{ z \}_2$ or
$\{ z \}_3$ is a single cross-ratio (rather than,
as might have been the case, some arbitrary algebraic function of
cross-ratios); the second feature is that of the thousands
of such cross-ratios
one can form at $n=7$, only a small handful actually
appear in the motivic amplitudes.
The structure of the results presented here will be extensively discussed
in subsequent sections.

\subsection{The \texorpdfstring{$\Lambda^2 \B_2$}{B2 wedge B2} component for \texorpdfstring{$n=7$}{n=7}}

We find that the $\Lambda^2 \B_2$ component of $\delta(\mathcal{A}_{7,2}^{\mathcal{M}})$
can be expressed as
\begin{multline}
  \label{eq:seven-pt-b2wb2}
\sum_{\rm dihedral}\Bigg(
\Big\lbrace \frac {\langle 7 \times 1, 2 \times 3, 4 \times 5\rangle}{\langle 1 2 7\rangle \langle 3 4 5\rangle}\Big\rbrace_{2}
\wedge
\Big\lbrace \frac {\langle 2 \times 3, 4 \times 6, 7 \times 1 \rangle}{\langle 1 6 7\rangle \langle 2 3 4\rangle}\Big\rbrace_{2}\\
+ \frac{1}{2} \Big\lbrace \frac {\langle 3 1 2\rangle \langle 3 4 7\rangle}{\langle 3 7 1\rangle \langle 3 4 2\rangle}\Big\rbrace_{2}
\wedge
\Big\lbrace \frac {\langle 7 1 3\rangle \langle 7 4 6\rangle}{\langle 7 1 6\rangle \langle 7 3 4\rangle}\Big\rbrace_{2}\Bigg) + \mbox{parity conjugate},
\end{multline}
where the sum indicates that one should sum over the dihedral group
acting on the particle labels, resulting in a total of 42 distinct terms (the $1/2$ in front of the second term is a symmetry factor).

It is important to note that this expression is of course not unique.
There are two kinds of ambiguities:  first of all, there are dilogarithm
identities which hold at the level of $\B_2$ and are a consequence
of Abel's pentagon relation~(\ref{eq:fiveterm}).
One slightly non-trivial, but easily checked,
example is
\begin{multline}
\label{eq:b2identity}
0 =
   -\left\{\frac{\langle 127\rangle  \langle 156\rangle  \langle 345\rangle
    }{\langle 157\rangle  \langle 1\times 2,3\times 4,5\times 6\rangle
    }\right\}_2+\left\{\frac{\langle 127\rangle  \langle 256\rangle
\langle
    345\rangle }{\langle 257\rangle  \langle 1\times 2,3\times 4,5\times
    6\rangle }\right\}_2\\+\left\{-\frac{\langle 567\rangle  \langle 1\times
    2,3\times 4,5\times 6\rangle }{\langle 256\rangle  \langle 1\times
    7,3\times 4,5\times 6\rangle }\right\}_2-\left\{-\frac{\langle
127\rangle
    \langle 345\rangle  \langle 567\rangle }{\langle 157\rangle  \langle
    2\times 7,3\times 4,5\times 6\rangle }\right\}_2-\left\{\frac{\langle
    571\rangle  \langle 562\rangle }{\langle 512\rangle  \langle 567\rangle
    }\right\}_2.
\end{multline}
Such identities relate different, but equivalent, expressions for
the $\Lambda^2 \B_2$ component which in general have different numbers of terms.

Secondly there is ambiguity in
writing~(\ref{eq:seven-pt-b2wb2}) due to the trivial identities
\begin{equation}
\label{eq:squareidentity}
\{ x \}_2 \wedge \{ y \}_2 = - \{ 1/x \}_2 \wedge \{ y \}_2
= \{ 1/x \}_2 \wedge \{ 1/y \}_2 = - \{ x \}_2 \wedge \{ 1/y \}_2
\end{equation}
which preserve the number of terms.

Given these ambiguities one may wonder about the value in providing any
explicit
formula such as~(\ref{eq:seven-pt-b2wb2}).  Is there some invariant
way of presenting the $\Lambda^2 \B_2$ part of the coproduct of this amplitude,
without having to commit any particular representation?
To phrase this question
in language familiar to physicists: if we think about the
equations~(\ref{eq:b2identity})
and~(\ref{eq:squareidentity}) as generating some kind of gauge transformations,
then what is
the gauge-invariant content of the $\Lambda^2 \B_2$ component of this
coproduct?

We pose this question here merely as a teaser; to answer it fully requires
the mathematical machinery to be built up in
section~\ref{sec:intr-clust-algebr}.
Nevertheless we will also tease the reader here with the answer: the
42 terms in~(\ref{eq:seven-pt-b2wb2}) are naturally in correspondence
with certain quadrilateral faces of the $E_6$ Stasheff polytope, and
this is the shortest manifestly symmetric (dihedral $+$ parity) representative
with this property.

\subsection{The \texorpdfstring{$\B_3 \otimes \mathbb{C}^*$}{B3 x C*} component for \texorpdfstring{$n=7$}{n=7}}
\label{sec:mhv-1}

In order to save space, we will make use of the dihedral symmetry
and invariance under rescaling of the MHV amplitude to write only the two independent
$\B_3$ components necessary to express the full answer.
First, for the
$\langle 124\rangle$ component of $\C^*$ we find the $\B_3$ element
{\allowdisplaybreaks
\begin{multline}
\label{eq:mctwo}
\left\lbrace\frac{\langle 127\rangle \langle 256\rangle \langle
    345\rangle }{\langle 257\rangle \langle 1\times 2,3\times
    4,5\times 6\rangle }\right\rbrace_3 +
\left\lbrace\frac{\langle 125\rangle \langle 234\rangle \langle
    567\rangle }{\langle 257\rangle \langle 1\times 2,3\times
    4,5\times 6\rangle }\right\rbrace_3\\
-\left\lbrace\frac{\langle 127\rangle \langle 234\rangle \langle
    345\rangle \langle 567\rangle }{\langle 257\rangle \langle
    347\rangle \langle 1\times 2,3\times 4,5\times
    6\rangle}\right\rbrace_3 +
\left\lbrace\frac{\langle 124\rangle \langle 157\rangle}{\langle
    127\rangle \langle 145\rangle}\right\rbrace_3 +
\left\lbrace-\frac{\langle 12 7\rangle \langle 25 4\rangle\langle 51
    6\rangle}{\langle 12 4\rangle \langle 25 6\rangle \langle
    51 7\rangle}\right\rbrace_3\\
-\left\lbrace\frac{\langle 2 14\rangle \langle 2 57\rangle}{\langle 2
    17\rangle \langle 2 45\rangle}\right\rbrace_3
+\left\lbrace\frac{\langle 4 12\rangle \langle 4 35\rangle}{\langle 4
    15\rangle \langle 4 23\rangle}\right\rbrace_3
-\left\lbrace\frac{\langle 4 12\rangle \langle 4 37\rangle}{\langle 4
    17\rangle \langle 4 23\rangle}\right\rbrace_3
-\left\lbrace-\frac{\langle 14 7\rangle \langle 45 2\rangle \langle 51
    6\rangle}{\langle 14 2\rangle \langle 45 6\rangle \langle
    51 7\rangle}\right\rbrace_3\\
-\left\lbrace-\frac{\langle 24 7\rangle \langle 45 3\rangle \langle 52
    6\rangle}{\langle 24 3\rangle \langle 45 6\rangle \langle 52
    7\rangle}\right\rbrace_3 +
\left\lbrace\frac{\langle 4 12\rangle \langle 457\rangle}{\langle 4 17\rangle \langle 4 25\rangle}\right\rbrace_3 -
\left\lbrace\frac{\langle 5 12\rangle \langle 5 47\rangle}{\langle 5
    17\rangle \langle 5 24\rangle}\right\rbrace_3 +
\left\lbrace\frac{\langle 7 12\rangle \langle 7 45\rangle}{\langle 7
    15\rangle \langle 7 24\rangle}\right\rbrace_3\\
-\left\lbrace\frac{\langle 4 13\rangle \langle 4 57\rangle}{\langle 4
    17\rangle \langle 4 35\rangle}\right\rbrace_3
+\left\lbrace\frac{\langle 4 23\rangle \langle 4 57\rangle}{\langle 4
    27\rangle \langle 4 35\rangle}\right\rbrace_3
-\left\lbrace-\frac{\langle 24 1\rangle \langle 45 3\rangle \langle 52
    7\rangle}{\langle 24 3\rangle \langle 45 7\rangle \langle 52
    1\rangle }\right\rbrace_3 +
\left\lbrace-\frac{\langle 24 1\rangle \langle 47 3\rangle \langle 72
    5\rangle}{\langle 24 3\rangle \langle 47 5\rangle \langle 72
    1\rangle}\right\rbrace_3\\
+\left\lbrace\frac{\langle 5 12\rangle \langle 5 67\rangle}{\langle 5
    17\rangle \langle 5 26\rangle}\right\rbrace_3
-\left\lbrace-\frac{\langle 25 4\rangle \langle 57 6\rangle \langle 72
    1\rangle}{\langle 25 6\rangle \langle 57 1\rangle \langle
    72 4\rangle}\right\rbrace_3 -
\left\lbrace\frac{\langle 5 14\rangle \langle 5 67\rangle}{\langle 5
    17\rangle \langle 5 46\rangle}\right\rbrace_3 +
\left\lbrace-\frac{\langle 45 2\rangle \langle 57 6\rangle \langle 74
    1\rangle}{\langle 45 6\rangle \langle 57 1\rangle  \langle 74
    2\rangle}\right\rbrace_3\\
-\left\lbrace\frac{\langle 5 24\rangle \langle 5 67\rangle}{\langle 5
    27\rangle \langle 5 46\rangle}\right\rbrace_3
+\left\lbrace-\frac{\langle 45 3\rangle \langle 57 6\rangle \langle
    74 2\rangle}{\langle 45 6\rangle \langle 57 2\rangle \langle
    74 3\rangle }\right\rbrace_3.
\end{multline}}

Secondly, for $\langle 125 \rangle$ (this is symmetric under $1\leftrightarrow 2, 7 \leftrightarrow 3, 6 \leftrightarrow 4$)
we find
{\allowdisplaybreaks
\begin{multline}
\label{eq:mcthree}
-\left\lbrace\frac{\langle 157\rangle \langle 234\rangle}{\langle
    1\times 2,3\times 4,5\times 7\rangle }\right\rbrace_3-
\left\lbrace\frac{\langle 123\rangle \langle 457\rangle}{\langle
    1\times 2,3\times 4,5\times 7\rangle}\right\rbrace_3
-\left\lbrace\frac{\langle 127\rangle \langle 134\rangle \langle
    567\rangle}{\langle 167\rangle \langle 1\times 2,3\times 4,5\times
    7\rangle}\right\rbrace_3\\
+\left\lbrace\frac{\langle 127\rangle \langle 234\rangle \langle
    567\rangle}{\langle 267\rangle \langle 1\times 2,3\times 4,5\times
    7\rangle}\right\rbrace_3
+\left\lbrace\frac{\langle 123\rangle \langle 345\rangle \langle
    567\rangle}{\langle 356\rangle \langle 1\times 2,3\times 4,5\times
    7\rangle}\right\rbrace_3\\
-\left\lbrace\frac{\langle 124\rangle \langle 345\rangle \langle
    567\rangle}{\langle 456\rangle \langle 1\times 2,3\times 4,5\times
    7\rangle}\right\rbrace_3
-\left\lbrace\frac{\langle 123\rangle \langle 145\rangle}{\langle
    125\rangle \langle 134\rangle}\right\rbrace_3
-2 \left\lbrace\frac{\langle 125\rangle  \langle 167\rangle}{\langle 127\rangle \langle 156\rangle}\right\rbrace_3\\
+\left\lbrace\frac{\langle 4 12\rangle \langle 4 35\rangle}{\langle 4 15\rangle \langle 4 23\rangle}\right\rbrace_3
+\left\lbrace\frac{\langle 5 12\rangle \langle 5 34\rangle}{\langle 5 14\rangle \langle 5 23\rangle}\right\rbrace_3
-\left\lbrace-\frac{\langle 23 1\rangle \langle 35 4\rangle \langle
    52 6\rangle}{\langle 23 4\rangle \langle 35 6\rangle \langle
    52 1\rangle}\right\rbrace_3
+\left\lbrace-\frac{\langle 24 1\rangle \langle 45 3\rangle \langle
    52 6\rangle}{\langle 24 3\rangle \langle 45 6\rangle \langle
    52 1\rangle }\right\rbrace_3 \\
+\left\lbrace\frac{\langle 5 23\rangle \langle 5 46\rangle}{\langle 5 26\rangle \langle 5 34\rangle}\right\rbrace_3 +
(1\leftrightarrow 2, 7 \leftrightarrow 3, 6 \leftrightarrow 4).
\end{multline}}

The full $\B_3 \otimes \mathbb{C}^*$
part of the coproduct of the two-loop $n=7$ MHV motivic amplitude
is assembled in terms of these two building blocks as
\begin{equation}
\label{eq:7pt-MHV-B3C}
\sum_{\text{dihedral}} \left[
(\text{eq.~\ref{eq:mctwo}}) \otimes
\frac {\langle 124\rangle \langle 567\rangle}{\langle 127\rangle \langle 456\rangle} +
\frac{1}{2}
(\text{eq.~\ref{eq:mcthree}}) \otimes
\frac {\langle 125\rangle \langle 167\rangle \langle 234\rangle}{\langle 123\rangle \langle 127\rangle \langle 456\rangle}
\right].
\end{equation}

Let us comment briefly on the action of \emph{parity} (defined and discussed in detail in appendix~\ref{sec:parity-conj-twist}) on eq.~\eqref{eq:7pt-MHV-B3C}.  If we apply parity to the first term in eq.~\eqref{eq:7pt-MHV-B3C} and then the permutation $(1,7,6,5,4,3,2)$, the $\mathbb{C}^{*}$ term is unchanged. The corresponding $B_{3}$ terms are identical up to a $40$-term $\Li_{3}$ identity, of the type discussed in appendix~\ref{sec:trilog-ident}.  The parity invariance of the second term in eq.~\eqref{eq:7pt-MHV-B3C} is easier to prove: applying parity followed by the permutation $(4,5,6,7,1,2,3)$ leaves both the $\mathbb{C}^{*}$ term and the corresponding $\B_3$ part invariant.  Therefore in this case the parity invariance is manifest, without requiring any $\Li_{3}$ identities.
We note with interest that the aforementioned 40-term identity is the \emph{only} nontrivial $\Li_3$ identity which plays a role in elucidating the
structure of this amplitude.

\section{Cluster Coordinates and Cluster Algebras}
\label{sec:intr-clust-algebr}

Having presented concrete results of the calculation of the coproduct for the two-loop $n=7$ MHV amplitude, we now turn to a second main theme of the paper:  establishing the connection between the coproduct of motivic amplitudes and cluster algebras.
This section is aimed primarily at physicists since
most of this material
is a review of fairly well-known mathematical facts, with
a focus on the intended application
to $\Conf_n(\mathbb{P}^3)$, the space on which scattering amplitudes live.
We start with a short introduction to cluster algebras, which were discovered and first developed in a series of papers~\cite{1021.16017, 1054.17024} by Fomin and Zelevinsky.
The configuration space $\Conf_n(\mathbb{P}^{k-1})$ of $n$ points
in $\mathbb{P}^{k-1}$ has the structure of a \emph{cluster Poisson variety}~\cite{FG03b}---a
structure closely related to cluster algebras.
Our aim in this section is to guide the reader quickly to an understanding of what
cluster variables are, and how they may be systematically constructed
via a process called mutation.
Cluster coordinates come in two types, referred to as cluster $\mathcal{A}$-coordinates for Grassmannians and cluster $\mathcal{X}$-coordinates for the
configuration spaces $\Conf_n(\mathbb{P}^{k-1})$.

To guide the reader, let us specify the spaces  where different types of cluster coordinates live. 
Recall (see sec.~\ref{sec:kinematics}) that the Grassmannian $\Gr(k,n)$ is birationally isomorphic to the configuration space of vectors  
$\Conf_n(k)$. The latter projects onto the space of projective configurations $\Conf_n(\mathbb{P}^{k-1})$:
\begin{equation} \label{pi}
\Gr(k,n) \stackrel{\sim}{\lra} \Conf_n(k) \stackrel{\pi}{\lra} \Conf_n(\mathbb{P}^{k-1}).
\end{equation}
Cluster $\mathcal{A}$-coordinates live naturally on the Pl\"ucker cone $\widetilde \Gr(k,n)$ 
over the Grassmannian 
rather than on the Grassmannian itself.  
This cone can be identified  birationally with the configuration space $\widetilde \Conf_n(k)$ of $n$ vectors 
modulo the action of the group $SL_k$ rather than $GL_k$. 
Abusing terminology, one often refers to them as coordinates on the Grassmannian. 
Cluster $\mathcal{X}$-coordinates live naturally on the smaller space $\Conf_n(\mathbb{P}^{k-1})$---its
dimension is $n$ less then that of the $\widetilde \Conf_n(k)$.
They describe a collection of log-canonical Darboux coordinate systems 
for the natural cyclic invariant Poisson structure on the space $\Conf_n(\mathbb{P}^{k-1})$.  
Using the canonical projection 
$$
\widetilde \pi: \widetilde \Gr(k,n) \lra \Conf_n(\mathbb{P}^{k-1})
$$ one can lift the cluster $\mathcal{X}$-coordinates 
on $\Conf_n(\mathbb{P}^{k-1})$, and express them as monomials of the cluster 
variables in the Grassmannian cluster algebra. 
One cannot extend, however, the Poisson structure to the cluster algebra without breaking the cyclic invariance.  
Notice that the cluster algebra structure itself on $\widetilde \Conf_n(k)$ is 
(twisted) cyclic invariant. The cyclic invariance is a crucial feature of the planar scattering amplitudes.  
 So it is important that both Grassmannian cluster algebra  and the cluster 
Poisson structure on $\Conf_n(\mathbb{P}^{k-1})$ are cyclic invariant.

\subsection{Introduction and definitions}
\label{sec:defin-intr}

We can informally define cluster algebras as follows: they are commutative algebras constructed from distinguished generators (called \emph{cluster variables}) grouped into non-disjoint sets of constant cardinality (called \emph{clusters}), which are constructed iteratively from an initial cluster by an operation called \emph{mutation}.  The number of variables in a cluster is called the \emph{rank} of the cluster algebra.

A simple example is the ${\rm A}_{2}$ cluster algebra
defined by the following data:
\begin{itemize}
\item cluster variables: $a_{m}$ for $m \in \mathbb{Z}$, subject to
$a_{m-1} a_{m+1} = 1 + a_{m}$
\item rank: $2$
\item clusters: $\lbrace a_{m}, a_{m+1}\rbrace$ for $m \in \mathbb{Z}$
\item initial cluster: $\lbrace a_{1}, a_{2}\rbrace$
\item mutation: $\lbrace a_{m-1}, a_{m}\rbrace \to \lbrace a_{m}, a_{m+1}\rbrace$.
\end{itemize}
Using the exchange relation $a_{m-1} a_{m+1} = 1 + a_m$ one sees that
\begin{equation}
 a_{3} = \frac {1+a_{2}}{a_{1}},\quad
 a_{4} = \frac {1+a_{1}+a_{2}}{a_{1} a_{2}},\quad
 a_{5} = \frac {1+a_{1}}{a_{2}},\quad
 a_{6} = a_{1}, \quad a_{7} = a_{2}.
\end{equation}  Therefore, the sequence $a_{m}$ is periodic with period five and the number of cluster variables is finite.

When expressing the cluster variables $a_{m}$ in terms of $a_{1}$ and $a_{2}$, we encounter two interesting features.  First, the denominators of the cluster variables are always monomials.  In general the structure of an algebra is such that one might expect the cluster variables to be more general rational functions of the initial cluster variables, but in fact the denominator is always a monomial.  This is known as the `Laurent phenomenon' 
(see~\cite{MR1888840,1021.16017}).  A second feature is that, conjecturally, 
 the numerator of each cluster variable, expressed in terms of the original cluster variables, 
 is always a polynomial with positive integer coefficients.

Some cluster algebras may be defined in terms of another piece of mathematical machinery: quivers. A quiver is an oriented graph, and in the following we restrict to connected, finite quivers without loops (arrows with the same origin and target) and two-cycles (pairs of arrows going in opposite directions between two nodes).

Given a quiver and a choice of some node $k$ on that quiver we can define a new quiver obtained by \emph{mutating} at node $k$.  The new quiver is obtained by applying the following operations on the initial quiver:
\begin{itemize}
\item for each path $i \to k \to j$, add an arrow $i \to j$,
\item reverse all arrows on the edges incident with $k$,
\item and remove any two-cycles that may have formed.
\end{itemize}
The mutation at $k$ is an involution; when applied twice in succession at the same node we come back to the original quiver.

Quivers of the special type under consideration are in one-to-one correspondence with skew-symmetric matrices defined as
\begin{equation}
b_{i j} = (\# \text{arrows}\; i \to j) - (\# \text{arrows}\; j \to i).
\label{eq:bijdef}
\end{equation}
Since at most one of the terms above is nonvanishing, $b_{i j} = -b_{j i}$.  Under a mutation at node $k$ the matrix $b$ transforms to $b'$ given by
\begin{equation}
  \label{eq:b-mutation}
  b'_{i j} =
  \begin{cases}
    -b_{i j}, &\quad \text{if $k \in \lbrace i, j\rbrace$,}\\
    b_{i j}, &\quad \text{if $b_{i k} b_{k j} \leq 0$,}\\
    b_{i j} + b_{i k} b_{k j}, &\quad \text{if $b_{i k}, b_{k j} > 0$,}\\
    b_{i j} - b_{i k} b_{k j}, &\quad \text{if $b_{i k}, b_{k j} < 0$.}
  \end{cases}
\end{equation}

If we start with a quiver with $n$ nodes and associate to each node $i$ a variable $a_{i}$, we can use the skew-symmetric matrix $b$ to define a mutation relation at the node $k$ by
\begin{equation}
  \label{eq:mutation}
  a_{k} a_{k}' = \prod_{i \vert b_{i k} > 0} a_{i}^{b_{i k}} + \prod_{i \vert b_{i k} < 0} a_{i}^{-b_{i k}},
\end{equation} with the understanding that an empty product is set to one.  The mutation at $k$ changes $a_{k}$ to $a_{k}'$ defined by eq.~(\ref{eq:mutation}) and leaves the other cluster variables unchanged.

To illustrate these ideas we note that the initial cluster of the $a_{2}$ cluster algebra can be expressed by the quiver $a_{1} \to a_{2}$.  Then, a mutation at $a_{1}$ replaces it by $a_{1}' = \frac {1+a_{2}}{a_{1}} \equiv a_{3}$ and reverses the arrow.  A mutation at $a_{2}$ replaces it by $a_{2}' = \frac {1+a_{1}}{a_{2}} \equiv a_{5}$ and preserves the direction of the arrow.

\subsection{Cluster Poisson varieties}
These are defined using the same combinatorial skeleton: 
quivers and mutations of quivers. We assign now to the nodes of the quiver 
variables $\{x_i\}$ which mutate according to the following rule: 
\begin{equation}
  \label{eq:x-coords-mutation}
  x_{i}' =
  \begin{cases}
    x_{k}^{-1}, &\quad i=k,\\
    x_{i} (1+x_{k}^{\sgn b_{i k}})^{b_{i k}}, &\quad i \neq k.
  \end{cases}
\end{equation}
There is a natural Poisson bracket on the cluster $\mathcal{X}$-coordinates.  It is enough to define the Poisson bracket between the $\mathcal{X}$-coordinates in a given cluster, for which it is given in terms of the
antisymmetric
matrix $b_{i j}$ defined in~(\ref{eq:bijdef}) by
\begin{equation}
  \label{eq:poisson-x-coords}
  \lbrace x_{i}, x_{j}\rbrace = b_{i j} x_{i} x_{j}.
\end{equation}
An important property of this Poisson bracket is that it is invariant
under mutations, in the sense that
\begin{equation}
  \lbrace x_{i}', x_{j}'\rbrace = b_{i j}' x_{i}' x_{j}'
\end{equation}
whenever $x_{i}', x_{j}'$ and $b_{i j}'$ are obtained from $x_{i}, x_{j}$ and $b_{i j}$ by
a mutation.

Given a quiver described by the matrix $b$, the cluster $\mathcal{A}$- and $\mathcal{X}$-coordinates 
can be related as follows:
\begin{equation}
  \label{eq:xcoords-def}
  x_{i} = \prod_{j} a_{j}^{b_{i j}}.
\end{equation}  
This relation is preserved under mutations 
(changing, of course, the $b_{ij}$-matrix). 

We would like to stress that the cluster $\mathcal{A}$- and $\mathcal{X}$-coordinates 
live on different spaces of the same dimension, denoted  $\mathcal{A}$ and $\mathcal{X}$. Indeed, 
they are parametrized by the same set, the set of nodes of the corresponding quiver.  
Formula (\ref{eq:xcoords-def}) is just a coordinate way 
to express a canonical map of spaces $p: \mathcal{A} \to \mathcal{X}$. 
The dimension of the fibers of this map is the corank of the matrix $b$. 

A sequence of cluster mutations can result in reproducing the original quiver, while 
providing a non-trivial transformation of the cluster $\mathcal{A}$- and $\mathcal{X}$-coordinates. 
Such transformations form a group, introduced in \cite{FG03b} under the name \emph{cluster modular group}.

\subsection{Grassmannian cluster algebras and cluster Poisson spaces
  \texorpdfstring{$\Conf_n(\mathbb{P}^{k-1})$}{Conf n P k-1}}

In our application we are interested in a special class of cluster algebras called \emph{cluster algebras of geometric type}.  They are also described by quivers, but some of the nodes are special and called \emph{frozen nodes}.  Edges connecting two frozen nodes are not allowed,\footnote{There is a different approach, advocated in \cite{FG03b}, where half-edges between frozen variables are not only allowed, but in fact play a crucial role in the amalgamation construction: building bigger cluster structures from the smaller ones.}  and we also do not allow mutations on the frozen nodes.  The variables associated to the frozen nodes are called \emph{coefficients} instead of \emph{cluster variables} (and the rank of the algebra is equal only to the number of unfrozen nodes).  We define the \emph{principal part} of such a quiver to be the quiver obtained by erasing the frozen nodes as well as all edges which connect them to any of the non-frozen nodes.
When drawing these special kinds of quivers, we will indicate each frozen
node by placing its label inside a box.

We now review the cluster algebras of geometric type which arise from the Grassmannian $\Gr(k,n)$~\cite{1088.22009}.  
The Pl\"ucker coordinates $\langle i_1 \dotsc i_k \rangle$, being the minors obtained by computing the determinant of the indicated columns $i_1 \dotsc i_k$ of a matrix representative of a point in $\Gr(k,n)$, are examples of $\mathcal{A}$-coordinates.

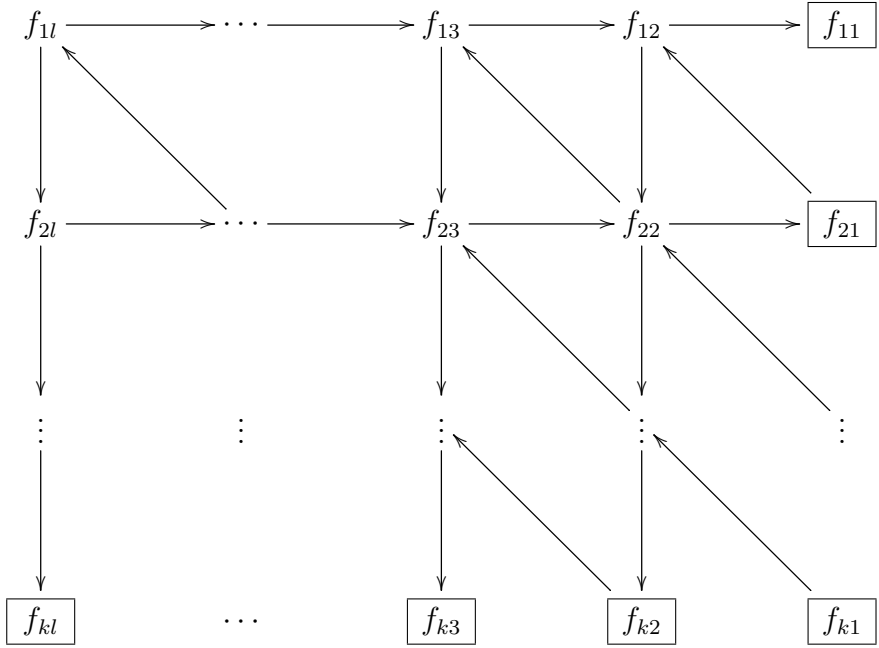
\begin{figure}
\centering
\begin{xy} 0;<1pt,0pt>:<0pt,-1pt>::
(0,0) *+{f_{1 l}} ="0",
(75,0) *+{\cdots} ="1",
(150,0) *+{f_{13}} ="2",
(225,0) *+{f_{12}} ="3",
(300,0) *+{\framebox[5ex]{$f_{11}$}} ="4",
(0,75) *+{f_{2 l}} ="5",
(75,75) *+{\cdots} ="6",
(150,75) *+{f_{23}} ="7",
(225,75) *+{f_{22}} ="8",
(300,75) *+{\framebox[5ex]{$f_{21}$}} ="9",
(0,150) *+{\vdots} ="10",
(75,150) *+{\vdots} ="11",
(150,150) *+{\vdots} ="12",
(225,150) *+{\vdots} ="13",
(300,150) *+{\vdots} ="14",
(0,225) *+{\framebox[5ex]{$f_{kl}$}} ="15",
(75,225) *+{\cdots} ="16",
(150,225) *+{\framebox[5ex]{$f_{k3}$}} ="17",
(225,225) *+{\framebox[5ex]{$f_{k2}$}} ="18",
(300,225) *+{\framebox[5ex]{$f_{k1}$}} ="19",
(-60,0) *+{} ="20",
"0", {\ar"1"},
"0", {\ar"5"},
"6", {\ar"0"},
"1", {\ar"2"},
"2", {\ar"3"},
"2", {\ar"7"},
"8", {\ar"2"},
"3", {\ar"4"},
"3", {\ar"8"},
"9", {\ar"3"},
"5", {\ar"6"},
"5", {\ar"10"},
"6", {\ar"7"},
"7", {\ar"8"},
"7", {\ar"12"},
"13", {\ar"7"},
"8", {\ar"9"},
"8", {\ar"13"},
"14", {\ar"8"},
"10", {\ar"15"},
"12", {\ar"17"},
"19", {\ar"13"},
"18", {\ar"12"},
"13", {\ar"18"},
\end{xy}
\caption{The initial quiver for the $\Gr(k,n)$ cluster algebra (see refs.~\cite{1057.53064,1215.16012}).}
\label{fig:gkn}
\end{figure}

The Pl\"ucker coordinates satisfy the relation
\begin{equation}
  \label{eq:plucker-rel}
  \ket{i, j, I} \ket{k, l, I} = \ket{i, k, I} \ket{j, l, I} +
\ket{i, l, I}\ket{j, k, I},
\end{equation}
where $I$ is a multi-index with $k-2$ entries,
which bears a resemblance to the exchange relation shown in
eq.~(\ref{eq:mutation}).
Indeed the cluster algebra for $\Gr(k,n)$ is constructed by starting
with an initial cluster whose variables are certain Pl\"ucker coordinates.
The operation of mutation generates additional Pl\"ucker coordinates,
as well as other, more complicated, cluster $\mathcal{A}$-coordinates.
For general $k$ and $n$ and with $l = n-k$, the appropriate initial quiver is given in ref.~\cite{1057.53064} (this construction is also reviewed in ref.~\cite{1215.16012}) and shown in fig.~\ref{fig:gkn},
where
\begin{equation}
  f_{i j} =
  \begin{cases}
    \frac {\langle i+1, \dotsc, k, k+j, \dotsc, i+j+k-1\rangle}{\langle 1, \dotsc, k\rangle}, \qquad &i \leq l-j+1,\\
    \frac {\langle 1, \dotsc, i+j-l-1, i+1, \dotsc, k, k+j, \dotsc, n\rangle}{\langle 1, \dotsc, k\rangle}, \qquad &i > l-j+1
  \end{cases}.
\end{equation}
The boxes identify the frozen variables while the rest of the variables are unfrozen.

In order to obtain the quivers we will use below we need to make one last change to the quiver above.  We rescale all the coordinates, frozen and unfrozen, by $\langle 1, \dotsc, k\rangle$.  This produces a frozen variable $\langle 1, \dotsc, k\rangle$ which connects to the node labeled by $f_{1l}$ by an ingoing arrow.  After this modification all the unfrozen nodes of the initial quiver have an equal number of ingoing and outgoing arrows.

The simplest nontrivial example is that of $\Gr(2,5)$, which is relevant for configurations of five points in $\mathbb{P}^1$.  In this case the initial quiver is simply
\begin{equation}
\begin{gathered}
\begin{xy} 0;<1pt,0pt>:<0pt,-1pt>::
(25,25) *+{\langle 13\rangle} ="0",
(75,25) *+{\langle 14\rangle} ="1",
(125,25) *+{\framebox[5ex]{$\langle 15\rangle$}} ="2",
(125,75) *+{\framebox[5ex]{$\langle 45\rangle$}} ="3",
(75,75) *+{\framebox[5ex]{$\langle 34\rangle$}} ="4",
(25,75) *+{\framebox[5ex]{$\langle 23\rangle$}} ="5",
(0,0) *+{\framebox[5ex]{$\langle 12\rangle$}} ="6",
(145,75) *+{},
"0", {\ar"1"},
"4", {\ar"0"},
"0", {\ar"5"},
"6", {\ar"0"},
"1", {\ar"2"},
"3", {\ar"1"},
"1", {\ar"4"},
\end{xy}
\end{gathered}
\end{equation}
where each node is labeled by its $\mathcal{A}$-coordinate.
In this
case it is easy to check that successive mutations on the two unfrozen nodes, in any order, generate only five distinct quivers.  The algebra so generated is nothing but the $A_2$ algebra defined at the beginning of sec.~\ref{sec:defin-intr}.
The name of the algebra comes from the fact that the principal part of this
initial quiver is the same as the Dynkin diagram of the $A_2$ Lie algebra.

Let us use this example to calculate $\mathcal{X}$-coordinates.  There is one such coordinate associated to each unfrozen node $k$, expressed by taking the product of the $\mathcal{A}$-coordinates on the nodes connected to $k$ by an incoming arrow, divided by the product of the $\mathcal{A}$-coordinates on the nodes
connected to $k$ by an outgoing arrow.
So in the $A_2$ quiver shown above the $\mathcal{X}$-coordinates
associated to the nodes $\ket{13}$ and $\ket{14}$ are
$\ket{12}\ket{34}/\ket{23}\ket{14}$ and $\ket{13}\ket{45}/\ket{34}\ket{15}$
respectively.
Cluster $\mathcal{A}$-coordinates are
not invariant under rescaling of the individual vectors in a $\Gr(k,n)$ matrix,
but the $\mathcal{X}$-coordinates are.  Therefore only the latter
are good coordinates on the quotient $\Gr(k,n)/(\mathbb{C}^*)^{n-1}
= \Conf_n(\mathbb{P}^{k-1})$ and are hence appropriate objects to
appear in motivic amplitudes.

To connect with the above general discussion of cluster Poisson varieties and 
$\mathcal{X}$-coordinates, in this example 
we use only the unfrozen part of the quiver 
to build the $\mathcal{X}$-coordinates. This leads to a reduced cluster $\mathcal{X}$-space 
$\mathcal{X}'$, and the map described coordinately in 
(\ref{eq:xcoords-def}) reduces to a surjective map 
$p': \mathcal{A} \to \mathcal{X}'$ describing the projection
$\widetilde \Gr(k,n) 
\to\Conf_n(\mathbb{P}^{k-1})$.

The two simplest examples relevant to SYM theory scattering amplitudes
are those for 6 or 7 points
in $\mathbb{P}^3$ (or, equivalently, in $\mathbb{P}^1$ or $\mathbb{P}^2$, respectively).
For the former
it is evident from fig.~\ref{fig:gkn} that the principal part of the quiver
is the same as the $A_3$ Dynkin diagram.  For the latter the initial quiver
is slightly more complicated:
\begin{equation}
\begin{gathered}
\begin{xy} 0;<1pt,0pt>:<0pt,-1pt>::
(50,50) *+{\langle 267\rangle} ="0",
(100,50) *+{\langle 367\rangle} ="1",
(150,50) *+{\langle 467\rangle} ="2",
(200,50) *+{\framebox[5ex]{$\langle 567\rangle$}} ="3",
(200,100) *+{\framebox[5ex]{$\langle 456\rangle$}} ="4",
(200,150) *+{\framebox[5ex]{$\langle 345\rangle$}} ="5",
(150,150) *+{\framebox[5ex]{$\langle 234\rangle$}} ="6",
(150,100) *+{\langle 346\rangle} ="7",
(100,100) *+{\langle 236\rangle} ="8",
(100,150) *+{\framebox[5ex]{$\langle 123\rangle$}} ="9",
(50,100) *+{\langle 126\rangle} ="10",
(50,150) *+{\framebox[5ex]{$\langle 127\rangle$}} ="11",
(0,0) *+{\framebox[5ex]{$\langle 167\rangle$}} ="12",
"0", {\ar"1"},
"8", {\ar"0"},
"0", {\ar"10"},
"12", {\ar"0"},
"1", {\ar"2"},
"7", {\ar"1"},
"1", {\ar"8"},
"2", {\ar"3"},
"4", {\ar"2"},
"2", {\ar"7"},
"7", {\ar"4"},
"5", {\ar"7"},
"7", {\ar"6"},
"6", {\ar"8"},
"8", {\ar"7"},
"8", {\ar"9"},
"10", {\ar"8"},
"9", {\ar"10"},
"10", {\ar"11"},
\end{xy}\end{gathered}.
\end{equation}
If we label the nodes occupied initially by $\langle 267\rangle$, $\langle 367\rangle$, $\langle 467\rangle$, $\langle 126\rangle$, $\langle 236\rangle$, $\langle 346\rangle$ by numbers $1$ through $6$, then after a sequence of mutations at nodes $4$, $3$, $2$, $5$, $1$, $4$, $3$, $4$, $6$,
the principal part of the quiver is brought into the form of the
$E_6$ Dynkin diagram\footnote{If we order them in the same way as in the initial cluster, the $\mathcal{A}$-coordinates after this sequence of mutations are $\langle 3 \times 4, 5 \times 6, 7 \times 1\rangle$, $\langle 256\rangle$, $\langle 124\rangle$, $\langle 247\rangle$, $\langle 5 \times 6, 7 \times 2, 3 \times 4\rangle$, $\langle 157\rangle$.}
\begin{equation}
\begin{gathered}
\begin{xy} 0;<1pt,0pt>:<0pt,-1pt>::
(50,50) *+{\langle 124\rangle} ="3",
(100,50) *+{\langle 247\rangle} ="4",
(180,0) *+{\langle 256\rangle} ="2",
(180,50) *+{\scriptstyle{\langle 5 \times 6, 7 \times 2, 3 \times 4\rangle}} ="5",
(270,50) *+{\scriptstyle{\langle 3 \times 4, 5 \times 6, 7 \times 1\rangle}} ="1",
(340,50) *+{\langle 157\rangle} ="6",
"2", {\ar"5"},
"4", {\ar"5"},
"1", {\ar"5"},
"4", {\ar"3"},
"1", {\ar"6"},
\end{xy}\end{gathered}.
\end{equation}
Therefore the $\Gr(3,7)$ cluster algebra is also called the $E_6$
algebra.

In~\cite{1054.17024} Fomin and Zelevinsky showed that a cluster algebra is of finite type (i.e., it has a finite number of cluster variables) if there exists a sequence of mutations which turns the principal part of its quiver into the Dynkin diagram of some classical Lie algebra.  However, if the principal part of the quiver contains a subgraph which is an affine Dynkin diagram, then the cluster algebra is of infinite type.

In ref.~\cite{1088.22009}, Scott classified all the Grassmannian cluster algebras of finite type.
As discussed in sec.~\ref{sec:kinematics}, the relevant Grassmannian for scattering amplitudes in SYM theory is $\Gr(4,n)$, for $n \geq 6$. If $n=6$ we need $\Gr(4,6) = \Gr(2,6)$ which is of finite type $A_3$.  If $n=7$ we need $\Gr(4,7) = \Gr(3,7)$ which is again of finite type $E_6$.  However, starting at $n=8$ the relevant cluster algebras are not of finite type anymore.  This indicates that there are infinitely many different $\mathcal{A}$-coordinates which could appear in the symbol of these amplitudes, and infinitely many different
$\mathcal{X}$-coordinates could appear in their coproduct.

Besides the usual Pl\"ucker determinants, mutations also lead in general
to more complicated $\mathcal{A}$-coordinates. For example, with $\Gr(4,8)$ one encounters
\begin{equation}
\label{eq:atwo}
  \langle 1 2 (3 4 5) \cap (6 7 8)\rangle \equiv \langle 1 3 4 5\rangle \langle 2 6 7 8\rangle - \langle 2 3 4 5\rangle \langle 1 6 7 8\rangle.
\end{equation}
As discussed in sec.~\ref{sec:kinematics}, the notation with $\cap$ emphasizes the following geometrical fact: the composite bracket $\langle 1 2 (3 4 5) \cap (6 7 8)\rangle$ vanishes whenever the points $1$ and $2$ are coplanar with the projective line $(3 4 5) \cap (6 7 8)$ obtained by intersecting the two projective planes $(3 4 5)$ and $(6 7 8)$.

One miraculous feature of the mutations is that the denominator can always be
canceled by the numerator, after using Pl\"ucker identities.  Therefore, the
$\mathcal{A}$-coordinates always end up being \emph{polynomials} in the Pl\"ucker
coordinates.  This is an analog of the Laurent phenomenon mentioned above.
An example which appears for $\Gr(4,8)$ is
\begin{equation}
\label{eq:aone}
 \frac {\langle 1 2 3 7\rangle \langle 1 2 4 5\rangle \langle 1 6 7 8\rangle + \langle 1 2 7 8\rangle \langle 4 5 (6 7 1)\cap(1 2 3)\rangle}{\langle 1 2 6 7\rangle}   =  \langle 4 5 (7 8 1) \cap (1 2 3)\rangle.
\end{equation}
Here the left-hand side is the raw expression obtained
for a certain $\mathcal{A}$-coordinate
following some mutation,
while the right-hand side is a simplified expression where the denominator has been
canceled by applying a Pl\"ucker relation to the numerator in order
to pull out all overall factor of $\ket{1267}$.

Both of these types of non-Pl\"ucker $\mathcal{A}$-coordinates, eqs.~(\ref{eq:aone})
and~(\ref{eq:atwo}),
appear
in the symbol of the two-loop $n=8$ MHV
amplitude~\cite{CaronHuot:2011ky}, and the simpler ones of the
type in eq.~(\ref{eq:aone}) appear already for $n=7$---indeed
the long formulas in
section~\ref{sec:mc} are littered with these composite brackets (though
expressed there in the $\mathbb{P}^2$ language).

However, since the number of $\mathcal{A}$-coordinates is infinite for $\Gr(4,8)$,
mutations must eventually generate even more exotic $\mathcal{A}$-coordinates.
A still relatively simple example is
\begin{equation}
-\langle (1 2 3) \cap (3 4 5), (5 6 7) \cap (7 8 1)\rangle.
\end{equation}
This quantity vanishes when the lines $(1 2 3) \cap (3 4 5)$ and $(5 6 7) \cap (7 8 1)$ intersect,
which is equivalent to saying that the lines $(3 4 5) \cap (5 6 7)$ and $(7 8 1) \cap (1 2 3)$ intersect.
But there are even more complicated $\mathcal{A}$-coordinates such as
\begin{multline}
  \langle 1 2 4 6\rangle \langle 1 2 5 6\rangle \langle 1 3 7 8\rangle \langle 3 4 5 7\rangle -
  \langle 1 2 4 6\rangle \langle 1 2 5 7\rangle \langle 1 3 7 8\rangle \langle 3 4 5 6\rangle -\\
  \langle 1 2 4 6\rangle \langle 1 2 7 8\rangle \langle 1 3 5 6\rangle \langle 3 4 5 7\rangle +
  \langle 1 2 7 8\rangle \langle 1 2 5 7\rangle \langle 1 3 4 6\rangle \langle 3 4 5 6\rangle +\\
  \langle 1 2 3 6\rangle \langle 1 2 7 8\rangle \langle 1 4 5 7\rangle \langle 3 4 5 6\rangle.
\end{multline}
Neither of these complicated quantities appears as an entry in the symbol
of the
$n=8$ MHV amplitude at two loops, but we know of no reason why they cannot appear
at higher loops.
It would be extremely interesting
to understand more about these algebras and their relation to amplitudes.

One final, important comment has to do with cyclic symmetry, which is
an exact symmetry of MHV amplitudes (and of all super-amplitudes).
Notice that the initial quivers we have been using break the cyclic symmetry of the configuration of points.  In order to see that the cyclic symmetry is preserved we need to show that by mutations one can reach another quiver whose labels are permuted by one unit.  For the case of $\Gr(3,7)$ described above, this cannot be done in fewer than six mutations, since all the unfrozen $\mathcal{A}$-coordinates need to change.  Indeed one can easily show that after mutating in the nodes which are initially labeled by $\langle 126 \rangle$, $\langle 267\rangle$, $\langle 236\rangle$, $\langle 367\rangle$, $\langle 346\rangle$ and $\langle 467\rangle$, we obtain the cluster with the node labels shifted by one $\langle 123\rangle \to \langle 234\rangle$, etc.  This proves the cyclic symmetry for $\Gr(3,7)$. It is not hard to imagine that a similar procedure can be applied in the general $\Gr(k,n)$ case, but we do not provide a complete proof of cyclic symmetry here.

\subsection{Generalized Stasheff polytopes}
\label{sec:associahedra}

In this section we review the connection~\cite{1057.52003} between cluster
algebras and certain
polytopes, including the generalized Stasheff
polytope or associahedron~\cite{0114.39402}.
Many additional details and examples may also be found
in~\cite{MR2383126}.

The unfrozen nodes of a cluster algebra of rank $r$ can be taken to be the vertices of an $r-1$-simplex.  A $k$-simplex is a generalization of the notion of a triangle and can be defined as a convex hull of its $k+1$ vertices.  A $0$-simplex is a point, a $1$-simplex is a line, a $2$-simplex is a triangle, a $3$-simplex is a tetrahedron, and so on.

If we take a subset of $l+1$ vertices of the $k+1$ vertices of a $k$-simplex we can form an $l$-simplex which is called a \emph{face} of the $k$-simplex.  The number of $l$-faces of a $k$-simplex is $\binom{k+1}{l+1}$.

When doing a mutation in a cluster algebra of rank $r$, one of the $\mathcal{A}$-coordinates changes while the other $r-1$ stay unchanged.  They define an $r-2$-face of the initial $r-1$-simplex.  Therefore, after mutation we obtain a new $r-1$-simplex which shares an $r-2$-face with the initial simplex.  We can glue them along this face to form an $r-1$-dimensional polytope.  For cluster algebras of finite type, by doing all possible mutations, we can build a polytope out of finitely many simplices.

In this language, the $\mathcal{X}$-coordinates of the cluster algebra correspond to $r-2$-faces of the $r-1$-simplex.  The number of such faces is $\binom{r}{r-1} = r$, which equal to the rank of the algebra. Under mutations, one of the $\mathcal{X}$-coordinates transforms to its inverse while the others transform in a more complicated way.  So more properly we should associate to each $r-2$-face the pair consisting of an $\mathcal{X}$-coordinate and its inverse.

The dual of the polytope we have described is a generalized Stasheff
polytope (or generalized associahedron, for a reason we will
describe in the following section) associated to the cluster algebra.
It is a theorem which is deduced from~\cite{1021.16017}, 
that for any cluster algebras,
the faces are always either quadrilaterals or pentagons.  For example,
the Stasheff polytope associated to $\Gr(2,6)$ has 3 quadrilateral
faces and 6 pentagonal faces, while the $\Gr(3,7)$ polytope
has 1785 quadrilaterals and 1071 pentagons.

\subsection{Poisson bracket and generalized Stasheff polytopes}
\label{sec:poisson-brackets}

There is a simple connection between the Poisson bracket and
the geometry of the Stasheff polytope, which plays an important role 
in elucidating the structure of motivic amplitudes. 

Two $\mathcal{X}$-coordinates
$x_1,x_2$ have zero Poisson bracket if the
Stasheff polytope has some quadrilateral face containing both $x_1$ and
$x_2$ at each node in the configuration
\begin{equation}
\label{eq:square}
\begin{gathered}
\begin{xy} 0;<1pt,0pt>:<0pt,-1pt>::
(150,0) *+{\framebox[20ex]{$\{1/x_1,x_2,\ldots\}$}}  ="0",
(150,50) *+{\framebox[20ex]{$\{1/x_1,1/x_2,\ldots\}$}}  ="1",
(0,50) *+{\framebox[20ex]{$\{x_1,1/x_2,\ldots\}$}} ="2",
(0,0) *+{\framebox[20ex]{$\{x_1,x_2,\ldots\}$}} ="3",
"0";"1" **\dir{-};
"3";"0" **\dir{-};
"1";"2" **\dir{-};
"2";"3" **\dir{-};
\end{xy}\end{gathered},
\end{equation}
where the dots stand for other $\mathcal{X}$-coordinates (some of
which may overlap between some, but not all, of the four corners).
In such a case the variables $x_1,x_2$ form a closed
$A_1 \times A_1$ subalgebra. 
Moving left-to-right or up-to-down is accomplished by mutating on
$x_1$ or $x_2$, respectively.

Two $\mathcal{X}$-coordinates $x_1,x_2$ have Poisson
bracket\footnote{Here and in the following, when we say that two cluster
$\mathcal{X}$ coordinates have Poisson bracket $\pm 1$ we mean that
the Poisson bracket of their logarithms is $\pm 1$.}
$\pm 1$ if they form a closed $A_2$ subalgebra.
In this case the Stasheff polytope has some pentagonal face containing
$x_1$ and $x_2$ in the configuration
\begin{equation}
\label{eq:pentagon}
\begin{gathered}
\begin{xy} 0;<1pt,0pt>:<0pt,-1pt>::
(200,50) *+{\framebox[20ex]{$\{\frac{x_1}{1+x_2},x_2\ldots\}$}} ="0",
(200,100) *+{\framebox[20ex]{$\{\frac{1+x_2}{x_1},\frac{x_1x_2}{1+x_1+x_2},\ldots\}$}} ="1",
(0,100) *+{\framebox[20ex]{$\{\frac{1+x_1+x_2}{x_1x_2},\frac{1+x_1}{x_2},\ldots\}$}} ="2",
(0,50) *+{\framebox[20ex]{$\{1/x_1,\frac{x_2}{1+x_1},\ldots\}$}} ="3",
(100,0) *+{\framebox[20ex]{$\{x_1,1/x_{2},\ldots\}$}} ="4",
"0";"1" **\dir{-};
"4";"0" **\dir{-};
"1";"2" **\dir{-};
"2";"3" **\dir{-};
"3";"4" **\dir{-};
\end{xy}\end{gathered},
\end{equation}
in which case the variables $x_1,x_2$ form an $A_2$ subalgebra.

\subsection{Parity invariance}
\label{sec:parity-invariance}

In this section we show that the parity operation (reviewed in
appendix~\ref{sec:parity-conj-twist}) is an element of the cluster modular group.
Specifically, using the identities in appendix~\ref{sec:parity-conj-twist},
we verify that the parity transform of any quiver is related by a sequence
of mutations to the original quiver.  This implies that the
set of cluster $\mathcal{X}$-coordinates
is closed under parity.
It would be very interesting to see if the rest of the cluster modular
group plays some role, or has a nice interpretation when
acting on motivic amplitudes.

Of course, in simple cases like six points in $\mathbb{P}^{3}$ the parity invariance of the set of cluster $\mathcal{X}$-coordinates can be explicitly checked by enumerating all of them.   However, due to the large number of cluster coordinates, this is much more difficult for seven points and it is impossible for more than seven points since then the cluster algebras are of infinite type.

For six points in $\mathbb{P}^{3}$ the initial quiver is shown in fig.~\ref{fig:six-point-parity-a}.
\begin{figure}
\centering
  \begin{subfigure}[b]{.3\textwidth}
  \begin{xy} 0;<1pt,0pt>:<0pt,-1pt>::
(50,50) *+{\scriptstyle{\langle 1 2 3 5\rangle}} ="0",
(50,100) *+{\scriptstyle{\langle 1 2 4 5\rangle}} ="1",
(50,150) *+{\scriptstyle{\langle 1 3 4 5\rangle}} ="2",
(50,200) *+{\framebox[5ex]{$\scriptstyle{\langle 2 3 4 5\rangle}$}} ="3",
(100,200) *+{\framebox[5ex]{$\scriptstyle{\langle 3 4 5 6\rangle}$}} ="4",
(100,150) *+{\framebox[5ex]{$\scriptstyle{\langle 1 4 5 6\rangle}$}}="5",
(100,100) *+{\framebox[5ex]{$\scriptstyle{\langle 1 2 5 6\rangle}$}} ="6",
(100,50) *+{\framebox[5ex]{$\scriptstyle{\langle 1 2 3 6\rangle}$}} ="7",
(0,0) *+{\framebox[5ex]{$\scriptstyle{\langle 1 2 3 4\rangle}$}} ="8",
"0", {\ar"1"},
"6", {\ar"0"},
"0", {\ar"7"},
"8", {\ar"0"},
"1", {\ar"2"},
"5", {\ar"1"},
"1", {\ar"6"},
"2", {\ar"3"},
"4", {\ar"2"},
"2", {\ar"5"},
  \end{xy}
  \caption{}
  \label{fig:six-point-parity-a}
  \end{subfigure} \hspace{2cm}
  \begin{subfigure}[b]{.3\textwidth}
  \begin{xy} 0;<1pt,0pt>:<0pt,-1pt>::
(50,50) *+{\scriptstyle{\langle 2 4 5 6\rangle}} ="0",
(50,100) *+{\scriptstyle{\langle 1 2 4 5\rangle}} ="1",
(50,150) *+{\scriptstyle{\langle 1 2 4 6\rangle}} ="2",
(50,200) *+{\framebox[5ex]{$\scriptstyle{\langle 1 2 5 6\rangle}$}} ="3",
(100,200) *+{\framebox[5ex]{$\scriptstyle{\langle 1 2 3 6\rangle}$}} ="4",
(100,150) *+{\framebox[5ex]{$\scriptstyle{\langle 1 2 3 4\rangle}$}} ="5",
(100,100) *+{\framebox[5ex]{$\scriptstyle{\langle 2 3 4 5\rangle}$}} ="6",
(100,50) *+{\framebox[5ex]{$\scriptstyle{\langle 3 4 5 6\rangle}$}}  ="7",
(0,0) *+{\framebox[5ex]{$\scriptstyle{\langle 1 4 5 6\rangle}$}} ="8",
"0", {\ar"1"},
"6", {\ar"0"},
"0", {\ar"7"},
"8", {\ar"0"},
"1", {\ar"2"},
"5", {\ar"1"},
"1", {\ar"6"},
"2", {\ar"3"},
"4", {\ar"2"},
"2", {\ar"5"},
  \end{xy}
  \caption{}
  \label{fig:six-point-parity-b}
  \end{subfigure}
  \caption{The initial quiver (a) for $\Gr(4,6)$ and its parity conjugate (b).}
  \label{fig:six-point-parity}
\end{figure}
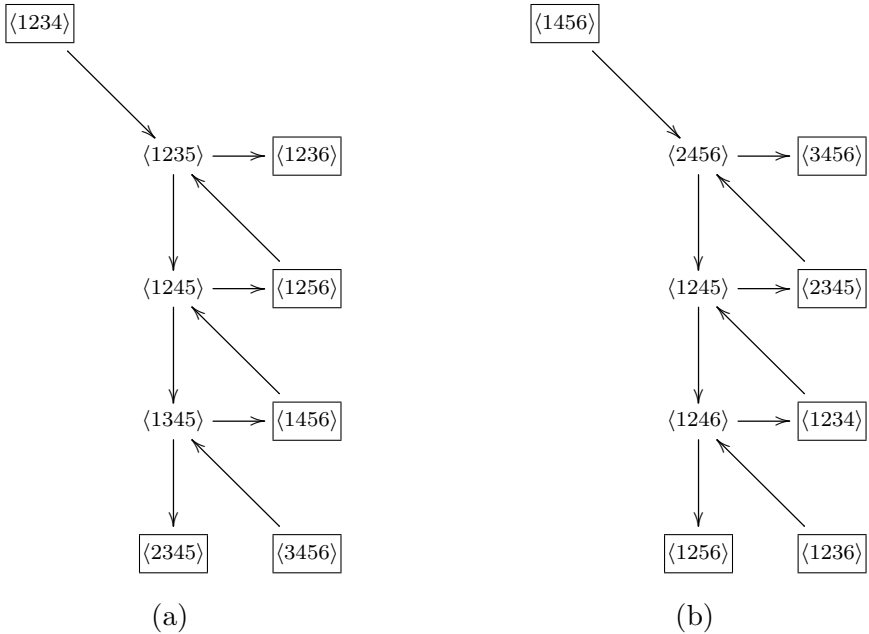  Parity amounts to replacing $\langle i j k l\rangle \to [i j k l]$.  The angle brackets $\langle i j k l\rangle$ are invariants made up of twistors $Z_{i}, Z_{j}, Z_{k}, Z_{k}$ whereas the square brackets $[i,j,k,l]$ are invariants made up of dual twistors $W_{i}, W_{j}, W_{k}, W_{k}$.  Dual twistors can be written in terms of twistors as\footnote{This is often written as $W_{i} = \tfrac {Z_{i-1} \wedge Z_{i} \wedge Z_{i+1}}{\langle i-1 i\rangle \langle i i+1\rangle}$ such that $W_{i}$ and $Z_{i}$ scale with opposite weight.  The two-brackets $\langle i j\rangle$ are defined by choosing an arbitrary line $I$ (also called `infinity twistor') and setting $\langle i j\rangle = \langle I i j\rangle$.  When constructing cross-ratios these two-brackets cancel out so in the following we will not keep track of them.} $W_{i} = Z_{i-1} \wedge Z_{i} \wedge Z_{i+1}$.  Then we rewrite the $[i j k l]$ in terms of angle brackets, as follows
\begin{gather*}
  [1 2 3 5] = \langle 6 1 2 3\rangle \langle 1 2 3 4\rangle \langle 2 4 5 6\rangle, \qquad
  [1 2 4 5] = \langle 6 1 2 3\rangle \langle 3 4 5 6\rangle \langle 1 2 4 5\rangle, \\
  [1 3 4 5] = \langle 2 3 4 5\rangle \langle 3 4 5 6\rangle \langle 6 1 2 4\rangle, \qquad
  [1 2 3 4] = \langle 6 1 2 3\rangle \langle 1 2 3 4\rangle \langle 2 3 4 5\rangle,\\ \qquad \text{cyclic permutations of $[1 2 3 4]$}.
\end{gather*}  The $\mathcal{X}$-coordinates of the quiver in fig.~\ref{fig:six-point-parity-b} generate parity conjugates of the $\mathcal{X}$-coordinates of the quiver in fig.~\ref{fig:six-point-parity-a}.  This quiver can be obtained from the initial quiver by mutations, but with opposite directions of the arrows.  Switching the direction of all arrows replaces all the cross-ratios by their inverses.  This does not change the set of cluster coordinates since if a cross-ratio is a cluster $\mathcal{X}$-coordinate then its inverse is also a cluster $\mathcal{X}$-coordinate.

We should note here that the parity transformation is, up to signs, the same as shifting all the points by three.  For example, $\langle 1 2 3 4\rangle \to \langle 1 4 5 6\rangle$, $\langle 1 2 3 5\rangle \to \langle 2 4 5 6\rangle$.

The seven-point case is a bit more complicated.  Here also we start with the initial quiver in fig.~\ref{fig:seven-point-parity-a} to which we apply parity $\langle i j k l\rangle \to [i j k l]$.  Let us focus on the $\mathcal{X}$-coordinate sitting at the node labeled by $\langle 1 2 3 5\rangle$ in fig.~\ref{fig:seven-point-parity-a}.  It is given by $\tfrac {\langle 1 2 3 4\rangle \langle 1 2 5 6\rangle}{\langle 1 2 3 6\rangle \langle 1 2 4 5\rangle}$.

\begin{figure}
  \centering
\begin{subfigure}[b]{.4\textwidth}
  \begin{xy} 0;<1pt,0pt>:<0pt,-1pt>::
(50,50) *+{\scriptstyle{\langle 1 2 3 5\rangle}} ="0",
(50,100) *+{\scriptstyle{\langle 1 2 4 5\rangle}} ="1",
(50,150) *+{\scriptstyle{\langle 1 3 4 5\rangle}} ="2",
(100,50) *+{\scriptstyle{\langle 1 2 3 6\rangle}} ="3",
(100,100) *+{\scriptstyle{\langle 1 2 5 6\rangle}} ="4",
(100,150) *+{\scriptstyle{\langle 1 4 5 6\rangle}} ="5",
(0,0) *+{\framebox[5ex]{$\scriptstyle{\langle 1 2 3 4\rangle}$}} ="6",
(150,50) *+{\framebox[5ex]{$\scriptstyle{\langle 1 2 3 7\rangle}$}} ="7",
(150,100) *+{\framebox[5ex]{$\scriptstyle{\langle 1 2 6 7\rangle}$}} ="8",
(150,150) *+{\framebox[5ex]{$\scriptstyle{\langle 1 5 6 7\rangle}$}}="9",
(150,200) *+{\framebox[5ex]{$\scriptstyle{\langle 4 5 6 7\rangle}$}}="10",
(100,200) *+{\framebox[5ex]{$\scriptstyle{\langle 3 4 5 6\rangle}$}} ="11",
(50,200) *+{\framebox[5ex]{$\scriptstyle{\langle 2 3 4 5\rangle}$}}="12",
"0", {\ar"1"},
"0", {\ar"3"},
"4", {\ar"0"},
"6", {\ar"0"},
"1", {\ar"2"},
"1", {\ar"4"},
"5", {\ar"1"},
"2", {\ar"5"},
"11", {\ar"2"},
"2", {\ar"12"},
"3", {\ar"4"},
"3", {\ar"7"},
"4", {\ar"5"},
"4", {\ar"8"},
"8", {\ar"3"},
"9", {\ar"4"},
"5", {\ar"9"},
"10", {\ar"5"},
"5", {\ar"11"},
\end{xy}
  \caption{}
  \label{fig:seven-point-parity-a}
  \end{subfigure} \hspace{0cm}
  \begin{subfigure}[b]{.4\textwidth}
  \begin{xy} 0;<1pt,0pt>:<0pt,-1pt>::
(50,50) *+{\scriptstyle{\langle 1 2 5 7\rangle}} ="0",
(50,100) *+{\scriptstyle{\langle 1 2 5 6\rangle}} ="1",
(50,150) *+{\scriptstyle{\langle 2 5 6 7\rangle}} ="2",
(100,50) *+{\scriptstyle{\langle 1 2 4 7\rangle}} ="3",
(100,100) *+{\scriptstyle{\langle 1 2 4 5\rangle}} ="4",
(100,150) *+{\scriptstyle{\langle 2 4 5 6\rangle}} ="5",
(0,0) *+{\framebox[5ex]{$\scriptstyle{\langle 1 2 6 7\rangle}$}}  ="6",
(150,50) *+{\framebox[5ex]{$\scriptstyle{\langle 1 2 3 7\rangle}$}} ="7",
(150,100) *+{\framebox[5ex]{$\scriptstyle{\langle 1 2 3 4\rangle}$}}  ="8",
(150,150) *+{\framebox[5ex]{$\scriptstyle{\langle 2 3 4 5\rangle}$}}="9",
(150,200) *+{\framebox[5ex]{$\scriptstyle{\langle 3 4 5 6\rangle}$}} ="10",
(100,200) *+{\framebox[5ex]{$\scriptstyle{\langle 4 5 6 7\rangle}$}} ="11",
(50,200) *+{\framebox[5ex]{$\scriptstyle{\langle 1 5 6 7\rangle}$}} ="12",
"0", {\ar"1"},
"0", {\ar"3"},
"4", {\ar"0"},
"6", {\ar"0"},
"1", {\ar"2"},
"1", {\ar"4"},
"5", {\ar"1"},
"2", {\ar"5"},
"11", {\ar"2"},
"2", {\ar"12"},
"3", {\ar"4"},
"3", {\ar"7"},
"4", {\ar"5"},
"4", {\ar"8"},
"8", {\ar"3"},
"9", {\ar"4"},
"5", {\ar"9"},
"10", {\ar"5"},
"5", {\ar"11"},
\end{xy}
  \caption{}
  \label{fig:seven-point-parity-b}
  \end{subfigure}
  \caption{The initial quiver (a) for $\Gr(4,7)$ and its partner (b) which makes the parity conjugation property manifest.}
  \label{fig:seven-point-parity}
\end{figure}
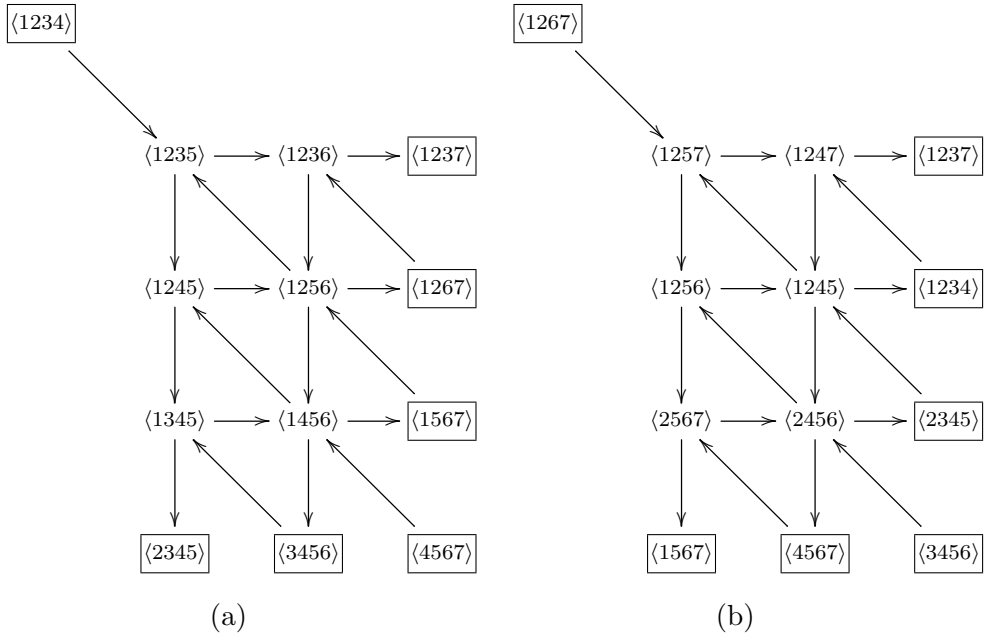

Using equations similar to the ones we used above for six points, we can write the parity conjugate of the quiver $\mathcal{X}$-coordinate at node $\langle 1 2 3 5\rangle$ as
\begin{equation}
  \frac {\langle 1 2 3 4\rangle \langle 1 2 5 6\rangle}{\langle 1 2 3 6\rangle \langle 1 2 4 5\rangle} \to \frac {[1 2 3 4] [1 2 5 6]}{[1 2 3 6] [1 2 4 5]} = \frac {\langle 1 2 5 6\rangle \langle 2 3 4 5\rangle \langle 4 5 6 7\rangle}{\langle 1 2 4 5\rangle \langle 2 5 6 7\rangle \langle 3 4 5 6\rangle}.
\end{equation}  The parity conjugated cross-ratio is the same as the inverse cross-ratio sitting at the opposite corner, at $\langle 2 4 5 6\rangle$ in the partner quiver in fig.~\ref{fig:seven-point-parity-b}.  Each of the unfrozen variables has a correspondent among the unfrozen variables of the partner quiver.  The $\mathcal{X}$-coordinates of nodes which are in correspondence are inverse and parity conjugate to one another.  The correspondence between nodes is defined as follows: the unfrozen nodes fit in a rectangular pattern.  Flip this rectangular pattern along the vertical and the horizontal.  After these flips the pattern of unfrozen nodes fits over the pattern of unfrozen nodes of the partner quiver (note that after superposition all the arrows point in the opposite directions after this sequence of operations).  For example, the correspondence between the nodes in figs.~\ref{fig:seven-point-parity-a},~\ref{fig:seven-point-parity-b} is the following: $\langle 1 2 3 5\rangle \leftrightarrow \langle 2 4 5 6\rangle$, $\langle 1 2 4 5\rangle \leftrightarrow \langle 1 2 4 5\rangle$, $\langle 1 2 5 6\rangle \leftrightarrow \langle 1 2 5 6\rangle$, $\langle 1 2 3 6\rangle \leftrightarrow \langle 2 5 6 7\rangle$, $\langle 1 3 4 5\rangle \leftrightarrow \langle 1 2 4 7\rangle$, $\langle 1 4 5 6\rangle \leftrightarrow \langle 1 2 5 7\rangle$.

Now we can show that the set of cluster $\mathcal{X}$-coordinates is closed under parity conjugation.  First, it is easy to show that the quiver in fig.~\ref{fig:seven-point-parity-b} can be obtained from the quiver in fig.~\ref{fig:seven-point-parity-a} after four mutations.  Another way to show that the two quivers can be obtained from one another by mutations is to notice that the quiver in fig.~\ref{fig:seven-point-parity-a} can be transformed to the quiver in fig.~\ref{fig:seven-point-parity-b} by a dihedral transformation of external data $1 \leftrightarrow 2$, $3 \leftrightarrow 7$, $4 \leftrightarrow 6$.  Then the conclusion follows from the dihedral symmetry of the cluster algebra.  So they generate the same cluster algebra.  Moreover, for every sequence of mutations in the quiver in fig.~\ref{fig:seven-point-parity-a}, we can perform the same sequence of mutations in the corresponding nodes of the partner quiver in fig.~\ref{fig:seven-point-parity-b} and we obtain the inverses of parity conjugate $\mathcal{X}$-coordinates.  This analysis can be extended without difficulty to the general case of cluster algebras $G(4,n)$.

\subsection{Cluster algebras and the positive Grassmannian}

The positive Grassmannian is defined as the subspace of the real
Grassmannian for which the ordered Pl\"ucker coordinates are all positive (see also
appendix~\ref{sec:parity-conj-twist} for more details):
\begin{equation}
\Gr^+(k,n) = \{ ( c_1 \cdots c_n ) \in \Gr(k,n,\mathbb{R}) :
\ket{c_{a_1} \cdots c_{a_k}} > 0\; \text{whenever $a_1 < \cdots < a_k$}\}.
\nonumber
\end{equation}
The mutation relation~(\ref{eq:mutation}) clearly respects
positivity:  if the $\mathcal{A}$-coordinates in the initial cluster are
all chosen to be positive, then all subsequently generated
$\mathcal{A}$-coordinates, in every cluster, will continue to be positive. The same trivially holds for $\mathcal{X}$-coordinates since they are just products of $\mathcal{A}$-coordinates.

It is manifest that set of positive configurations of points on $\mathbb{P}^3$
constitutes a subspace of what physicists call the Euclidean region, in which
scattering amplitudes are expected to be smooth, real-valued functions.
This is in strong accord with the main slogan of~\cite{ArkaniHamed:2012nw},
though we are talking here about positivity in the external kinematic data,
rather than in the Grassmannian of internal (i.e., loop integration) variables.
It seems clear that the full power of positivity has not yet been unleashed.

By construction, any
cluster $\mathcal{X}$-coordinate $x$ has the property that
$1+x$ factors into a product of $\mathcal{A}$-coordinates.
We have found empirically that
positivity allows for a quick criterion to go the other way around:
suppose we have identified some cross-ratio $r$ for which $1+r$
so factors; how do we determine whether or not $r$ is
a cluster $\mathcal{X}$-coordinate?  The answer is simply to evaluate the
triple
\begin{equation}
\{  r,  - 1 - r, - 1 - 1/r \}
\end{equation}
at a random point in the positive Grassmannian.  In all of our experience
to date, one of these three quantities will be positive and the other two
negative; the positive one is an $\mathcal{X}$-coordinate and the other two
are not.  Certainly this criterion is valid for $n=6,7$ where we can enumerate
all such possibilities, and it has held true at higher $n$ in all cases
we have looked at.  However in the infinite-dimensional algebras
we cannot exclude the possibility
that there might exist some
sufficiently complicated cross-ratio $r$ for which $1+r$ factorizes, yet
no member of the above list is an $\mathcal{X}$-coordinate.
Let us note also that for a given cross-ratio
$r$, different elements of
the above list may be $\mathcal{X}$-coordinates with respect to different
orderings of the external points.

\section{Cluster Coordinates and Motivic Analysis for \texorpdfstring{$n=6,7$}{n=6,7}}
\label{sec:examples}

We now have built up all of the machinery we need in order to carry
out a full motivic analysis of the two-loop $n=6,7$ MHV amplitudes.
To that end we begin this section with a detailed discussion of the
cluster coordinates and Stasheff polytopes for $\Gr(2,6)$ and $\Gr(3,7)$.

\subsection{Clusters and coordinates for \texorpdfstring{$\Gr(2,6)$}{G(2,6)}}

Beginning with the initial quiver for $\Gr(2,6)$, we can generate all
of the clusters and their $\mathcal{A}$- and $\mathcal{X}$-coordinates
by successively mutating at various nodes.
The $A_3$
cluster
algebra generated in this manner
has a total of 15 $\mathcal{A}$-coordinates, which are the standard Pl\"ucker coordinates $\ket{i j}$ on $\Gr(2,6)$.  The six coordinates with $i$ and $j$ adjacent (mod 6) are frozen, while the remaining nine are unfrozen.

However, in the special case of $\Gr(2,n)$ cluster algebras the mutations can also be given a simple geometric interpretation which we now describe.  The discussion in this section follows ref.~\cite{FG03b}.

We start with a configuration of $n$ points in $\mathbb{P}^{1}$ with coordinates $z_{i}$, $i = 1, \dotsc, n$ and we fix a cyclic ordering.  To these points we associate a convex polygon, with each vertex of the polygon labeled by one coordinate $z_{i}$.  Then, consider a complete triangulation $T$ of this polygon, as in fig.~\ref{fig:triangulation1}.

A triangulation $T$ and a diagonal $E$ of that triangulation uniquely determine a quadrilateral for which $E$ is a diagonal.  The points in $\mathbb{P}^{1}$ corresponding to the vertices of this quadrilateral have a cross-ratio in $\mathbb{P}^{1}$.  For example, to the diagonal $E$ in fig.~\ref{fig:triangulation1} we associate the cross-ratio
\begin{equation}
  r(3,5,1,2) = r(1,2,3,5) \equiv \frac{(z_{1}-z_{2})(z_{3}-z_{5})}{(z_{2}-z_{3})(z_{1}-z_{5})}.
\end{equation}  Note that the vertices of the quadrilateral are read in the same order as the cyclic order of the $n$ points, starting at one of the points incident with the diagonal $E$.  Note that it doesn't matter which of the two points incident with $E$ we start with, due to the identity
\begin{equation}
r(k,l,i,j) = r(i,j,k,l).
\end{equation}
Moreover, if we read the list of vertices in the opposite order, we obtain the inverse cross-ratio since
\begin{equation}
r(i,l,k,j) = \frac{1}{r(i,j,k,l)}.
\end{equation}

\begin{figure}
  \centering
  \includegraphics{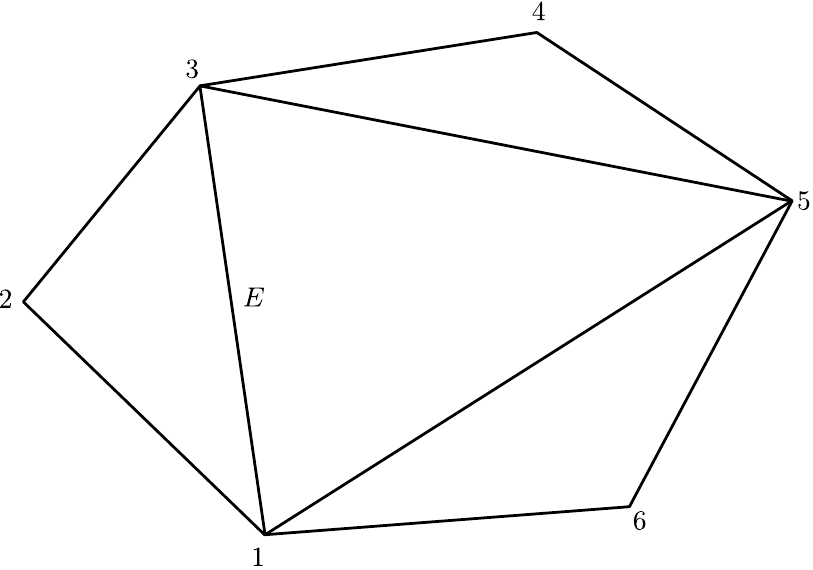}
  \caption{A triangulation $T$ of the hexagon.  One of the edges of the triangulation is marked by $E$.}
  \label{fig:triangulation1}
\end{figure}

Now we introduce a function $b$ which associates to a pair of diagonals in a triangulation a number which is $0$ or $\pm 1$.  Two diagonals in a triangulation are called adjacent if they are the sides of one of the triangles of the triangulation.  If two diagonals $E$ and $F$ are not adjacent we set $b_{E F} = 0$.  If $E$ and $F$ are adjacent we set $b_{E F} = 1$ if $E$ comes before $F$ when listing the diagonals at $E \cap F$ in anticlockwise order.  Otherwise we set $b_{E F} = -1$.

Now we can define cluster transformations (or mutations).  The starting point is a triangulation, to which we can associate a set of cross-ratios, as described above (one cross-ratio for each diagonal).  A cluster transformation is obtained by picking one of the diagonals and replacing it with the other diagonal in the same quadrilateral.  A sequence of such mutations is represented in fig.~\ref{fig:triangulation2}, where after five steps we reach the original configuration.

\begin{figure}
  \centering
  \includegraphics{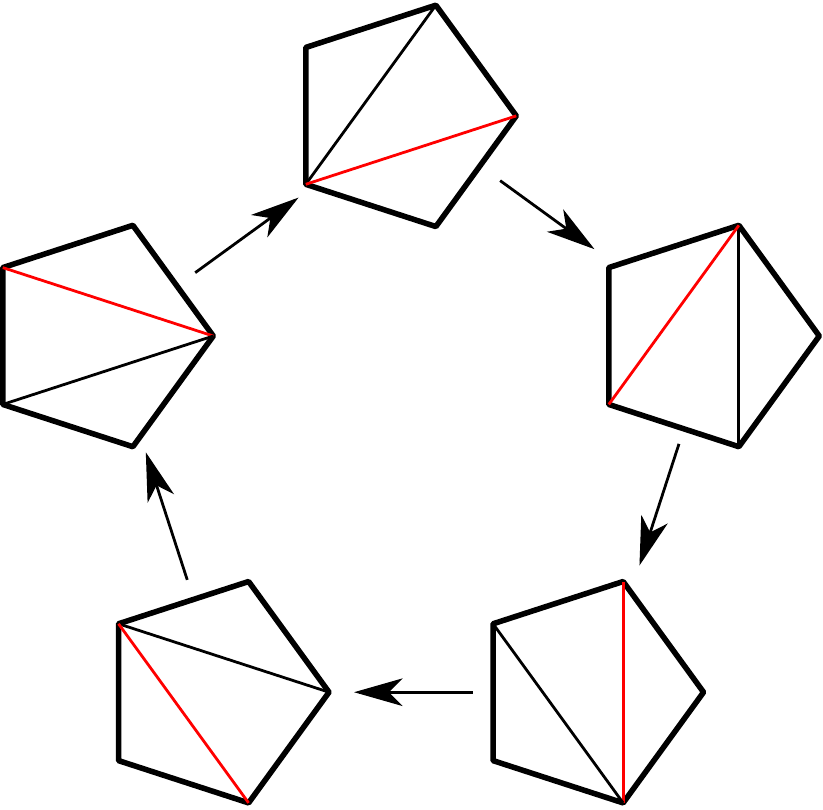}
  \caption{A sequence of mutations for five points.  At each step the side colored in red gets flipped.}
  \label{fig:triangulation2}
\end{figure}

It is not hard to show that the effect of a mutation on the diagonal labeled by $k$ on the cross-ratios $X_{j}$ corresponding to the other diagonals is given by
\begin{equation}
  \label{eq:mutation-CP1}
  X_{i}' = \mu_{k} X_{i} =
  \begin{cases}
    X_{k}^{-1}, &\quad i = k,\\
    X_{i} (1+ X_{k}^{\sgn(b_{ik})})^{b_{ik}}, &\quad i \neq k.
  \end{cases}
\end{equation}  This is the same as eq.~\eqref{eq:x-coords-mutation}.  As we noted before, flipping the diagonal produces the inverse of the initial cross-ratio.  The mutation in eq.~\eqref{eq:mutation-CP1} reproduces this.  Also, the cross-ratios corresponding to diagonals non-adjacent to the diagonal being flipped remain unchanged.  This is obvious since in this case $b_{ik} = 0$.

Now let us specialize to the case $n=6$ of interest.
A hexagon admits 14 distinct complete triangulations, each having three
diagonals.  These correspond to the 14 different clusters, each with
three $\mathcal{X}$ coordinates.
Out of these 42 are 15 distinct coordinates (as mentioned before
we never count both $x$ and $1/x$ separately).
Nine of these ratios were already displayed in eq.~(\ref{eq:nineratios});
the remaining six have not been given a name in previous literature
on $n=6$ scattering amplitudes since they do not appear in the two-loop
MHV amplitude.  For completeness let us list here all 15 $\mathcal{X}$-coordinates
\begin{alignat}{3}
v_1 &= r(3,5,6,2), &\qquad v_2 &= r(1,3,4,6), &\qquad v_3 &= r(5,1,2,4),\nonumber \\
x^+_1 &= r(2,3,4,1), &\qquad x^+_2 &= r(6,1,2,5), &\qquad x^+_3 &= r(4,5,6,3),\nonumber \\
x^-_1 &= r(1,4,5,6), &\qquad x^-_2 &= r(5,2,3,4), &\qquad x^-_3 &= r(3,6,1,2), \\
e_1 &= r(1,2,3,5), &\qquad e_2 &= r(2,3,4,6), &\qquad e_3 &= r(3,4,5,1), \nonumber\\
e_4 &= r(4,5,6,2), &\qquad e_5 &= r(5,6,1,3), &\qquad e_6 &= r(6,1,2,4),\nonumber
\end{alignat}
in terms of the $\mathbb{P}^1$ cross-ratio defined
in eq.~(\ref{eq:crossratio}).

Out of the 45 possible cross-ratios of the form $r(i,j,k,l)$,
the 15 $\mathcal{X}$-coordinates are special in that
they are precisely those in which the points
$i,j,k,l$ come in cyclic order.  The three most well-known
cross-ratios which are \emph{not} cluster $\mathcal{X}$-coordinates are
the ones known in the literature as
\begin{equation}
u_1 = -r(3,6,5,2), \qquad
u_2 = -r(1,4,3,6), \qquad
u_3 = -r(5,2,1,4).
\end{equation}
These are related to cluster $\mathcal{X}$-coordinates by $u_i = 1/(1 + v_i)$.

\subsection{The generalized Stasheff polytope for \texorpdfstring{$\Gr(2,6)$}{G(2,6)}}

Let us now discuss the geometry of the Stasheff polytope for the $A_3$ cluster
algebra detailed in the previous section.
In this case each cluster is in correspondence with a $2$-simplex, or a triangle.  There are $14$ clusters, to each of which corresponds a triangle.  We can label each triangle by the three $\mathcal{A}$-coordinates which appear on its vertices:
\begin{equation}
\begin{aligned}
 &\ket{13}, \ket{14}, \ket{15}, \qquad 
  \ket{14}, \ket{15}, \ket{24}, \qquad 
  \ket{13}, \ket{15}, \ket{35}, \qquad 
  \ket{13}, \ket{14}, \ket{46}, \\ 
 &\ket{15}, \ket{24}, \ket{25}, \qquad 
  \ket{14}, \ket{24}, \ket{46}, \qquad 
  \ket{15}, \ket{25}, \ket{35}, \qquad 
  \ket{13}, \ket{35}, \ket{36}, \\ 
 &\ket{13}, \ket{36}, \ket{46}, \qquad 
  \ket{24}, \ket{25}, \ket{26}, \qquad 
  \ket{24}, \ket{26}, \ket{46}, \qquad 
  \ket{25}, \ket{26}, \ket{35}, \\ 
 &\ket{26}, \ket{35}, \ket{36}, \qquad 
  \ket{26}, \ket{36}, \ket{46}. 
\end{aligned}
\end{equation}
These triangles fit together in a polytope with $14$ triangular faces, shown in fig.~\ref{fig:G26polytope}.  The polytope has $9$ vertices given by the non-frozen $\mathcal{A}$-coordinates $\ket{13}$, $\ket{14}$, $\ket{15}$, $\ket{24}$, $\ket{25}$, $\ket{26}$, $\ket{35}$, $\ket{36}$ and $\ket{46}$, and 21 edges.

\begin{figure}
  \centering
  \includegraphics[width=10.0truecm]{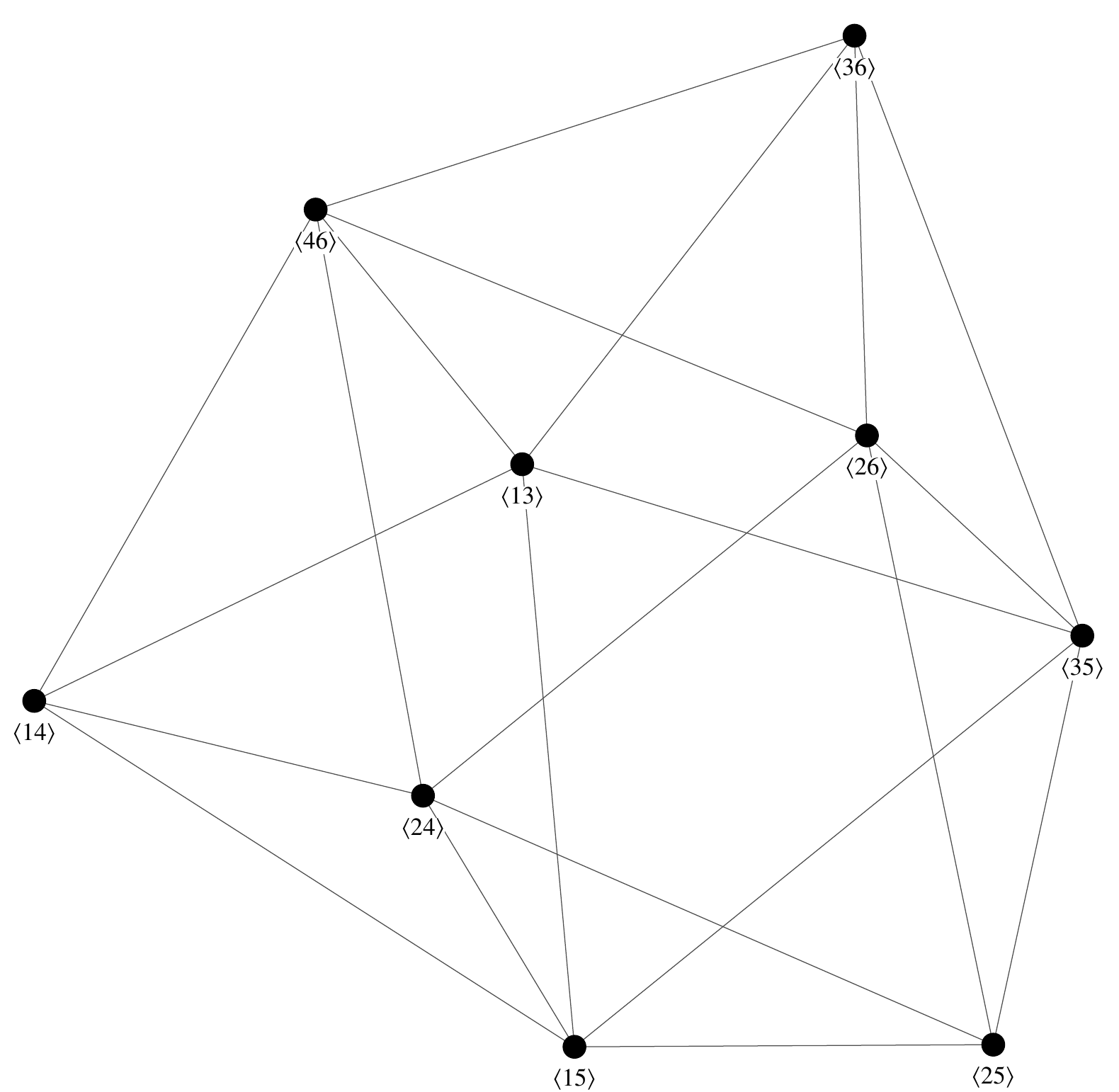}
  \caption{The polytope obtained by gluing together the triangles associated to clusters of the $\Gr(2,6)$ (i.e., $A_3$) cluster algebra.}
  \label{fig:G26polytope}
\end{figure}

All faces are triangles, but there are two different types of vertices: $\ket{14}$, $\ket{25}$ and $\ket{36}$ have valence four (they are incident with four edges) while the other six vertices have valence five.  The polytope has the topology of a sphere as can be confirmed by computing the Euler characteristic $\chi = V-E+F = 9 - 21 + 14 = 2$.

\begin{figure}
  \centering
\begin{subfigure}{.45 \textwidth}
  \includegraphics{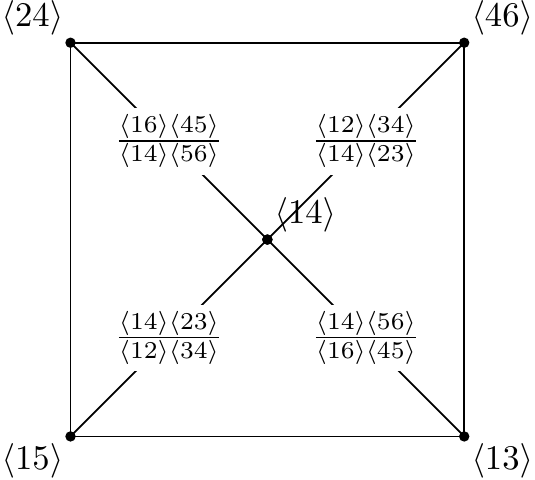}
\caption{}
  \label{fig:G26-4vertex}
\end{subfigure}
\begin{subfigure}{.45 \textwidth}
  \includegraphics{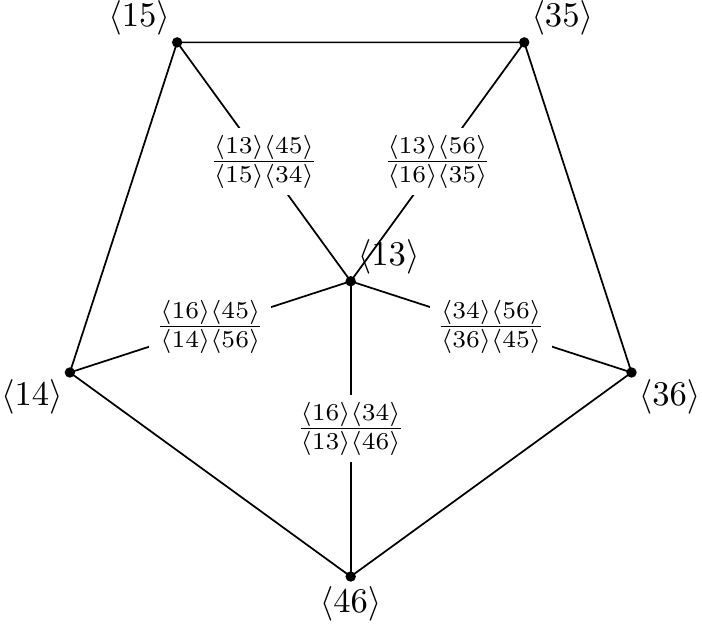}
\caption{}
  \label{fig:G26-5vertex}
\end{subfigure}
  \caption{The cross-ratios ($\mathcal{X}$-coordinates) around a valence 4 vertex (a) and a valence 5 vertex (b) of the polytope (fig.~\ref{fig:G26polytope}) associated to the $\Gr(2,6)$ cluster algebra. For clarity we have omitted the crucial overall minus sign in front of each $\mathcal{X}$-coordinate.}
\end{figure}

Now recall that to each edge of the polytope we can associate a pair
consisting of an $\mathcal{X}$-coordinate and its inverse. Let us take a
closer look at the $\mathcal{X}$-coordinates corresponding to the edges
incident on the two kinds
of vertices.
In order for the association between $\mathcal{X}$-coordinates and edges to be one-to-one, we need to pick an orientation.  Consider for example the valence 4 vertex shown in fig.~(\ref{fig:G26-4vertex}).  As we go around it we encounter the cross-ratios
\begin{equation}
  \label{eq:4vertex}
 \frac {\ket{14}\ket{23}}{\ket{12}\ket{34}}, \qquad
 \frac {\ket{14}\ket{56}}{\ket{16}\ket{45}}, \qquad
 \frac {\ket{12}\ket{34}}{\ket{14}\ket{23}}, \qquad
 \frac {\ket{16}\ket{45}}{\ket{14}\ket{56}}.
\end{equation}
The third cross-ratio is an inverse of the first while the fourth is an inverse of the second.  Therefore, the cluster coordinates are the same as for the $A_{1} \times A_{1}$ cluster algebra.
This is the dual of the statement shown in eq.~(\ref{eq:square}).

On the other hand, if we consider for example the valence 5 vertex shown
in fig.~\ref{fig:G26-5vertex}, the corresponding
list of cross-ratios is
\begin{equation}
  \label{eq:5vertex}
 \frac {\ket{13}\ket{45}}{\ket{15}\ket{34}}, \qquad
 \frac {\ket{13}\ket{56}}{\ket{16}\ket{35}}, \qquad
 \frac {\ket{34}\ket{56}}{\ket{36}\ket{45}}, \qquad
 \frac {\ket{16}\ket{34}}{\ket{13}\ket{46}}, \qquad
 \frac {\ket{16}\ket{45}}{\ket{14}\ket{56}}.
\end{equation}
These are the $\mathcal{X}$-coordinates of an $A_{2}$ cluster algebra.  It can be checked that these are precisely minus the arguments of dilogarithms in the five-term dilogarithm identity~(\ref{eq:fiveterm}).
This is the dual of the statement shown in eq.~(\ref{eq:pentagon}).

\begin{figure}
  \centering
  \includegraphics[width=15.0truecm]{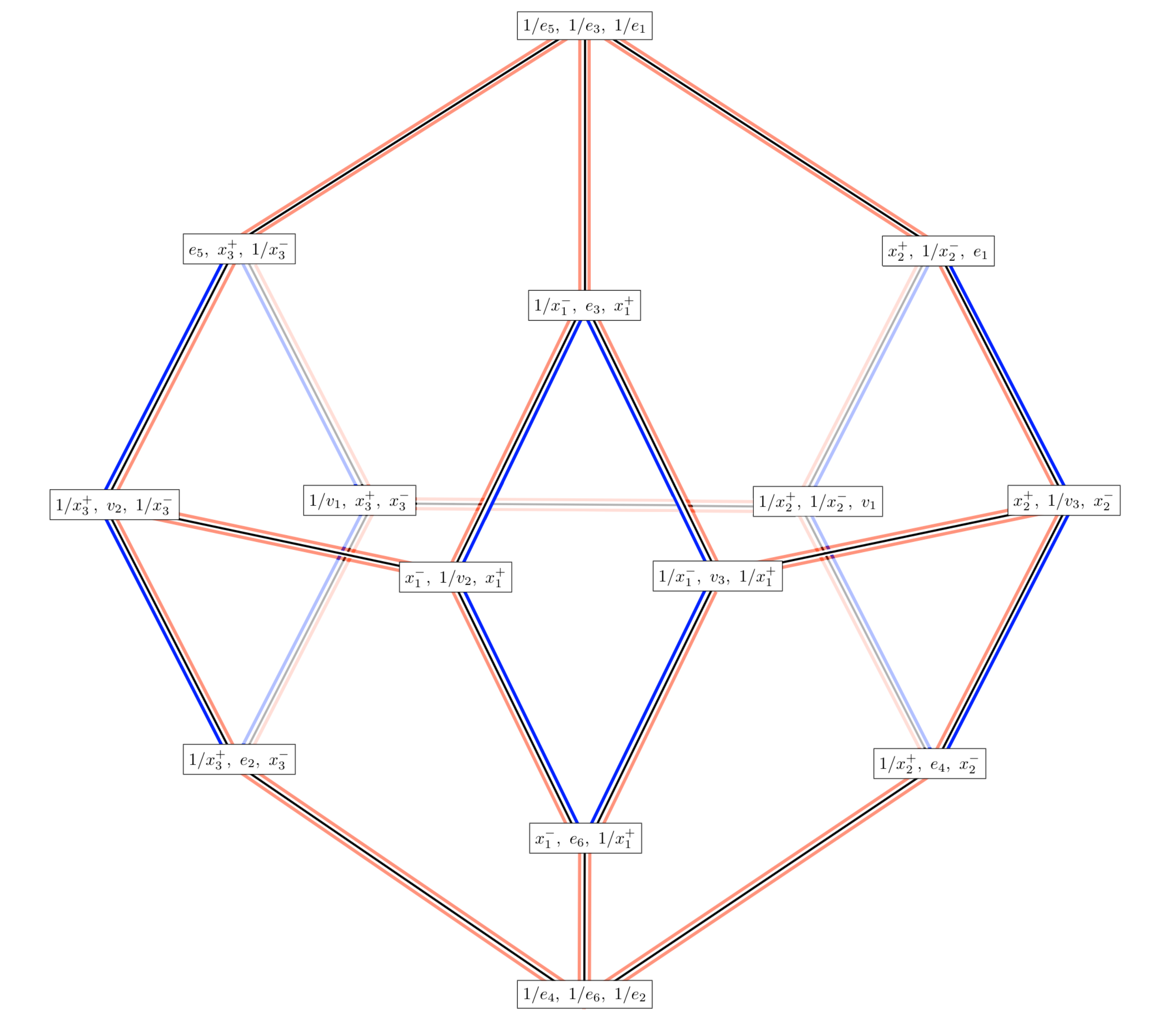}
\caption{The Stasheff polytope for the $\Gr(2,6)$ (or $A_3$) cluster algebra,
also known as the $K_5$ associahedron.
Each of the 14 vertices (clusters)
is labeled by its three $\mathcal{X}$-coordinates.
The 21 edges connect pairs of
clusters which are related to each other by some mutation.
The 9 faces comprise three quadrilaterals (shown in blue) and
six pentagons (shown in red).  These correspond respectively
to $A_1 \times A_1$ and $A_2$ subalgebras.}
  \label{fig:G26-associahedron}
\end{figure}

The dual polytope, shown in fig.~\ref{fig:G26-associahedron}, has $14$ vertices and $9$ faces, three of which are quadrilaterals and six of which are pentagons.  This is the Stasheff polytope or the $K_5$ associahedron.  The name associahedron comes from the following construction: consider $n$ (in the case of $K_5$ we take $n=5$) non-commutative variables and all the ways of inserting parentheses.  For example, we have $((ab)(c d))e$, $(((ab)c)d)e$, etc.  In total there are $C_{n-1}$ ways of parenthesizing $n$ variables, where $C_n$ is the $n$th Catalan number.  Then, join together two such expressions if one can be obtained from the other by applying the associativity rule once.  By joining all these expressions, we build up the Stasheff polytope.

\subsection{Cluster coordinates for \texorpdfstring{$\Gr(3,7)$}{G(3,7)}}

Beginning with the initial quiver for $\Gr(3,7)$, we can similarly
generate all of the clusters and their $\mathcal{A}$- and
$\mathcal{X}$-coordinates by successive mutations until all possibilities
are exhausted.
The $E_6$ algebra generated in this manner has a total of 49
well-known $\mathcal{A}$-coordinates (see for example~\cite{1088.22009}), composed of the 35 Pl\"ucker coordinates
$\ket{i j k}$ on $\Gr(3,7)$ together with 14 composite brackets of the form
\begin{equation}
\ket{1 \times 2, 3 \times 4, 5 \times 6}, \qquad
\ket{1 \times 2, 3 \times 4, 5 \times 7}
\end{equation}
and their cyclic images.
The seven coordinates $\ket{123}, \ldots, \ket{712}$ are frozen,
while the remaining forty-two are unfrozen.

Mutation generates 833 distinct clusters, which altogether contain a total
of 385 distinct $\mathcal{X}$-coordinates.  We list all of them here by
separating them into four classes, and use the notation
\begin{equation}
r(a\vert b,c,d,e) = \frac{\ket{a b c} \ket{a d e}}{\ket{a c d} \ket{a b e}}
\end{equation}
as well as the $\mathbb{P}^2$ cross-ratio defined in
eq.~(\ref{eq:crossratio6}).

First of all there are $3 \times 7 = 21$ coordinates of the form
\begin{equation}
r(2\vert 1,3,5,6), \quad
 \frac{\langle 2 3 1\rangle  \langle 4 5 6\rangle }{\langle 4\times 5, 6\times 1 ,2\times 3\rangle }, \quad
\frac{\langle 1 2 7\rangle  \langle 2 3 4\rangle  \langle 3 4 5\rangle  \langle 5 6 7\rangle }{\langle 2 5 7\rangle  \langle 3 4 7\rangle  \langle 1\times 2 ,3\times 4 ,5\times 6\rangle }
\end{equation}
together with their cyclic images.
Each of these cross-ratios is real (that is, equal to its parity conjugate---see appendix~\ref{sec:parity-conj-twist}),
and it suffices to take only their cyclic images since a dihedral
transformation (i.e., $i \to 8 - i$) maps this set back to itself.

Secondly there are a further $2 \times 14 = 28$ real cross-ratios of the form
\begin{equation}
r_3(1,2,5,6,3,4),
\quad
 \frac{\langle 1 2 7\rangle  \langle 2 5 6\rangle  \langle 3 4 5\rangle }{\langle 2 5 7\rangle  \langle 1\times 2, 3\times 4, 5\times 6\rangle },
\end{equation}
together with their dihedral images.

Next come $6 \times 2 \times 7 = 84$ complex (that is, not equal
to their parity conjugates) cross-ratios of the form
\begin{multline}
r(2\vert 1,3,4,5), \quad
r(1\vert 6,3,4,5), \quad
r(3\vert 2,4,5,1), \\
r_{3}(1,5,3,2,6,4), \quad
r_{3}(1,4,6,2,3,5), \quad
\frac{\langle 1 2 7\rangle  \langle 2 3 4\rangle  \langle 5 6 7\rangle }{\langle 2 6 7\rangle  \langle 3\times 4, 5\times 7 ,1\times 2\rangle},
\end{multline}
together with their parity conjugates and all cyclic images thereof.

Finally we have the $9 \times 2 \times 14 = 252$ complex cross-ratios
\begin{multline}
r(1\vert 5,2,3,4), \quad
r(1\vert 6,2,3,4), \quad
r(1\vert 6,2,3,5), \quad
r(1\vert 6,2,4,5), \quad
r(2\vert 1,3,4,6),\\
r_{3}(1,2,4,6,3,5), \quad
r_{3}(1,4,3,6,5,2), \quad
r_{3}(1,3,5,6,2,4), \quad
\frac{\langle 2 6 1\rangle  \langle 3 4 5\rangle }{\langle 4\times 5 ,6\times 1, 2\times 3\rangle },
\end{multline}
together with their parity conjugates and all dihedral images thereof.

Note that we have expressed most of the cross-ratios above in a form
in which they do not depend explicitly on point number 7.
The three most complicated cross-ratios are exceptions to this, and
for these three we find it worthwhile to display here the
non-trivial factorizations
\begin{equation}
\begin{aligned}
1+\frac{\langle 1 2 7\rangle  \langle 2 5 6\rangle  \langle 3 4 5\rangle }{\langle 2 5 7\rangle  \langle 1\times 2, 3\times 4, 5\times 6\rangle } &=
\frac{\ket{125} \ket{7 \times 2, 3 \times 4, 5 \times 6}}{\ket{257} \ket{1 \times 2, 3 \times 4, 5 \times 6}}, \\
1+\frac{\langle 1 2 7\rangle  \langle 2 3 4\rangle  \langle 5 6 7\rangle }{\langle 2 6 7\rangle  \langle 3\times 4, 5\times 7 ,1\times 2\rangle} &=
\frac{\ket{257} \ket{1 \times 2, 3 \times 4, 6 \times 7}}{\ket{267} \ket{1 \times 2, 3 \times 4, 5 \times 7}}, \\
1+\frac{\langle 1 2 7\rangle  \langle 2 3 4\rangle  \langle 3 4 5\rangle  \langle 5 6 7\rangle }{\langle 2 5 7\rangle  \langle 3 4 7\rangle  \langle 1\times 2 ,3\times 4 ,5\times 6\rangle} &=
\frac{\ket{1 \times 2, 3 \times 4, 5 \times 7} \ket{7 \times 2, 3 \times 4, 5 \times 6}}{\ket{257} \ket{347} \ket{1 \times 2, 3 \times 4, 5 \times 6}}.
\end{aligned}
\end{equation}
All three of the cluster $\mathcal{X}$-coordinates on the left-hand side of this equation appear in the
$\B_3 \otimes \mathbb{C}^*$ part of the coproduct of the two-loop $n=7$ MHV amplitude.

\subsection{Structure of the motivic two-loop \texorpdfstring{$n=7$}{n=7} MHV amplitude}

Obviously it is impractical for us to display the
Stasheff polytope for the $\Gr(3,7)$ cluster algebra, with its 833 vertices, 2499 edges, and
2856 two-dimensional faces (of which 1785 are quadrilaterals and the
other 1071 are pentagons).
However, we are in a position now to carry out a `motivic analysis'
of the two-loop $n=7$ MHV amplitude using the information contained in the
previous section.

First of all we note the amazing fact that all of the entries
of $\{z\}_2$ and $\{z\}_{3}$ in the results in sec.~\ref{sec:mc} are always
cluster
$\mathcal{X}$-coordinates of $\Conf_7(\mathbb{P}^2)$.
Interestingly, of the 385 such coordinates available, only 231 of them
actually appear in the $n=7$ MHV amplitude at two loops.
This might be a two-loop accident, but if it continues to hold at higher
loop order it would be important to find some sort of geometric
explanation.

Turning our attention now to the expression for the $\Lambda^2 \B_2$
part of the coproduct shown in eq.~(\ref{eq:seven-pt-b2wb2}), we note first of all
the further highly nontrivial fact that for each term $\{x_1 \}_2
\wedge \{ x_2 \}_2$, there is always at least one of the 833 clusters
which contains both $x_1$ and $x_2$.  And more spectacularly, the
variables always appear in pairs with Poisson bracket $\{ x_1,  x_2 \} = 0$.
Now we understand the geometric meaning of the ambiguity mentioned
in eq.~(\ref{eq:squareidentity}), in light
of eq.~(\ref{eq:square})---it is exactly the ambiguity of trying
to chose one of the four
vertices of a quadrilateral, when there is no reason at all to
have to make a choice:  each term in $\Lambda^2 \B_2$ corresponds
naturally to a certain quadrilateral face.

We conclude that the most canonical, invariant
way of expressing the $\Lambda^2 \B_2$
part of the coproduct of the two-loop $n=7$ MHV amplitude is not by the
formula~(\ref{eq:seven-pt-b2wb2}), but by writing it as a sum of
42 quadrilateral faces of the $E_6$ Stasheff polytope.  It is obviously
of paramount importance to understand what makes these 42 special,
out of the 1785 such faces available.

An analysis of the $\B_3 \otimes \mathbb{C}^*$ part of the coproduct
requires a classification of all of the possible $A_3$, $A_2 \times A_1$
and $A_1 \times A_1 \times A_1$ subalgebras of $E_6$.
The generalized Stasheff polytope of this algebra has 1547 three-dimensional faces, consisting of
357 cubes ($A_1 \times A_1 \times A_1$),
714 pentaprisms ($A_2 \times A_1$) and
476 of the $A_3$ polytopes shown in fig.~\ref{fig:G26-associahedron}.
We expect these to play a role in unlocking further structure
in the two-loop $n$-point MHV amplitudes, which we will explore in
future work.

\section{Conclusion}

Appropriately defined scattering amplitudes in
maximally supersymmetric
Yang-Mills
theory are functions on $\Conf_n(\mathbb{P}^3)$ which have a very
rich mathematical structure but do not, in general, appear to admit any particular
canonical or even preferred functional representation.
The one important exception is the two-loop MHV amplitude
for $n=6$ reviewed in section~(\ref{sec:GSVV}), which does have a
canonical form (up to trivial $\Li_m$ identities):  that in which
it is expressed only in terms of the classical polylogarithm functions,
with only (minus) cluster $\mathcal{X}$-coordinates as arguments.

More general amplitudes may be computed numerically if desired
(for example all two-loop MHV amplitudes may be evaluated with
reasonable efficiency~\cite{Anastasiou:2009kna}), but we do
not strive to find any particular explicit analytic formulas for them.
It often happens in physics and in mathematics that when one reaches
sufficiently deep into a subject, one realizes that the appropriate objects
of study are not what one originally thought, but some suitable generalization
or modification thereof.
In this vein we have proposed that our focus on the mathematical structure
of scattering amplitudes in SYM theory should fall on what we call
motivic amplitudes, and in particular on their coproducts, which capture
the `mathematically most complicated part' of an amplitude.

By drawing on our explicit results for the two-loop $n=6$ and $n=7$ MHV
motivic amplitudes,
we have shown that an important role is played by cluster coordinates,
which are preferred
sets of coordinates
on $\Conf_n(\mathbb{P}^3)$ with very rich mathematical structure.
Specifically, we conjecture based on the examples presented here, as well
as others that we have computed, that all coproduct components of
all two-loop MHV motivic amplitudes are expressible in terms of
Bloch group elements $\{x\}_k$ with only
${\mathcal{X}}$-coordinates $x$ appearing.
The algebras relevant for $n=6,7$ are of finite type, being
the $A_3$
and $E_6$ algebras respectively, while for $n>7$ the relevant algebras
are of infinite type, although only a finite subset of these variables
actually appear at two loops.

If one accepts that cluster $\mathcal{X}$-coordinates answer the
`kinematical' question \emph{which variables do motivic amplitudes
depend on?}, it is natural to turn attention next to the `dynamical' question
of exactly in which combinations they appear in amplitudes.
We have provided a first glimpse by showing
that the terms in $\Lambda^2 \B_2$ component of the
coproduct of the two-loop $n=7$ MHV amplitude---the component which
measures the obstruction to writing this amplitude in terms of the
classical polylogarithm functions only---are
in correspondence with quadrilateral faces
of the
relevant Stasheff polytope (i.e., with $A_1 \times A_1$ subalgebras
of the cluster algebra).  Again based on other examples which we have
analyzed, we conjecture that this statement remains true for all two-loop
MHV amplitudes.  However a great deal of structure remains to be understood.
In particular, only a very small number of all possible quadrilaterals actually
make an appearance in $\Lambda^2 \B_2$; what, if anything, is the special
geometric significance of these particular quadrilaterals?
What is the geometric
significance of the cluster $\mathcal{X}$-coordinates appearing in
$\B_3 \otimes \mathbb{C}^*$, or in non-MHV amplitudes, or at higher
loops?  A few of these questions will be addressed
in future work.

Many other interesting questions are also raised by our work.
For example, Dixon, Drummond and Henn have employed with great success
the strategy of studying the space of all integrable,
conformally invariant symbols
whose letters are drawn from the collection of available
$\mathcal{A}$-coordinates
at $n=6$.  By imposing all physical constraints available to
them at the time, they
were able to determine the symbol of the two-loop NMHV amplitude
exactly~\cite{Dixon:2011nj},
and that of the three-loop MHV amplitude up to two
parameters~\cite{Dixon:2011pw}
which were subsequently determined in~\cite{CaronHuot:2011kk}.
We have proposed that only $\mathcal{X}$-coordinates can appear in the
coproduct of MHV amplitudes, which is a stronger condition than
that only $\mathcal{A}$-coordinates can appear in their symbols.
The functions $\Li_m(1 + x)$ and $\Li_m(1 + 1/x)$
for any $\mathcal{X}$-coordinate $x$ for example satisfy the latter
but not the former.  Hence in particular we expect to see neither $\Li_3(u_i)$
nor $\Li_3(1-u_i)$ in the coproduct of any two-loop MHV motivic amplitudes.
It would be very interesting to investigate
in detail how restrictive this condition is in the space of all integrable
symbols, in order to see whether our new `motivic' constraint could
aid future computations employing this strategy.

It is also important to point out that in the examples we've
looked at, only
a fraction of all available $\mathcal{X}$-coordinates actually make
an appearance.  For example in section~\ref{sec:GSVV} we emphasized
that only the 9 coordinates~(\ref{eq:nineratios}) enter the two-loop
$n=6$ amplitude, out of the 15 available.
Then in section~\ref{sec:mc} we found that only 231 of the 385
available $\mathcal{X}$-coordinates make an appearance in the two-loop
$n=7$ motivic amplitude.
We do not yet have an understanding of the criterion which selects
these particular subsets of $\mathcal{X}$-coordinates, nor whether
this phenomenon is just an accident at two loops or continues to hold
at higher loop order.
If it does, this obviously constrains the space of possible motivic amplitudes
even more strongly than just the $\mathcal{X}$-coordinate condition
discussed in the previous paragraph.
We cannot help but note with amusement the fact that $9/15 = 231/385$, but we
certainly have too little data to speculate on whether or not
this is just a coincidence.

A number of interesting questions about the connection
between motivic amplitudes,
cluster coordinates and other recent approaches to scattering amplitudes
can now be asked.
A fair amount of recent work has considered the behavior of amplitudes
in various restricted domains, such as two-dimensional
or multi-Regge kinematics (see for
example~\cite{DelDuca:2010zp,Heslop:2010kq,Heslop:2011hv,Goddard:2012cx,Ferro:2012wa}
or~\cite{Bartels:2011xy,Prygarin:2011gd,Bartels:2011ge,Lipatov:2012gk,Dixon:2012yy,Pennington:2012zj}, respectively),
where in either case considerable simplification occurs.
Also, it has long been appreciated that
the behavior of amplitudes under collinear (and especially multi-collinear)
limits strongly constrains their structure, and the operator product
expansion (OPE) to the null
polygonal Wilson
loop~\cite{Alday:2010ku,Gaiotto:2010fk,Gaiotto:2011dt,Basso:2013vsa}
aims to compute arbitrary amplitudes at finite coupling in a systematic
expansion around the collinear limit.
It would very nice to have a
thorough understanding of these limits and expansions directly
at the level of cluster algebras, or even individual quivers.

Finally, we have so far made no explicit reference to the integrability
of planar SYM theory (see the review~\cite{Beisert:2010jr}),
which clearly plays a crucial but so far not fully
exploited part in unlocking the structure of its amplitudes (approaches
other than the OPE mentioned above
include for example~\cite{Drummond:2009fd,Alday:2009dv,Alday:2010vh,Ferro:2012xw}).
We hope that motivic amplitudes and cluster coordinates will be found
to be useful in these and other endeavors, just as
general `motivic' methods based on the symbol calculus of polylogarithm
functions are finding ever wider applications to physical computations
in quantum field theory,
including
Feynman integrals,
amplitudes, form factors, correlation functions, and
Wilson loops~\cite{DelDuca:2011wh,Duhr:2011zq,Bullimore:2011kg,Brandhuber:2012vm,Bogner:2012dn,Lipstein:2012vs,Naculich:2013xa,Drummond:2013nda},
not just in SYM theory
but even QCD~\cite{Duhr:2012fh,vonManteuffel:2012je,Gehrmann:2013vga,Anastasiou:2013srw,Henn:2013pwa}
and string theory~\cite{Schlotterer:2012ny,Drummond:2013vz,Broedel:2013tta,Broedel:2013aza} as well.

\acknowledgments

Various subsets of us have benefited
from stimulating discussions with Nima Arkani-Hamed,
Lance Dixon, James Drummond and David Skinner, and are grateful
to Andy Neitzke and
Yang-Hui He for very illuminating conversations on cluster algebras.
MS, CV and AV appreciate the generous support of the Kavli
Institute for Theoretical Physics during the initial stages of this work,
and CV in addition
acknowledges the hospitality of BIRS, ECT* and CERN during its course.
This work was supported by the US Department of Energy under contracts
DE-FG02-91ER40688 (JG, MS) and DE-FG02-11ER41742 (AV Early Career Award),
the Simons Fellowship in Theoretical Physics (AV),
and the
Sloan Research Foundation (AV).  The work of AG is supported by the NSF grants DMS-1059129  and DMS-1301776.

\appendix

\section{Parity Conjugation on \texorpdfstring{$\Conf_n(\mathbb{P}^{k-1})$}{Conf P k-1}}
\label{sec:parity-conj-twist}

\subsection{Positive configurations}
\label{app:positive}

Given a volume form $\omega$ in $V_k$, which is not a part of our data,
we can assign to a configuration of $k$ vectors $v_1, \dotsc, v_k$ a number:
\begin{equation}
\langle 1, \dotsc, k\rangle:= \langle v_1, \dotsc, v_k\rangle:= \omega(v_1, \dotsc, v_{k}).
\end{equation}
Given an orientation of a real vector space $V_k$, we can define \emph{positive configurations of vectors in $V_k$}.
Namely, choose a volume form $\omega$ compatible with the orientation of the space,
i.e.\ $\omega(v_1, \dotsc, v_{k})>0$ if $(v_1, \ldots , v_{k})$ is a positively oriented basis.
Then  a configuration $(v_1, \dotsc, v_n)$ is \emph{positive} if
$\langle v_{i_1}, \dotsc, v_{i_k}\rangle >0$ for any $i_1 < \cdots < i_k$. We denote by
$\Conf^+_n(k)$ the set of positive configurations.

There is a twisted cyclic shift map, obtained by moving the last vector to the front, and multiplying it by
$(-1)^{k-1}$.
It preserves positive configurations of vectors:
\begin{equation}
c: \Conf^+_n(k) \lra \Conf^+_n(k), \qquad (v_1, \dotsc, v_n) \lms ((-1)^{k-1}v_n, v_1, \dotsc, v_{n-1}).
\end{equation}

The subspace of \emph{positive configurations of points} $\Conf^+_n(\R\PP^{k-1})$
 is the image of the restriction of the projection map $\pi: \Conf_n(k) \lra \Conf_n(\PP^{k-1})$
to positive configurations of vectors:
\begin{equation}
\Conf^+_n(\R\PP^{k-1}) := \pi(\Conf^+_n(k)).
\end{equation}

\subsection{Parity conjugation} 
Let $P_n$ be an oriented convex $n$-gon. Let us denote by $\Conf_{P_n}(\R\PP^{k-1})$ the 
space of configurations of points
of $\PP^{k-1}$ parametrized 
by the set of vertices
of the polygon $P_n$.  
It is the space of orbits of the diagonal action of the group 
$PGL_k$ on collections of  points parametrized 
by the vertices. 

We emphasize that the points are parametrized by the vertices of the polygon, but there is no 
special parametrization of the vertices. If we choose an initial vertex $v$ of the polygon, 
then there is an isomorphism
\begin{equation}
i_v: \Conf_{P_n}(\R\PP^{k-1}) \lra \Conf_n(\PP^{k-1})
\end{equation}
defined by parameterizing the vertices by the set $\{1, \ldots, n\}$, starting from the vertex $v$
to which we assign $1$, and going
according to the orientation of the polygon.

The parity conjugation is a rational map, i.e.\ a map defined for generic configurations,
\begin{equation}\label{paritycona}
\ast: \Conf_{P_n}(\R\PP^{k-1}) \lra \Conf_{\ast P_n}(\R\PP^{k-1}),
\end{equation}
 described as follows.
A collection of points $\{x_{v}\}$ in $\PP^{k-1}$ parametrized by
the set of vertices $\{v\}$ of the polygon $P_n$ gives rise to a collection of hyperplanes
$\{H_{v}\}$ in $\PP^{k-1}$ parametrized by the same set.
Namely, let $\{x_1(v), \dotsc, x_{k-1}(v)\}$ be the points parametrized by the
  $(k-1)$ vertices of the polygon obtained by starting at the vertex $v$ and going around the polygon
following the orientation. So $x_1(v) = x_v$ and so on. We define the hyperplane $H_v$ to be
the span of these points:
\begin{equation}
H_v := \langle x_1(v), \dotsc, x_{k-1}(v)\rangle.
\end{equation}
Viewing the hyperplanes $\{H_{v}\}$ as points of the dual projective space, and reversing the orientation
of the polygon,
we get a point of $\Conf_{\ast P_n}(\R\PP^{k-1})$, which is the result of the parity conjugation
applied to the original configuration of points.

\begin{lemma}
The map (\ref{paritycona}) is a perfect duality: $\ast^2 = \Id$.
\end{lemma}

\begin{proof} To calculate the map $\ast^2$ we need to find the intersection
of the hyperplanes $H_w$ corresponding to the vertex $v$ and $k-2$ vertices preceding
it in the orientation of $P_n$. All of them, by the very definition,
contain the point $x_v$ parametrized by the vertex $v$.
\end{proof}

We emphasize that the map (\ref{paritycona}) is not a map of a space to itself, since
there is no invariant way to identify the left and right spaces in (\ref{paritycona}).

\subsection{Parity conjugation on configurations of vectors}
Similarly,
let us
denote by $\Conf_{P_n}(k)$ the
space of configurations of vectors  in a $k$-dimensional vector space,   parametrized
by the vertices of the polygon $P_n$.
Let us upgrade the projective parity conjugation to a  parity conjugation on configurations
vectors, given by a rational map
\begin{equation} \label{parityconav}
\ast: \Conf_{P_n}(k) \lra \Conf_{\ast P_n}(k).
\end{equation}

Having in mind computing the parity conjugation, we define it now by
breaking the symmetry, i.e.\ using the isomorphism
$i_v: \Conf_{P_n}(k) \to \Conf_{n}(k)$ determined by a choice of a specific vertex $v$.
So we use the polygon vertex $v$ to order the vectors of a configuration by
$(l_1, \dotsc, l_n)$. Let us define  a configuration of
covectors $(g_1, \dotsc, g_n)$ by setting
\begin{equation}
g_1(\bullet):=
\frac{\omega(l_1, \dotsc, l_{k-1}, \bullet)}{\omega(l_1, \dotsc, l_{k-1}, l_k)},
\end{equation}
and $g_i$ obtained by the twisted cyclic shift by $i-1$ of this formula.
The covectors $g_i$ evidently do not depend on the choice of the form $\omega$.

We define the parity conjugation $\ast$ as an automorphism  of the space $\Conf_n(k)$:
\begin{equation}
\ast: \Conf_n(k)\lra \Conf_n(k), \qquad \ast(l_1, \ldots, l_{n}):= (g_1, \ldots, g_{n}).
\end{equation}
Abusing notation, we use the same notation $\ast$.
Let us stress again that, although $g_i$'s are vectors of the dual space, the configurations of vectors in a space and in its dual are canonically identified.

\begin{proposition}
\label{6.14.11.1} The map
\begin{equation}
\ast: (l_1, \ldots , l_n) \lms (g_1, \ldots , g_n)
\end{equation}
is a duality.
\end{proposition}

For computational purposes, we define a non-normalized version of the parity map
\begin{equation}
f_1(\bullet):=
\omega(l_1, \dotsc, l_{k-1}, \bullet),
\end{equation}
and $f_i$ is obtained by the twisted cyclic shift by $i-1$.

A volume form $\omega$ in $V_k$ defines the dual volume form $\ast\omega$ on $V_k^*$.
We set
\begin{equation}
[1, \dotsc, k]:= [f_1, \dotsc, f_k]:= \ast\omega(f_1, \dotsc, f_{k}).
\end{equation}

\begin{lemma}
\label{invo}
One has
\begin{equation}
\label{form*}
[f_{1}, f_{2}, \ldots , f_{k}]=  \langle 1,2,\dotsc,k\rangle \langle 2,3,\dotsc,k+1\rangle \ldots
\langle k-1,k,\dotsc,2k-2\rangle.
\end{equation}
\end{lemma}
For example, for $k=4$ we have
\begin{equation}
[f_{1}, f_{2}, f_{3}, f_{4}]=  \langle 1,2,3, 4\rangle \langle 2,3,4, 5\rangle \langle 3,4, 5, 6\rangle.
\end{equation}

\begin{proof} The left-hand side in~(\ref{form*}) vanishes if any of the brackets on the right vanishes.
For example, if $\langle 1,2,\dotsc,k\rangle =0$, then $f_1$ and $f_2$ are proportional, and thus
$[f_{1}, \ldots, f_{k}]=0$.
Therefore the left hand side is divisible by their product. Since they are polynomials of the same degree,
the claim follows up to a constant. The constant  must be $\pm 1$:
indeed, it  is a rational number;
if it is divisible by a prime $p$, reducing mod $p$ we get a contradiction.
Notice that the two expressions in~(\ref{form*}) scale the same way under the rescaling
$v_i \lms \lambda_i v_i$. Rescaling $\omega \lms \lambda \omega$
we rescale $f_{\ast}$ by $\lambda$ and $[\ast]$ by $\lambda^{-1}$.
Thus both sides scale by $\lambda^{k-1}$.
\end{proof}

\begin{proof}[Proof of the Proposition.]
The map $\ast^2$, being projected to configurations of points,
 is the identity map. Here it is crucial that the change
of the orientation of the polygon $P_n$ is built in its definition.
 Lemma~\ref{invo} tells
\begin{equation}
\label{form*a}
[g_{1}, g_{2}, \ldots , g_{k}]=  \langle k,\dotsc,2k-1\rangle^{-1}.
\end{equation}
Therefore $\langle \ast g_{1}, \ldots , \ast g_{k}\rangle = \langle 1,\dotsc,k\rangle$.
This implies the proof for odd $n$, and thus, by employing a trick, for all $n$.
\end{proof}

\subsection{Parity conjugation for the projective plane}

It is convenient to use
the notation $f_{i j}$ for the functional defined as
$f_{i j}(v):= \omega(v_i, v_j, v)$.
One has
\begin{equation}
\label{a}
[f_{12}, f_{23}, f_{34}]=  \langle 1,2,3\rangle\langle 2,3,4\rangle.
\end{equation}
\begin{equation}
\label{b}
[f_{12}, f_{23}, f_{45}]=  \langle 1,2,3\rangle\langle 2,4, 5\rangle.
\end{equation}
To prove the second identity, notice that if any of the two factors becomes zero,
then the left hand side is zero. For example, if $v_5$ is a linear combination of $v_2$ and $v_4$, then
the left hand side is proportional to $[f_{12}, f_{23}, f_{24}]=0$.

\begin{figure}
\centering
\includegraphics{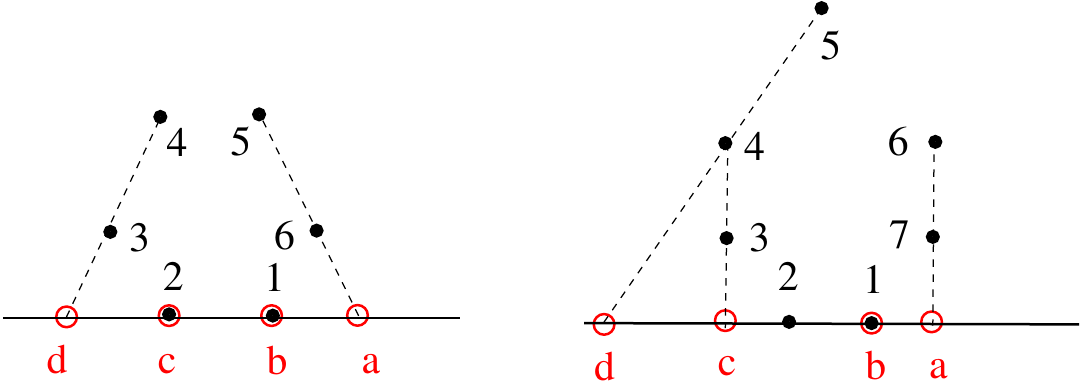}
\caption{$r(a,b, c,d) = \frac{\langle 234\rangle\langle
    156\rangle}{\langle 1\times 2, 3\times4, 5\times 6\rangle}$ (left), and
$r(a,b,c, d) = \frac{\langle 123\rangle\langle 345\rangle\langle
  671\rangle}{\langle 134\rangle\langle 1\times2, 4\times5, 6\times 7\rangle}$ (right).}
\label{sa2l1}
\end{figure}

Using this, we easily calculate some examples of the
parity involution action for configurations of points in $\PP^2$:
\begin{align}
\label{sa11}
r(1\vert 2,3,4,5) &=  \frac{\langle 123\rangle\langle
  145\rangle}{\langle 134\rangle\langle 125\rangle} \lms r_3(2,5,4; 1,6,3) =
 -\frac{\langle 251\rangle\langle 546\rangle\langle 423\rangle}{\langle
   256\rangle\langle 543\rangle\langle 421\rangle}, \\
\label{sa12}
r(1\vert 2,3,5, 7) &=   \frac{\langle 123\rangle\langle
  157\rangle}{\langle 135\rangle\langle 127\rangle} \lms
 \frac{\langle 234\rangle\langle 156\rangle}{\langle 1\times 2, 3\times4, 5\times 6\rangle}, \\
\label{sa13}
r(1\vert 3,4,6,7) &=  \frac{\langle 134\rangle\langle 167\rangle} {\langle
  146\rangle\langle 137\rangle} \lms \frac{\langle 124\rangle\langle
  345\rangle\langle 671\rangle}{\langle 134\rangle\langle 1\times2, 4\times5, 6\times 7\rangle}, \\
\label{sa14}
r(1\vert 3,4,5, 6) &=  \frac{\langle 134\rangle\langle
  156\rangle}{\langle 145\rangle\langle 136\rangle} \lms
\frac{\langle 345\rangle\langle 124\rangle\langle 567\rangle\langle
  126\rangle}{\langle 456\rangle\langle 125\rangle\langle 1\times2, 3\times4, 6\times 7\rangle}, \\
\label{sa15}
 r_3(1,2,4;7,3,6) &=-\frac{\langle 127\rangle\langle
  243\rangle\langle 416\rangle}{\langle 123\rangle\langle
  246\rangle\langle 417\rangle}
\lms \frac{\langle 345\rangle \langle 1\times2, 4\times5, 6\times
  7\rangle}{\langle 145\rangle\langle 2\times3, 4\times5, 6\times 7\rangle}.
\end{align}
Let us give two examples of the proofs of these formulas.

1. The
formula illustrated in fig.~\ref{sa2l3} was proved in~\cite{G91a}.
Here is a different proof:
\begin{equation}
(f_{14}\vert f_{23}, f_{25}, f_{61}, f_{36}) =
\frac{[f_{14}, f_{23}, f_{25}][f_{14}, f_{61}, f_{36}]}{[f_{14},f_{25}, f_{61}][f_{14}, f_{23}, f_{36}]}
\stackrel{(\ref{a}), (\ref{b})}{=}
\end{equation}
\begin{equation}
\frac{\langle 325\rangle\langle 214\rangle\langle 361\rangle\langle
  614\rangle}{\langle 614\rangle\langle 125\rangle\langle
  236\rangle\langle 314\rangle} =
-\frac{\langle 124\rangle\langle 235\rangle\langle 316\rangle}{\langle
  125\rangle\langle 236\rangle\langle 314\rangle}.
\end{equation}
The cross-ratio $r(f_{14}\vert f_{23}, f_{25}, f_{61}, f_{36})$ can be calculated by viewing the
points $f_{i j}$ of the dual projective plane
as the lines, denoted by $L_{i j}$,  in the original projective plane. Then it is the cross-ratio of the
configuration of four points
obtained by intersecting the line $L_{14}$ with the lines $L_{23}, L_{25}, L_{61}, L_{36}$.
This configuration is nothing else but the configuration of points $(a,b,c,d)$ on the line $L_{14}$.
So we arrive at
the geometric interpretation of the triple ratio as the cross-ratio~\cite{G95}
given on the left of fig.~\ref{sa2l3}:
\begin{equation} \label{tripleratio1}
r(1\vert 2,4,3,25\cap 36) = r_3(1,2,3;4,5,6)= -\frac{\langle
  124\rangle\langle 235\rangle\langle 316\rangle}{\langle
  125\rangle\langle 236\rangle\langle 314\rangle}.
\end{equation}
Notice that
$r(2\vert 3,5,1,14\cap 36)= r(1\vert 2,4,3,25\cap 36)$; to see this, project onto the $36$ line.

\begin{figure}
\centering
\includegraphics{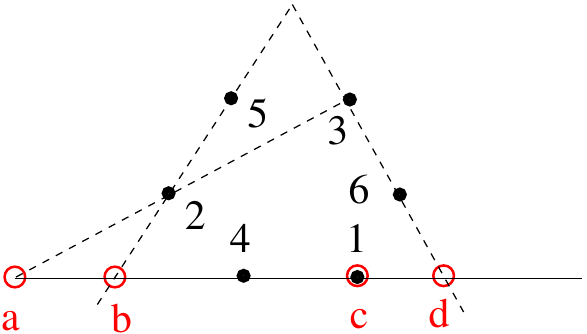}
\caption{$r(a,b,c,d) = r(1\vert 2,4,3,25\cap 36) = r_3(1,2,3;4,5,6)=
  -\frac{\langle 124\rangle\langle 235\rangle\langle
    316\rangle}{\langle 125\rangle\langle 236\rangle\langle 314\rangle}$.}
\label{sa2l3}
\end{figure}

2. To check formula~(\ref{sa12}) we write
\begin{equation}
r(f_{12}\vert f_{23}, f_{34}, f_{56}, f_{71}) = \frac{[f_{12}, f_{23}, f_{34}]
[f_{12}, f_{56}, f_{71}]}{[f_{12}, f_{34}, f_{56}][f_{12}, f_{23}, f_{71}]} \stackrel{(\ref{a}), (\ref{b})}{=}
\end{equation}
\begin{equation}
\frac{\langle 123\rangle\langle 234\rangle
\langle 712\rangle\langle 156\rangle}{[f_{12}, f_{34}, f_{56}]\langle
123\rangle\langle 271\rangle} =
\frac{\langle 156\rangle\langle 234\rangle}{\langle {1\times2}, {3\times4}, {5\times6}\rangle}.
\end{equation}
This provides the geometric interpretation of the cluster $\mathcal{X}$-coordinate~(\ref{sa12})  as the cross-ratio
given on the left of fig.~\ref{sa2l1}.

Therefore all of the more complicated cluster $\mathcal{X}$-coordinates for $\Conf_7(\PP^3)$ can be obtained by
applying the parity involution to the standard cross-ratios $r$ and the triple ratios $r_3$.

\begin{figure}
\centering
\includegraphics{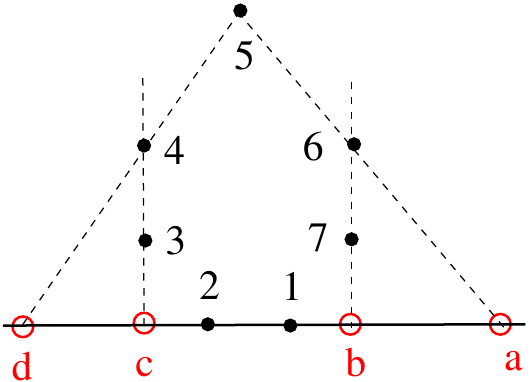}
\caption{$r(a,b,c,d) = \frac{(345)(124)(567)(126)}{(456)(125)(1\times2, 3\times 4, 6\times 7)}$.}
\label{sa2l2}
\end{figure}

\section{An Identity for the Trilogarithm in Cluster \texorpdfstring{$\mathcal{X}$}{X}-Coordinates}
\label{sec:trilog-ident}

As we emphasized in sec.~\ref{sec:math}, Abel's pentagon equation for the dilogarithm
admits a form where the set of arguments is the set of all (negated)
cluster $\mathcal{X}$-coordinates on the space of  configurations of five points
in $\PP^1$: this is the example of the $A_2$ algebra reviewed
in sec.~\ref{sec:intr-clust-algebr}, though the properties of the pentagram
of arguments were studied already by Gauss~\cite{Gauss}.

It is natural to ask whether there is a natural generalization of this
feature of Abel's identity to higher polylogarithms.
There is a generic functional equation for the trilogarithm related to the space of configurations
of seven points in $\PP^2$, see~\cite{G91a,G95}, from which any functional equation
for the trilogarithm can be deduced.
Its  arguments are the triple ratios, which, as we now know, are the simplest examples  of cluster
$\mathcal{X}$-coordinates which go beyond the $\PP^1$ cross-ratios. However,
the collection of all its arguments is invariant under the action of the permutation group
$S_7$, and so cannot be a subset of the set of cluster coordinates.
Instead, we have found the following

\begin{theorem}
\label{fe3l}
Given a configuration of six points in $\PP^2$, there is a 40-term functional
equation for the classical trilogarithm:
\begin{multline}\label{B1}
\{(1\vert 2,3,4,5)\}_3 + \{(1\vert 2,4,5,6)\}_3 + \{r_3(1,4,5; 2,3,6)\}_3 +\\
+\frac{1}{3}\{r_3(1,3,5;6,2,4)\}_3 + \textnormal{signed dihedral permutations} = 0
\end{multline} in $\B_3(\C)$.
Here the six cyclic permutations are taken with plus sign, and the six
anticyclic permutations are with minus sign.
\end{theorem}

The remarkable feature of the functional equation~(\ref{B1}) for the trilogarithm is that each of its $40$ arguments
is a cluster $\mathcal{X}$-coordinate on the space of configurations of six points in $\PP^2$.
This is the first known
functional equation for the trilogarithm with the property that all arguments
are cluster $\mathcal{X}$-coordinates of the same algebra.

\begin{figure}
\centering
\includegraphics{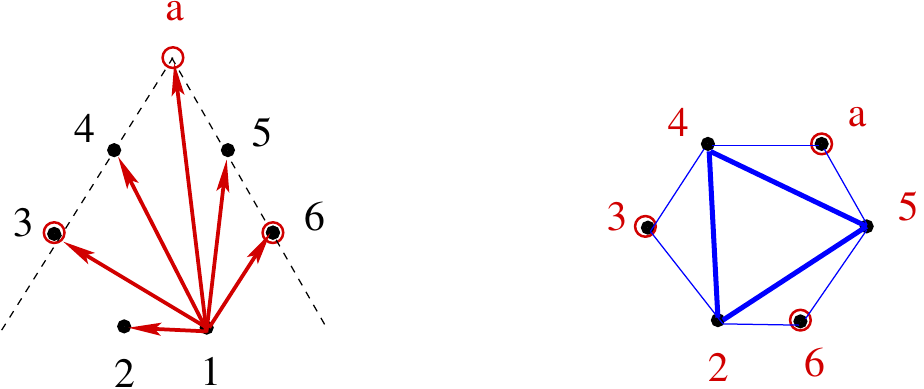}
\caption{Geometric meaning of the first three terms in~(\ref{B1}), given by cross-ratios~(\ref{threecr}).}
\label{sa2l5}
\end{figure}

Here is a geometric interpretation of the functional equation~(\ref{B1}).
Take a configuration of six points $(x_1, \dotsc, x_6)$ in $\PP^2$.
 Let $a:= {x_3x_4}\cap {x_5x_6}$ be the intersection point of the lines ${x_3x_4}$ and ${x_5x_6}$.
The lines passing through $x_1$ form a projective line.
The points $x_2, x_3, x_4, a, x_5, x_6$ determine a configuration of six point on that line, denoted by
$(x_1\vert x_2, x_3, x_4, a, x_5, x_6)$, and illustrated on the left in fig.~\ref{sa2l5}.
The three cross-ratios
\begin{equation}
\label{threecr}
r(x_1\vert x_2, x_3, x_4, x_5),  \qquad r(x_1\vert x_2, x_4, x_5, x_6),  \qquad r(x_1\vert x_4, a, x_5, x_2)
\end{equation}
are the arguments of the first three terms in~(\ref{B1}), up to the inversion of the third one,
which does not affect the corresponding element of the group $\B_3$. Indeed, the identity~(\ref{tripleratio1}) shows that
\begin{equation}
r(x_1\vert x_4, a, x_5, x_2) = r(x_1\vert x_4, x_2, x_5, a)^{-1}  = r_3(1,4,5; 2,3,6)^{-1}.
\end{equation}

 The three cross-ratios~(\ref{threecr}) are the cluster $\mathcal{X}$-coordinates
on the space of configurations $(x_1\vert x_2, x_3, x_4, a, x_5, x_6)$ of six points on $\PP^1$
corresponding to the triangulation of the hexagon in fig.~\ref{sa2l5}. They are
assigned   to the three diagonals of the triangulation;
 the points of the configuration $(x_1\vert x_2, x_3, x_4, a, x_5, x_6)$ are situated at the vertices of the hexagon.

\begin{proof}[Proof of Theorem~\ref{fe3l}]
We need to check that applying the cobracket $\delta$ to the
$40$-term expression~(\ref{B1}) we get zero in $\B_2(\C) \otimes \C^*$.
Due to the (signed) dihedral symmetry, it is sufficient to
calculate the expressions $\B_2^{\langle 123\rangle} \otimes \langle 123\rangle$,
$\B_2^{\langle 124\rangle}\otimes \langle 124\rangle$ and $\B_2^{\langle 135\rangle}\otimes
\langle 135\rangle$ in $\Z[\C]\otimes \C^*$. One has $\B_2^{\langle 135\rangle}=0$. The other two
we present as sums of  the five-term relations of geometric
origin. Before we formulate the answer, let us recall that
\begin{align}
\frac{\langle 123\rangle\langle 456\rangle}{\langle 1\times 2, 3\times 4, 5\times 6\rangle} &= r(1\vert 2,3,4, 34\cap 56), \\
\frac{\langle 1\times 2, 3\times 4, 5\times 6\rangle}{\langle
  156\rangle\langle234\rangle} &= r(5\vert 1, 2, 12\cap 34, 6).
\end{align}

Using this, one can calculate that
\begin{multline}
\B_2^{(123)}=
-\Bigl(\{r(5\vert 6,1,3,4)\}_2 + \{r_3(1,5,4; 2,6,3)\}_2+
       \{r_3(1,3,5; 2,4,6)\}_2 +\\ \{r(1\vert 5,2,3,4)\}_2 +
       \{r(1\vert 2,3,4, 34\cap 56)\}_2\Bigr)\\
-\Bigl(\{r(3\vert 2,4,6,1)\}_2 + \{r_3(2,3,6; 1,4,5)\}_2 +
       \{r_3(1,6,3; 2,5,4)\}_2 +\\ \{r(6\vert 5,1,2,3)\}_2 +
       \{r(5\vert 1, 2, 12\cap 34, 6)\}_2\Bigr)\\
+\Bigl(\{r(4\vert 5, 1,2,3)\}_2 + \{r_3(2,4,5; 1,3,6)\}_2 +
       \{r_3(1,5,4; 2,6,3)\}_2 +\\ \{r(5\vert 6, 1,2,4)\}_2 +
       \{r(5\vert 1, 2, 12\cap 34, 6)\}_2\Bigr)\\
-\Bigl(\{r(5\vert 2,3,4,6)\}_2 +\{r_3(2,4,5; 1,3,6)\}_2 +
       \{r_3(2,5,3; 1,6,4)\}_2 +\\ \{r(2\vert 1,3,4,5)\}_2 +
       \{r(2\vert 1, 34\cap 56, 4, 3)\}_2\Bigr) - \text{permutation by
         $(321654)$}.
\end{multline}

Similarly,
\begin{multline}
\B_2^{(124)}=
\Bigl(\{r(1\vert 2,4,5,6)\}_2 + \{r(2\vert 4,5,6,1)\}_2 +
      \{r(4\vert 5,6,1,2)\}_2 +\\+ \{r(5\vert 6,1,2,4)\}_2 +
      \{r(6\vert 1,2,4,5)\}_2\Bigr)\\
-\Bigl(\{r(5\vert 1, 6, 12\cap 34, 2)\}_2 + \{r(4\vert 1,2,3,6)\}_2 +
       \{r_3(2,6,4; 1,5,3)\}_2 +\\+ \{r_3(1,4,6; 2,3,5)\}_2 +
       \{r(6\vert 1,2,4,5)\}_2 \Bigr)\\
-\Bigl(\{r(5\vert 1, 2, 12\cap 34, 6)\}_2 + \{r(4\vert 2,3,5,1)\}_2 +
       \{r_3(2,4,5; 1,3,6)\}_2 +\\+ \{r_3(1,5,4; 2,6,3)\}_2 +
       \{r(5\vert 2,4,6,1)\}_2 \Bigr).
\end{multline}

To present each of the seven five-term summands as a five term relation, define
\begin{multline}
\partial_\mathrm{cyc}
(1,2,3,4,5):= \{r(2,3,4,5)\} + \{r(3,4,5,1)\} +\\+ \{r(4,5,1,2)\} + \{r(5,1,2,3)\} + \{r(1,2,3,4)\}.
\end{multline}  We make a similar definition for the case of five
intersecting lines.  If the lines are $(12)$, $(13)$, $(14)$, $(15)$,
$(16)$, then we define
\begin{multline}
  \partial_\mathrm{cyc} (1\vert 2,3,4,5,6) := \{r(1\vert 2,3,4,5)\} +
  \{r(1\vert 3,4,5,6)\} +\\+ \{r(1\vert 4,5,6,2)\} + \{r(1\vert
  5,6,2,3)\} + \{r(1\vert 6,2,3,4)\}.
\end{multline}

Then,
\begin{multline} \label{5tr1}
\B_2^{(123)} =
\partial_\mathrm{cyc}
\Bigl(-(1\vert 2,3,4,34\cap 56, 5) + (3\vert 12\cap 56,1,2,4,6)\\ - (4\vert 12\cap 56,1,2,3,5) + (2\vert 1,3,4,34\cap 56,5)\Bigr).
\end{multline}
\begin{equation}
\label{5tr2}
\B_2^{(124)}=
\partial_\mathrm{cyc}
\Bigl((1,2,4,5,6)- (4\vert 6, 12 \cap 56,1,2,3)+(4\vert 5, 12 \cap 56,1,2,3)\Bigr).
\end{equation}
Notice that the  configuration
 $(1,2,4,5,6)$ in $\PP^2$ by duality, or by drawing the unique conic through
these five points, determines the configuration of five points on $\PP^1$.

Thus not only the functional equation~(\ref{B1}) itself, but also
the way it vanishes in $\B_2 \otimes \C^*$ is of cluster origin:
the five-term relations~(\ref{5tr1})-(\ref{5tr2}) correspond to certain pentagon
faces
in the four-dimensional Stasheff polytope of type $D_4$ describing the
cluster $\mathcal{X}$-variety $\Conf_6(\PP^2)$.

To complete the proof of the theorem it is sufficient to check that a certain specialization
of the functional equation is zero.  We leave this as an exercise.
\end{proof}

Finally, one may look for functional equations for the trilogarithm
which can be expressed via cluster $\mathcal{X}$-coordinates on the space $\Conf_7(\PP^2)$,
which is the case of interest for scattering amplitudes.
The six-point identity in eq.~\eqref{B1} can obviously also be used for seven points in $\mathbb{P}^{2}$:
as written, it simply does not depend on the seventh point.
Other $40$-term identities can be obtained from the  eq.~\eqref{B1} by
applying parity conjugation and dihedral permutations of the points
$1,\dotsc,7$. A less obvious type of transformation is parity followed
by a transposition of two points.  This results in transformations
sending points in $\mathbb{P}^{2}$ to lines in $\mathbb{P}^{2}$.  An
example of such a transformation is
\begin{equation}
1 \to (34), \quad
2 \to (35), \quad
3 \to (56), \quad
4 \to (67), \quad
5 \to (17), \quad
6 \to (12).
\end{equation}
Remarkably, applying this map to the 40-term identity~\eqref{B1} we get a functional equation written via cluster
$\mathcal{X}$-coordinates on the space $\Conf_7(\PP^2)$.
After considering all possible
transformations of these types
we obtain a total of 35 different 40-term identities written via cluster
$\mathcal{X}$-coordinates on the space $\Conf_7(\PP^2)$,
consisting of five families of seven related to each other by cyclic
rotations of the seven points.  Only 22 of these 35 identities are linearly independent.  Of course all of them reduce to the identity in eq.~\eqref{B1} by some in general complicated change of variables.  Such changes of variable arise from an embedding of the $D_4$ cluster algebra (that is, $\Gr(3,6)$) into the $E_6$ (or $\Gr(3,7)$) cluster algebra.
It would be very interesting to search for new functional
identities of cluster type at higher weight. Of particular interest is, of course, the next case---weight 4.
We have found that there
are no identities at weight 4 involving
cluster $\mathcal{X}$-coordinates on $\Conf_8(\PP^2)$, but we do
expect such identities on $\Conf_8(\PP^3)$.

\bibliographystyle{JHEP}
\bibliography{motives47}

\providecommand{\href}[2]{#2}\begingroup\raggedright\begin{thebibliography}{10}

\bibitem{Mangano:1990by}
M.~L. Mangano and S.~J. Parke, {\it {Multiparton amplitudes in gauge
  theories}},  {\em Phys.Rept.} {\bf 200} (1991) 301--367,
  [\href{http://xxx.lanl.gov/abs/hep-th/0509223}{{\tt hep-th/0509223}}].

\bibitem{Dixon:1996wi}
L.~J. Dixon, {\it {Calculating scattering amplitudes efficiently}},
  \href{http://xxx.lanl.gov/abs/hep-ph/9601359}{{\tt hep-ph/9601359}}.

\bibitem{Cachazo:2005ga}
F.~Cachazo and P.~Svrcek, {\it {Lectures on twistor strings and perturbative
  Yang-Mills theory}},  {\em PoS} {\bf RTN2005} (2005) 004,
  [\href{http://xxx.lanl.gov/abs/hep-th/0504194}{{\tt hep-th/0504194}}].

\bibitem{Bern:2007dw}
Z.~Bern, L.~J. Dixon, and D.~A. Kosower, {\it {On-Shell Methods in Perturbative
  QCD}},  {\em Annals Phys.} {\bf 322} (2007) 1587--1634,
  [\href{http://xxx.lanl.gov/abs/0704.2798}{{\tt arXiv:0704.2798}}].

\bibitem{GreyBook}
R.~Roiban, M.~Spradlin, and A.~Volovich, {\it Scattering amplitudes in gauge
  theories: Progress and outlook},  {\em J.Phys.} {\bf A44} (2011) 1.

\bibitem{Feng:2011np}
B.~Feng and M.~Luo, {\it {An Introduction to On-shell Recursion Relations}},
  \href{http://xxx.lanl.gov/abs/1111.5759}{{\tt arXiv:1111.5759}}.

\bibitem{Brink:1976bc}
L.~Brink, J.~H. Schwarz, and J.~Scherk, {\it {Supersymmetric Yang-Mills
  Theories}},  {\em Nucl.Phys.} {\bf B121} (1977) 77.

\bibitem{Gliozzi:1976qd}
F.~Gliozzi, J.~Scherk, and D.~I. Olive, {\it {Supersymmetry, Supergravity
  Theories and the Dual Spinor Model}},  {\em Nucl.Phys.} {\bf B122} (1977)
  253--290.

\bibitem{G02}
A.~Goncharov, {\it {Galois symmetries of fundamental groupoids and
  noncommutative geometry}},  {\em Duke Math. J.} {\bf 128} (2005), no.~2
  209--284, [\href{http://xxx.lanl.gov/abs/math/0208144}{{\tt math/0208144}}].

\bibitem{Schlotterer:2012ny}
O.~Schlotterer and S.~Stieberger, {\it {Motivic Multiple Zeta Values and
  Superstring Amplitudes}},  \href{http://xxx.lanl.gov/abs/1205.1516}{{\tt
  arXiv:1205.1516}}.

\bibitem{Drummond:2013vz}
J.~Drummond and E.~Ragoucy, {\it {Superstring amplitudes and the associator}},
  \href{http://xxx.lanl.gov/abs/1301.0794}{{\tt arXiv:1301.0794}}.

\bibitem{Broedel:2013tta}
J.~Broedel, O.~Schlotterer, and S.~Stieberger, {\it {Polylogarithms, Multiple
  Zeta Values and Superstring Amplitudes}},
  \href{http://xxx.lanl.gov/abs/1304.7267}{{\tt arXiv:1304.7267}}.

\bibitem{Broedel:2013aza}
J.~Broedel, O.~Schlotterer, S.~Stieberger, and T.~Terasoma, {\it {All order
  alpha'-expansion of superstring trees from the Drinfeld associator}},
  \href{http://xxx.lanl.gov/abs/1304.7304}{{\tt arXiv:1304.7304}}.

\bibitem{Drummond:2006rz}
J.~Drummond, J.~Henn, V.~Smirnov, and E.~Sokatchev, {\it {Magic identities for
  conformal four-point integrals}},  {\em JHEP} {\bf 0701} (2007) 064,
  [\href{http://xxx.lanl.gov/abs/hep-th/0607160}{{\tt hep-th/0607160}}].

\bibitem{Bern:2006ew}
Z.~Bern, M.~Czakon, L.~J. Dixon, D.~A. Kosower, and V.~A. Smirnov, {\it {The
  Four-Loop Planar Amplitude and Cusp Anomalous Dimension in Maximally
  Supersymmetric Yang-Mills Theory}},  {\em Phys.Rev.} {\bf D75} (2007) 085010,
  [\href{http://xxx.lanl.gov/abs/hep-th/0610248}{{\tt hep-th/0610248}}].

\bibitem{Alday:2007hr}
L.~F. Alday and J.~M. Maldacena, {\it {Gluon scattering amplitudes at strong
  coupling}},  {\em JHEP} {\bf 0706} (2007) 064,
  [\href{http://xxx.lanl.gov/abs/0705.0303}{{\tt arXiv:0705.0303}}].

\bibitem{Drummond:2007aua}
G.~Korchemsky, J.~Drummond, and E.~Sokatchev, {\it {Conformal properties of
  four-gluon planar amplitudes and Wilson loops}},  {\em Nucl.Phys.} {\bf B795}
  (2008) 385--408, [\href{http://xxx.lanl.gov/abs/0707.0243}{{\tt
  arXiv:0707.0243}}].

\bibitem{Drummond:2007cf}
J.~Drummond, J.~Henn, G.~Korchemsky, and E.~Sokatchev, {\it {On planar gluon
  amplitudes/Wilson loops duality}},  {\em Nucl.Phys.} {\bf B795} (2008)
  52--68, [\href{http://xxx.lanl.gov/abs/0709.2368}{{\tt arXiv:0709.2368}}].

\bibitem{Alday:2007he}
L.~F. Alday and J.~Maldacena, {\it {Comments on gluon scattering amplitudes via
  AdS/CFT}},  {\em JHEP} {\bf 0711} (2007) 068,
  [\href{http://xxx.lanl.gov/abs/0710.1060}{{\tt arXiv:0710.1060}}].

\bibitem{Drummond:2007au}
J.~Drummond, J.~Henn, G.~Korchemsky, and E.~Sokatchev, {\it {Conformal Ward
  identities for Wilson loops and a test of the duality with gluon
  amplitudes}},  {\em Nucl.Phys.} {\bf B826} (2010) 337--364,
  [\href{http://xxx.lanl.gov/abs/0712.1223}{{\tt arXiv:0712.1223}}].

\bibitem{Drummond:2008vq}
J.~Drummond, J.~Henn, G.~Korchemsky, and E.~Sokatchev, {\it {Dual
  superconformal symmetry of scattering amplitudes in N=4 super-Yang-Mills
  theory}},  {\em Nucl.Phys.} {\bf B828} (2010) 317--374,
  [\href{http://xxx.lanl.gov/abs/0807.1095}{{\tt arXiv:0807.1095}}].

\bibitem{Goncharov:2010jf}
A.~B. Goncharov, M.~Spradlin, C.~Vergu, and A.~Volovich, {\it {Classical
  Polylogarithms for Amplitudes and Wilson Loops}},  {\em Phys.Rev.Lett.} {\bf
  105} (2010) 151605, [\href{http://xxx.lanl.gov/abs/1006.5703}{{\tt
  arXiv:1006.5703}}].

\bibitem{G91b}
A.~Goncharov, {\it {Geometry of configurations, polylogarithms, and motivic
  cohomology}},  {\em Adv. Math.} {\bf 114} (1995), no.~2 197--318.

\bibitem{FG03b}
V.~V. Fock and A.~B. Goncharov, {\it {Cluster ensembles, quantization and the
  dilogarithm}},  {\em Ann. Sci. \'Ec. Norm. Sup\'er. (4)} {\bf 42} (2009),
  no.~6 865--930, [\href{http://xxx.lanl.gov/abs/math/0311245}{{\tt
  math/0311245}}].

\bibitem{1021.16017}
S.~Fomin and A.~Zelevinsky, {\it {Cluster algebras. I: Foundations}},  {\em J.
  Am. Math. Soc.} {\bf 15} (2002), no.~2 497--529.

\bibitem{1054.17024}
S.~Fomin and A.~Zelevinsky, {\it {Cluster algebras. II: Finite type
  classification}},  {\em Invent. Math.} {\bf 154} (2003), no.~1 63--121.

\bibitem{ArkaniHamed:2009dn}
N.~Arkani-Hamed, F.~Cachazo, C.~Cheung, and J.~Kaplan, {\it {A Duality For The
  S Matrix}},  {\em JHEP} {\bf 1003} (2010) 020,
  [\href{http://xxx.lanl.gov/abs/0907.5418}{{\tt arXiv:0907.5418}}].

\bibitem{ArkaniHamed:2009vw}
N.~Arkani-Hamed, F.~Cachazo, and C.~Cheung, {\it {The Grassmannian Origin Of
  Dual Superconformal Invariance}},  {\em JHEP} {\bf 1003} (2010) 036,
  [\href{http://xxx.lanl.gov/abs/0909.0483}{{\tt arXiv:0909.0483}}].

\bibitem{ArkaniHamed:2009sx}
N.~Arkani-Hamed, J.~Bourjaily, F.~Cachazo, and J.~Trnka, {\it {Local Spacetime
  Physics from the Grassmannian}},  {\em JHEP} {\bf 1101} (2011) 108,
  [\href{http://xxx.lanl.gov/abs/0912.3249}{{\tt arXiv:0912.3249}}].

\bibitem{ArkaniHamed:2009dg}
N.~Arkani-Hamed, J.~Bourjaily, F.~Cachazo, and J.~Trnka, {\it {Unification of
  Residues and Grassmannian Dualities}},  {\em JHEP} {\bf 1101} (2011) 049,
  [\href{http://xxx.lanl.gov/abs/0912.4912}{{\tt arXiv:0912.4912}}].

\bibitem{ArkaniHamed:2010kv}
N.~Arkani-Hamed, J.~L. Bourjaily, F.~Cachazo, S.~Caron-Huot, and J.~Trnka, {\it
  {The All-Loop Integrand For Scattering Amplitudes in Planar N=4 SYM}},  {\em
  JHEP} {\bf 1101} (2011) 041, [\href{http://xxx.lanl.gov/abs/1008.2958}{{\tt
  arXiv:1008.2958}}].

\bibitem{ArkaniHamed:2012nw}
N.~Arkani-Hamed, J.~L. Bourjaily, F.~Cachazo, A.~B. Goncharov, A.~Postnikov,
  and J.~Trnka, {\it {Scattering Amplitudes and the Positive Grassmannian}},
  \href{http://xxx.lanl.gov/abs/1212.5605}{{\tt arXiv:1212.5605}}.

\bibitem{1057.52003}
S.~Fomin and A.~Zelevinsky, {\it {$Y$-systems and generalized associahedra}},
  {\em Ann. Math. (2)} {\bf 158} (2003), no.~3 977--1018.

\bibitem{0114.39402}
J.~D. Stasheff, {\it {Homotopy associativity of $H$-spaces. I, II}},  {\em
  Trans. Am. Math. Soc. 108,} {\bf 275-292} (1963) 293--312.

\bibitem{Hodges:2009hk}
A.~Hodges, {\it {Eliminating spurious poles from gauge-theoretic amplitudes}},
  \href{http://xxx.lanl.gov/abs/0905.1473}{{\tt arXiv:0905.1473}}.

\bibitem{Penrose:1967wn}
R.~Penrose, {\it {Twistor algebra}},  {\em J.Math.Phys.} {\bf 8} (1967) 345.

\bibitem{Penrose:1972ia}
R.~Penrose and M.~A. MacCallum, {\it {Twistor theory: An Approach to the
  quantization of fields and space-time}},  {\em Phys.Rept.} {\bf 6} (1972)
  241--316.

\bibitem{Witten:2003nn}
E.~Witten, {\it {Perturbative gauge theory as a string theory in twistor
  space}},  {\em Commun.Math.Phys.} {\bf 252} (2004) 189--258,
  [\href{http://xxx.lanl.gov/abs/hep-th/0312171}{{\tt hep-th/0312171}}].

\bibitem{Mason:2009qx}
L.~Mason and D.~Skinner, {\it {Dual Superconformal Invariance, Momentum
  Twistors and Grassmannians}},  {\em JHEP} {\bf 0911} (2009) 045,
  [\href{http://xxx.lanl.gov/abs/0909.0250}{{\tt arXiv:0909.0250}}].

\bibitem{Witten:1978xx}
E.~Witten, {\it {An Interpretation of Classical Yang-Mills Theory}},  {\em
  Phys.Lett.} {\bf B77} (1978) 394.

\bibitem{Beisert:2012gb}
N.~Beisert and C.~Vergu, {\it {On the Geometry of Null Polygons in Full N=4
  Superspace}},  {\em Phys.Rev.} {\bf D86} (2012) 026006,
  [\href{http://xxx.lanl.gov/abs/1203.0525}{{\tt arXiv:1203.0525}}].

\bibitem{Beisert:2012xx}
N.~Beisert, S.~He, B.~U. Schwab, and C.~Vergu, {\it {Null Polygonal Wilson
  Loops in Full N=4 Superspace}},  {\em J.Phys.} {\bf A45} (2012) 265402,
  [\href{http://xxx.lanl.gov/abs/1203.1443}{{\tt arXiv:1203.1443}}].

\bibitem{Bern:2008ap}
Z.~Bern, L.~Dixon, D.~Kosower, R.~Roiban, M.~Spradlin, C.~Vergu, and
  A.~Volovich, {\it {The Two-Loop Six-Gluon MHV Amplitude in Maximally
  Supersymmetric Yang-Mills Theory}},  {\em Phys.Rev.} {\bf D78} (2008) 045007,
  [\href{http://xxx.lanl.gov/abs/0803.1465}{{\tt arXiv:0803.1465}}].

\bibitem{Drummond:2008aq}
J.~Drummond, J.~Henn, G.~Korchemsky, and E.~Sokatchev, {\it {Hexagon Wilson
  loop = six-gluon MHV amplitude}},  {\em Nucl.Phys.} {\bf B815} (2009)
  142--173, [\href{http://xxx.lanl.gov/abs/0803.1466}{{\tt arXiv:0803.1466}}].

\bibitem{DelDuca:2009au}
V.~Del~Duca, C.~Duhr, and V.~A. Smirnov, {\it {An Analytic Result for the
  Two-Loop Hexagon Wilson Loop in N = 4 SYM}},  {\em JHEP} {\bf 1003} (2010)
  099, [\href{http://xxx.lanl.gov/abs/0911.5332}{{\tt arXiv:0911.5332}}].

\bibitem{DelDuca:2010zg}
V.~Del~Duca, C.~Duhr, and V.~A. Smirnov, {\it {The Two-Loop Hexagon Wilson Loop
  in N = 4 SYM}},  {\em JHEP} {\bf 1005} (2010) 084,
  [\href{http://xxx.lanl.gov/abs/1003.1702}{{\tt arXiv:1003.1702}}].

\bibitem{B}
A.~Beilinson, {\it {Height pairing between algebraic cycles}},  in {\em
  K-Theory, Arithmetic and Geometry}.
\newblock {Berlin: Springer-Verlag}, 1987.

\bibitem{DG}
P.~Deligne and A.~B. Goncharov, {\it Groupes fondamentaux motiviques de {T}ate
  mixte},  {\em Ann. Sci. \'Ecole Norm. Sup. (4)} {\bf 38} (2005), no.~1 1--56.

\bibitem{Bl}
S.~J. Bloch, {\em {Higher regulators, algebraic $K$-theory, and zeta functions
  of elliptic curves}}.
\newblock {Providence, RI: American Mathematical Society (AMS)}, 2000.

\bibitem{Su}
A.~Suslin, {\it {$K\sb 3$ of a field and the Bloch group}},  {\em Proc. Steklov
  Inst. Math.} {\bf 183} (1990) 217--239.

\bibitem{Z}
D.~Zagier, {\it {Polylogarithms, Dedekind zeta functions, and the algebraic
  K-theory of fields}},  in {\em Arithmetic Algebraic Geometry}.
\newblock {Boston: Birkh\"auser}, 1991.

\bibitem{G91a}
A.~B. Goncharov, {\it Polylogarithms and motivic {G}alois groups},  in {\em
  Motives ({S}eattle, {WA}, 1991)}, vol.~55 of {\em Proc. Sympos. Pure Math.},
  pp.~43--96.
\newblock Amer. Math. Soc., Providence, RI, 1994.

\bibitem{G95}
A.~Goncharov, {\it {Deninger's conjecture on $L$-functions of elliptic curves
  at $s=3$}},  {\em J. Math. Sci., New York} {\bf 81} (1996), no.~3 2631--2656,
  [\href{http://xxx.lanl.gov/abs/alg-geom/9512016}{{\tt alg-geom/9512016}}].

\bibitem{CaronHuot:2011ky}
S.~Caron-Huot, {\it {Superconformal symmetry and two-loop amplitudes in planar
  N=4 super Yang-Mills}},  {\em JHEP} {\bf 1112} (2011) 066,
  [\href{http://xxx.lanl.gov/abs/1105.5606}{{\tt arXiv:1105.5606}}].

\bibitem{MR1888840}
S.~Fomin and A.~Zelevinsky, {\it The {L}aurent phenomenon},  {\em Adv. in Appl.
  Math.} {\bf 28} (2002), no.~2 119--144.

\bibitem{1088.22009}
J.~S. Scott, {\it {Grassmannians and cluster algebras}},  {\em Proc. Lond.
  Math. Soc., III. Ser.} {\bf 92} (2006), no.~2 345--380.

\bibitem{1057.53064}
M.~Gekhtman, M.~Shapiro, and A.~Vainshtein, {\it {Cluster algebras and Poisson
  geometry}},  {\em Mosc. Math. J.} {\bf 3} (2003), no.~3 899--934.

\bibitem{1215.16012}
B.~Keller, {\it {Cluster algebras, quiver representations and triangulated
  categories}},  in {\em Triangulated Categories}.
\newblock {Cambridge: Cambridge University Press}, 2010.

\bibitem{MR2383126}
S.~Fomin and N.~Reading, {\it Root systems and generalized associahedra},  in
  {\em Geometric combinatorics}, vol.~13 of {\em IAS/Park City Math. Ser.},
  pp.~63--131.
\newblock Amer. Math. Soc., Providence, RI, 2007.

\bibitem{Anastasiou:2009kna}
C.~Anastasiou, A.~Brandhuber, P.~Heslop, V.~V. Khoze, B.~Spence, and
  G.~Travaglini, {\it {Two-Loop Polygon Wilson Loops in N=4 SYM}},  {\em JHEP}
  {\bf 0905} (2009) 115, [\href{http://xxx.lanl.gov/abs/0902.2245}{{\tt
  arXiv:0902.2245}}].

\bibitem{Dixon:2011nj}
L.~J. Dixon, J.~M. Drummond, and J.~M. Henn, {\it {Analytic result for the
  two-loop six-point NMHV amplitude in N=4 super Yang-Mills theory}},  {\em
  JHEP} {\bf 1201} (2012) 024, [\href{http://xxx.lanl.gov/abs/1111.1704}{{\tt
  arXiv:1111.1704}}].

\bibitem{Dixon:2011pw}
L.~J. Dixon, J.~M. Drummond, and J.~M. Henn, {\it {Bootstrapping the three-loop
  hexagon}},  {\em JHEP} {\bf 1111} (2011) 023,
  [\href{http://xxx.lanl.gov/abs/1108.4461}{{\tt arXiv:1108.4461}}].

\bibitem{CaronHuot:2011kk}
S.~Caron-Huot and S.~He, {\it {Jumpstarting the All-Loop S-Matrix of Planar N=4
  Super Yang-Mills}},  {\em JHEP} {\bf 1207} (2012) 174,
  [\href{http://xxx.lanl.gov/abs/1112.1060}{{\tt arXiv:1112.1060}}].

\bibitem{DelDuca:2010zp}
V.~Del~Duca, C.~Duhr, and V.~A. Smirnov, {\it {A Two-Loop Octagon Wilson Loop
  in N = 4 SYM}},  {\em JHEP} {\bf 1009} (2010) 015,
  [\href{http://xxx.lanl.gov/abs/1006.4127}{{\tt arXiv:1006.4127}}].

\bibitem{Heslop:2010kq}
P.~Heslop and V.~V. Khoze, {\it {Analytic Results for MHV Wilson Loops}},  {\em
  JHEP} {\bf 1011} (2010) 035, [\href{http://xxx.lanl.gov/abs/1007.1805}{{\tt
  arXiv:1007.1805}}].

\bibitem{Heslop:2011hv}
P.~Heslop and V.~V. Khoze, {\it {Wilson Loops @ 3-Loops in Special
  Kinematics}},  {\em JHEP} {\bf 1111} (2011) 152,
  [\href{http://xxx.lanl.gov/abs/1109.0058}{{\tt arXiv:1109.0058}}].

\bibitem{Goddard:2012cx}
T.~Goddard, P.~Heslop, and V.~V. Khoze, {\it {Uplifting Amplitudes in Special
  Kinematics}},  {\em JHEP} {\bf 1210} (2012) 041,
  [\href{http://xxx.lanl.gov/abs/1205.3448}{{\tt arXiv:1205.3448}}].

\bibitem{Ferro:2012wa}
L.~Ferro, {\it {Differential equations for multi-loop integrals and
  two-dimensional kinematics}},  {\em JHEP} {\bf 1304} (2013) 160,
  [\href{http://xxx.lanl.gov/abs/1204.1031}{{\tt arXiv:1204.1031}}].

\bibitem{Bartels:2011xy}
J.~Bartels, L.~Lipatov, and A.~Prygarin, {\it {Collinear and Regge behavior of
  2 $\to$ 4 MHV amplitude in N = 4 super Yang-Mills theory}},
  \href{http://xxx.lanl.gov/abs/1104.4709}{{\tt arXiv:1104.4709}}.

\bibitem{Prygarin:2011gd}
A.~Prygarin, M.~Spradlin, C.~Vergu, and A.~Volovich, {\it {All Two-Loop MHV
  Amplitudes in Multi-Regge Kinematics From Applied Symbology}},  {\em
  Phys.Rev.} {\bf D85} (2012) 085019,
  [\href{http://xxx.lanl.gov/abs/1112.6365}{{\tt arXiv:1112.6365}}].

\bibitem{Bartels:2011ge}
J.~Bartels, A.~Kormilitzin, L.~Lipatov, and A.~Prygarin, {\it {BFKL approach
  and $2 \to 5$ maximally helicity violating amplitude in ${\cal N}=4$
  super-Yang-Mills theory}},  {\em Phys.Rev.} {\bf D86} (2012) 065026,
  [\href{http://xxx.lanl.gov/abs/1112.6366}{{\tt arXiv:1112.6366}}].

\bibitem{Lipatov:2012gk}
L.~Lipatov, A.~Prygarin, and H.~J. Schnitzer, {\it {The Multi-Regge limit of
  NMHV Amplitudes in N=4 SYM Theory}},  {\em JHEP} {\bf 1301} (2013) 068,
  [\href{http://xxx.lanl.gov/abs/1205.0186}{{\tt arXiv:1205.0186}}].

\bibitem{Dixon:2012yy}
L.~J. Dixon, C.~Duhr, and J.~Pennington, {\it {Single-valued harmonic
  polylogarithms and the multi-Regge limit}},  {\em JHEP} {\bf 1210} (2012)
  074, [\href{http://xxx.lanl.gov/abs/1207.0186}{{\tt arXiv:1207.0186}}].

\bibitem{Pennington:2012zj}
J.~Pennington, {\it {The six-point remainder function to all loop orders in the
  multi-Regge limit}},  {\em JHEP} {\bf 1301} (2013) 059,
  [\href{http://xxx.lanl.gov/abs/1209.5357}{{\tt arXiv:1209.5357}}].

\bibitem{Alday:2010ku}
L.~F. Alday, D.~Gaiotto, J.~Maldacena, A.~Sever, and P.~Vieira, {\it {An
  Operator Product Expansion for Polygonal null Wilson Loops}},  {\em JHEP}
  {\bf 1104} (2011) 088, [\href{http://xxx.lanl.gov/abs/1006.2788}{{\tt
  arXiv:1006.2788}}].

\bibitem{Gaiotto:2010fk}
D.~Gaiotto, J.~Maldacena, A.~Sever, and P.~Vieira, {\it {Bootstrapping Null
  Polygon Wilson Loops}},  {\em JHEP} {\bf 1103} (2011) 092,
  [\href{http://xxx.lanl.gov/abs/1010.5009}{{\tt arXiv:1010.5009}}].

\bibitem{Gaiotto:2011dt}
D.~Gaiotto, J.~Maldacena, A.~Sever, and P.~Vieira, {\it {Pulling the straps of
  polygons}},  {\em JHEP} {\bf 1112} (2011) 011,
  [\href{http://xxx.lanl.gov/abs/1102.0062}{{\tt arXiv:1102.0062}}].

\bibitem{Basso:2013vsa}
B.~Basso, A.~Sever, and P.~Vieira, {\it {Space-time S-matrix and Flux-tube
  S-matrix at Finite Coupling}},  \href{http://xxx.lanl.gov/abs/1303.1396}{{\tt
  arXiv:1303.1396}}.

\bibitem{Beisert:2010jr}
N.~Beisert, C.~Ahn, L.~F. Alday, Z.~Bajnok, J.~M. Drummond, et~al., {\it
  {Review of AdS/CFT Integrability: An Overview}},  {\em Lett.Math.Phys.} {\bf
  99} (2012) 3--32, [\href{http://xxx.lanl.gov/abs/1012.3982}{{\tt
  arXiv:1012.3982}}].

\bibitem{Drummond:2009fd}
J.~M. Drummond, J.~M. Henn, and J.~Plefka, {\it {Yangian symmetry of scattering
  amplitudes in N=4 super Yang-Mills theory}},  {\em JHEP} {\bf 0905} (2009)
  046, [\href{http://xxx.lanl.gov/abs/0902.2987}{{\tt arXiv:0902.2987}}].

\bibitem{Alday:2009dv}
L.~F. Alday, D.~Gaiotto, and J.~Maldacena, {\it {Thermodynamic Bubble Ansatz}},
   {\em JHEP} {\bf 1109} (2011) 032,
  [\href{http://xxx.lanl.gov/abs/0911.4708}{{\tt arXiv:0911.4708}}].

\bibitem{Alday:2010vh}
L.~F. Alday, J.~Maldacena, A.~Sever, and P.~Vieira, {\it {Y-system for
  Scattering Amplitudes}},  {\em J.Phys.} {\bf A43} (2010) 485401,
  [\href{http://xxx.lanl.gov/abs/1002.2459}{{\tt arXiv:1002.2459}}].

\bibitem{Ferro:2012xw}
L.~Ferro, T.~Lukowski, C.~Meneghelli, J.~Plefka, and M.~Staudacher, {\it
  {Harmonic R-matrices for Scattering Amplitudes and Spectral Regularization}},
   \href{http://xxx.lanl.gov/abs/1212.0850}{{\tt arXiv:1212.0850}}.

\bibitem{DelDuca:2011wh}
V.~Del~Duca, L.~J. Dixon, J.~M. Drummond, C.~Duhr, J.~M. Henn, et~al., {\it
  {The one-loop six-dimensional hexagon integral with three massive corners}},
  {\em Phys.Rev.} {\bf D84} (2011) 045017,
  [\href{http://xxx.lanl.gov/abs/1105.2011}{{\tt arXiv:1105.2011}}].

\bibitem{Duhr:2011zq}
C.~Duhr, H.~Gangl, and J.~R. Rhodes, {\it {From polygons and symbols to
  polylogarithmic functions}},  {\em JHEP} {\bf 1210} (2012) 075,
  [\href{http://xxx.lanl.gov/abs/1110.0458}{{\tt arXiv:1110.0458}}].

\bibitem{Bullimore:2011kg}
M.~Bullimore and D.~Skinner, {\it {Descent Equations for Superamplitudes}},
  \href{http://xxx.lanl.gov/abs/1112.1056}{{\tt arXiv:1112.1056}}.

\bibitem{Brandhuber:2012vm}
A.~Brandhuber, G.~Travaglini, and G.~Yang, {\it {Analytic two-loop form factors
  in N=4 SYM}},  {\em JHEP} {\bf 1205} (2012) 082,
  [\href{http://xxx.lanl.gov/abs/1201.4170}{{\tt arXiv:1201.4170}}].

\bibitem{Bogner:2012dn}
C.~Bogner and F.~Brown, {\it {Symbolic integration and multiple
  polylogarithms}},  {\em PoS} {\bf LL2012} (2012) 053,
  [\href{http://xxx.lanl.gov/abs/1209.6524}{{\tt arXiv:1209.6524}}].

\bibitem{Lipstein:2012vs}
A.~E. Lipstein and L.~Mason, {\it {From the holomorphic Wilson loop to `d log'
  loop-integrands for super-Yang-Mills amplitudes}},  {\em JHEP} {\bf 1305}
  (2013) 106, [\href{http://xxx.lanl.gov/abs/1212.6228}{{\tt
  arXiv:1212.6228}}].

\bibitem{Naculich:2013xa}
S.~G. Naculich, H.~Nastase, and H.~J. Schnitzer, {\it {All-loop
  infrared-divergent behavior of most-subleading-color gauge-theory
  amplitudes}},  {\em JHEP} {\bf 1304} (2013) 114,
  [\href{http://xxx.lanl.gov/abs/1301.2234}{{\tt arXiv:1301.2234}}].

\bibitem{Drummond:2013nda}
J.~Drummond, C.~Duhr, B.~Eden, P.~Heslop, J.~Pennington, et~al., {\it {Leading
  singularities and off-shell conformal integrals}},
  \href{http://xxx.lanl.gov/abs/1303.6909}{{\tt arXiv:1303.6909}}.

\bibitem{Duhr:2012fh}
C.~Duhr, {\it {Hopf algebras, coproducts and symbols: an application to Higgs
  boson amplitudes}},  {\em JHEP} {\bf 1208} (2012) 043,
  [\href{http://xxx.lanl.gov/abs/1203.0454}{{\tt arXiv:1203.0454}}].

\bibitem{vonManteuffel:2012je}
A.~von Manteuffel and C.~Studerus, {\it {Top quark pairs at two loops and
  Reduze 2}},  {\em PoS} {\bf LL2012} (2012) 059,
  [\href{http://xxx.lanl.gov/abs/1210.1436}{{\tt arXiv:1210.1436}}].

\bibitem{Gehrmann:2013vga}
T.~Gehrmann, L.~Tancredi, and E.~Weihs, {\it {Two-loop QCD helicity amplitudes
  for $g\,g \to Z\,g$ and $g\,g \to Z\,\gamma $}},  {\em JHEP} {\bf 1304}
  (2013) 101, [\href{http://xxx.lanl.gov/abs/1302.2630}{{\tt
  arXiv:1302.2630}}].

\bibitem{Anastasiou:2013srw}
C.~Anastasiou, C.~Duhr, F.~Dulat, and B.~Mistlberger, {\it {Soft triple-real
  radiation for Higgs production at N3LO}},  {\em JHEP} {\bf 1307} (2013) 003,
  [\href{http://xxx.lanl.gov/abs/1302.4379}{{\tt arXiv:1302.4379}}].

\bibitem{Henn:2013pwa}
J.~M. Henn, {\it {Multiloop integrals in dimensional regularization made
  simple}},  \href{http://xxx.lanl.gov/abs/1304.1806}{{\tt arXiv:1304.1806}}.

\bibitem{Gauss}
C.~F. Gauss, {\it Pentagramma mirificum},  in {\em Werke, Band III},
  pp.~481--490.
\newblock G\"ottingen, 1863.

\end{thebibliography}\endgroup

\end{document}